\documentclass[12pt]{article}

\usepackage[colorlinks, linkcolor=blue, citecolor=blue]{hyperref}           
\usepackage{color}
\usepackage{graphicx,subfigure,amsmath,amssymb,amsfonts,bm,epsfig,epsf,url,dsfont}
\usepackage{amsthm,mathrsfs}
\usepackage{bigints}

\newtheorem{thm}{Theorem}[section]
\newtheorem{lem}[thm]{Lemma}
\newtheorem{assum}[thm]{Assumption}
\newtheorem{proposition}[thm]{Proposition}
\newtheorem{remark}[thm]{Remark}

\newtheorem{definition}[thm]{Definition}
\newtheorem{problem}[thm]{Problem}
\newtheorem{lemma}[thm]{Lemma}
\newtheorem{corollary}[thm]{Corollary}

\usepackage{tikz}
\usepackage{bbm}
\usepackage[small,bf]{caption}
\usepackage[top=1in,bottom=1in,left=1in,right=1in]{geometry}
\usepackage{fancybox}
\usepackage{hyperref}
\usepackage{algorithm}
\usepackage{algpseudocode}
\usepackage{verbatim}
\usepackage[titletoc,toc]{appendix}
\usepackage{amsmath}
\usepackage{array}
\usepackage{tabularx}
\usepackage{booktabs}

\usepackage{mathtools}

\usepackage{scalerel,stackengine}
\stackMath
\newcommand\reallywidehat[1]{%
\savestack{\tmpbox}{\stretchto{%
  \scaleto{%
    \scalerel*[\widthof{\ensuremath{#1}}]{\kern-.6pt\bigwedge\kern-.6pt}%
    {\rule[-\textheight/2]{1ex}{\textheight}}%WIDTH-LIMITED BIG WEDGE
  }{\textheight}% 
}{0.5ex}}%
\stackon[1pt]{#1}{\tmpbox}%
}

\newcommand*{\rom}[1]{\expandafter\@slowromancap\romannumeral #1@}
\newcommand{\abs}[1]{\left|#1\right|}
\newcommand{\R}{\mathbb{R}}

\newcommand{\bmeta}{\bm{\eta}}

\newcommand{\SO}{\mathsf{SO}}
\newcommand{\SE}{\mathsf{SE}}
\DeclareMathOperator*{\argmax}{arg\,max}
\DeclareMathOperator*{\argmin}{arg\,min}

\newcommand{\p}[1]{\left(#1\right)}
\newcommand{\pp}[1]{\left[#1\right]}
\newcommand{\ppp}[1]{\left\{#1\right\}}
\newcommand{\norm}[1]{\left\|#1\right\|}

\newcommand\equalityAS{\mathrel{\stackrel{\makebox[0pt]{\mbox{\normalfont \scriptsize a.s.}}}{=}}}

\newcommand{\s}[1]{\mathsf{#1}}
\usepackage{xspace}

\numberwithin{equation}{section}

\allowdisplaybreaks
\usepackage{authblk}

\begin{document}
\title{Orbit recovery under the rigid motions group}

\author[1]{Amnon Balanov\thanks{Corresponding author: \url{amnonba15@gmail.com}}}
\author[1]{Tamir Bendory}
\author[2]{Dan Edidin}

\affil[1]{\normalsize School of Electrical and Computer Engineering, Tel Aviv University, Tel Aviv 69978, Israel}

\affil[2]{Department of Mathematics, University of Missouri, Columbia, MO 65211, USA}

\maketitle
\begin{abstract}
We study the orbit recovery problem under the rigid-motion group $\mathsf{SE}(n)$, where the objective is to reconstruct an unknown signal from multiple noisy observations subjected to unknown rotations and translations. This problem is fundamental in signal processing, computer vision, and structural biology.

Our main theoretical contribution is bounding the sample complexity of this problem. We show that if the $d$-th order moment under the rotation group $\mathsf{SO}(n)$  uniquely determines the signal orbit, then orbit recovery under $\mathsf{SE}(n)$ is achievable with $N\gtrsim \sigma^{2d+4}$ samples as the noise variance $\sigma^2 \to \infty$. The key technical insight is that the $d$-th order $\mathsf{SO}(n)$ moments can be explicitly recovered from $(d+2)$-order $\mathsf{SE}(n)$ autocorrelations, enabling us to transfer known results from the rotation-only setting to the rigid-motion case. 
We further harness this result to derive a matching bound to the sample complexity of the multi-target detection model that serves as an abstract framework for electron-microscopy-based technologies in structural biology, such as single-particle cryo-electron microscopy (cryo-EM) and cryo-electron tomography (cryo-ET). 

Beyond theory, we present a provable computational pipeline for rigid-motion orbit recovery in three dimensions. Starting from rigid-motion autocorrelations, we extract the $\mathsf{SO}(3)$ moments and demonstrate successful reconstruction of a 3-D macromolecular structure. Importantly, this algorithmic approach is valid at any noise level, suggesting that even very small macromolecules---long believed to be inaccessible using structural biology electron-microscopy-based technologies---may, in principle, be reconstructed given sufficient data.

\end{abstract}

\newpage
\setcounter{tocdepth}{2} 
\tableofcontents

\newpage
\section{Introduction}

\subsection{Problem setup}

We study the problem of recovering an unknown signal from multiple noisy observations, each of which is transformed by an unknown element of the rigid motion group, comprising both an unknown rotation and translation. This fundamental inverse problem arises across diverse fields, such as signal processing, computer vision, and structural biology \cite{frank2006three,hartley2003multiple,ozyecsil2017survey, barfoot2024state,myronenko2010point,bartesaghi2008classification}. 

Formally, we are given $N$ noisy observations $y_0, y_1, \dots, y_{N-1}$, each generated by applying an unknown transformation of a group $G$ to an underlying signal $f: \mathbb{R}^n \to \mathbb{R}$, followed by additive Gaussian noise~\cite{bandeira2014multireference}:
\begin{align}
y_i = g_i \cdot f + \xi_i, \quad g_i \in G. \label{eqn:MRA_model}
\end{align}
Here, $G$ is a known group acting on functions by $g \cdot f(\bm{x}) \triangleq f(g^{-1} \bm{x})$, where $\bm{x} \in \mathbb{R}^n$ denotes spatial coordinates, and each $g_i$ is independently drawn from a probability distribution $\rho$ over $G$. The noise terms $\{\xi_i\}_{i=0}^{N-1}$ are independent realizations of white Gaussian noise with variance $\sigma^2$. Since the statistics of $y$ is the same for $f$ and $g\cdot f$ for any $g\in G$, we seek to recover the $G$-orbit of $f$, that is, the set $\{g\cdot f\,|\, g\in G\}$. This formulation gives rise to the name \emph{orbit recovery problem}, also commonly referred to in the literature as the multi-reference alignment (MRA) problem. 
Throughout this work, we assume that $f$ is a continuous function supported within a ball of radius $R$, which is a natural and widely applicable assumption in scientific imaging problems, including those arising in this work.

\paragraph{The special orthogonal group and the rigid motion group.}
We focus on two practically important choices for $G$: the special Euclidean group $\mathsf{SE}(n)$, which forms the core of our analysis, and the special orthogonal group $\mathsf{SO}(n)$, corresponding to pure rotational uncertainty. We refer to these two settings, respectively, as \emph{orbit recovery under $\mathsf{SE}(n)$} and \emph{orbit recovery under $\mathsf{SO}(n)$}. 

The orbit-recovery problem under compact groups, such as $\mathsf{SO}(n)$, has been extensively studied owing to its analytical tractability and rich mathematical structure, e.g.,~\cite{bandeira2023estimation,perry2019sample,abbe2018multireference}.
In contrast, the non-compact setting has received substantially less attention. Only recently have researchers begun to explore extensions of MRA to non-compact transformations, such as translations and dilations~\cite{yin2024bispectrum}.
However, the problem of orbit recovery under the full rigid-motion group $\mathsf{SE}(n)$, where both the orientation and position of the underlying signal are unknown, and which is central to many real-world applications, remains largely unexplored, particularly in terms of sample-complexity analysis\footnote{Here, \emph{sample complexity} refers to the scaling of the number of independent observations required to reliably recover the orbit as the noise level increases. A rigorous definition is provided in Section~\ref{sec:sampleComplexityRigidMotion}.} and the construction of group-invariant representations.
By addressing this problem, we aim to bridge this gap and provide a rigorous foundation for inverse problems involving rigid motions. 

A comparison between the orbit recovery problems under $\mathsf{SO}(n)$ and $\mathsf{SE}(n)$ (Figure~\ref{fig:1}(b-c)) highlights a fundamental distinction: the rigid motion setting introduces an additional layer of uncertainty due to unknown translations, in addition to rotations. 
This added indeterminacy substantially increases the complexity of the recovery problem. 
Importantly, it was shown in~\cite{balanov2025note} that the sample complexity required for the orbit recovery under $\mathsf{SO}(n)$ provides a lower bound for the more general $\mathsf{SE}(n)$ case. 
The central objective of this work is to establish a complementary upper bound, thereby characterizing the sample complexity of orbit recovery under $\mathsf{SE}(n)$ in relation to the better understood $\mathsf{SO}(n)$ problem.

\begin{figure}
    \centering
    \includegraphics[width=0.95 \linewidth]{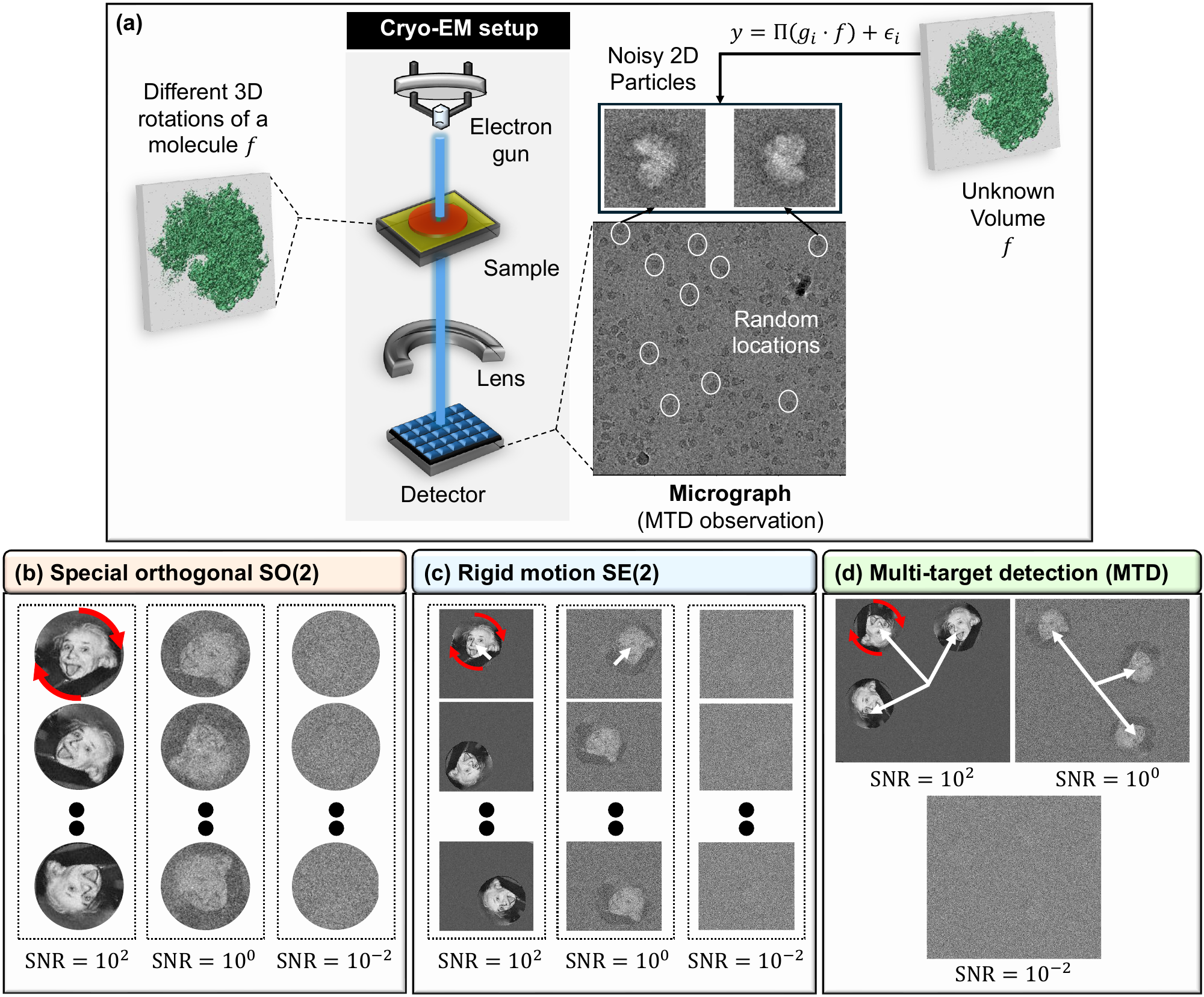}
    \caption{\textbf{Overview of the central problems studied in this work.}  
    \textbf{(a)} A primary motivation for this work arises from single-particle cryo-electron microscopy (cryo-EM), where the goal is to reconstruct 3D molecular structures from noisy 2D projection images~\cite{bendory2020single}. The raw measurements, known as \textit{micrographs}, contain many randomly oriented copies of the target molecule embedded in high levels of noise. Each particle instance is given by $f_i = \Pi(g_i \cdot f)$, where $g_i$ is a random 3D rotation and $\Pi$ is a tomographic projection operator. At low SNR, individual particles cannot be reliably detected, but direct structure recovery from micrographs may still be possible~\cite{bendory2023toward}.  
    \textbf{(b)} The orbit recovery problem under the special orthogonal group $\mathsf{SO}(n)$ (Problem~\ref{prob:orbitRecoverySOn}), where each observation is a noisy measurement of a randomly rotated signal (e.g., Einstein image), studied across varying signal-to-noise ratio (SNR) regimes.  
    \textbf{(c)} Orbit recovery under the rigid motion group $\mathsf{SE}(n)$ (Problem~\ref{prob:orbitRecoverySEn}), where each signal undergoes an unknown rotation and translation. This setting introduces additional uncertainty due to unknown spatial locations.
    \textbf{(d)} Signal recovery from MTD observations (Problem~\ref{prob:orbitRecoveryMTD}), where multiple rotated instances of a signal $f$ are embedded at unknown positions in a noisy measurement. At high SNR, particle instances can be detected, aligned, and averaged. However, at low SNR, detecting the particle locations becomes infeasible. In this challenging regime, addressed in this work, the goal is to characterize the sample complexity required for successful signal recovery despite the presence of both rotational and positional uncertainty, in high levels of noise. }
    \label{fig:1} 
\end{figure}

\subsection{Method of moments for orbit recovery}
A central concept in orbit recovery for compact groups is that of \emph{invariant features}, quantities that remain unchanged under the group action. 
In this work, we focus primarily on \emph{invariant moments}. Such moments play a fundamental role in the classical statistical approach to parameter estimation, as well as in structure-determination problems \cite{kam1980reconstruction,levin20183d}.
For orbit recovery under compact groups, the minimal moment order required for orbit identifiability directly governs the sample complexity. Specifically, if moments up to order $d$ suffice to uniquely determine the orbit, then the number of independent observations $N$ must grow faster than $\sigma^{2d}$ as the noise variance $\sigma^2 \to \infty$. In other words, the required number of observations scales as $N = \omega(\sigma^{2d})$, meaning that $N / \sigma^{2d} \to \infty$ as both $N \to \infty$ and  $\sigma \to \infty$ \cite{abbe2018estimation,bandeira2023estimation}.

Extending orbit recovery from compact groups to non-compact groups, such as the rigid motion group $\mathsf{SE}(n)$, poses significant challenges. 
Unlike compact groups, $\mathsf{SE}(n)$ admits no finite Haar measure and has no non-trivial finite-dimensional irreducible representations. 
In particular, its representations cannot be decomposed into a direct sum of irreducibles. 
Consequently, many of the standard tools available in the compact setting, such as uniform averaging over the group and orthogonal projection onto irreducible components, are no longer applicable. 
This makes the construction of non-trivial invariant statistics substantially more difficult and limits the applicability of classical representation-theoretic techniques. 
As a result, both the theoretical characterization of sample complexity and the development of effective orbit recovery algorithms remain considerably more challenging in the non-compact setting. 

\subsection{Motivating questions and main contributions}

In what follows, we outline the central questions that motivate the study and the main contributions.

\paragraph{Orbit recovery under the rigid motion group.}

To analyze the sample complexity of the orbit recovery problem under $\mathsf{SE}(n)$, our strategy is to relate the rigid motion problem to the better understood compact setting of orbit recovery under $\mathsf{SO}(n)$ through a structural comparison of the moments in the two settings\footnote{In this paper, we use \emph{autocorrelations} for the $\mathsf{SE}(n)$-invariant quantities (Definition~\ref{def:autoCorrelationNoiseFree}) and \emph{moments} for the $\mathsf{SO}(n)$-invariant quantities (Definition~\ref{def:mraTensorMoment}).}.
This comparison forms the backbone of our analysis: it provides a constructive mechanism for transferring identifiability guarantees and sample complexity bounds from the rotation-only case to the more general rigid motion case. 
Specifically, our approach is guided by the following question:
\begin{quote}
    \emph{How can the  $\mathsf{SO}(n)$  moments be recovered from the $\mathsf{SE}(n)$ autocorrelations?}
\end{quote}

As a key technical tool, we develop an algebraic procedure for extracting rotation-group moments from rigid-motion autocorrelations. 
In particular, we show that the $(d+2)$-order autocorrelations in the $\mathsf{SE}(n)$ problem suffice to recover the $d$-th order moments under $\mathsf{SO}(n)$.
This extraction is formalized in Theorem~\ref{thm:reductionFromAutocorrelationToTensorMoment} and is stated informally below. 

\begin{thm}[Informal: Extracting ${d}$-th order rotational moments from ${(d+2)}$-order rigid motion autocorrelations]
\label{thm:informalTheorem1}
Consider the problem of recovering a signal $f$ up to its orbit under either the rotation group $\mathsf{SO}(n)$ or the rigid motion group $\mathsf{SE}(n)$ (as defined in \eqref{eqn:MRA_model}). Then, any order-$d$ moment under the action of $\mathsf{SO}(n)$ can be recovered from an order-$(d+2)$ autocorrelation under the action of $\mathsf{SE}(n)$.
\end{thm}

As a direct consequence of Theorem~\ref{thm:informalTheorem1}, we have the following corollary.
\begin{corollary}
\label{cor:informalCorollary1}
If the $d$-th order $\mathrm{SO}(n)$ moment uniquely determines the orbit of $f$, then the order-$(d+2)$ $\mathsf{SE}(n)$ autocorrelation also uniquely determines the orbit of $f$ under rigid motions.
\end{corollary}

Together, Theorem~\ref{thm:informalTheorem1} and Corollary~\ref{cor:informalCorollary1} establish a structural relation between the orbit recovery problems under $\mathsf{SE}(n)$ and $\mathsf{SO}(n)$. 
They show that the moments of order $d$ for $\mathsf{SO}(n)$ can be explicitly derived from the autocorrelations of order $d+2$ for $\mathsf{SE}(n)$, thereby allowing identifiability guarantees and complexity bounds to transfer between the two models.

\paragraph{The multi-target detection model.} 
The orbit recovery problem under $\mathsf{SE}(n)$ is closely connected to the multi-target detection (MTD) model~\cite{bendory2019multi,bendory2020single}, a mathematical abstraction that captures key aspects of real-world challenges in imaging sciences, particularly in structural biology applications, such as cryo-electron microscopy (cryo-EM) and cryo-electron tomography (cryo-ET) (see Figure~\ref{fig:1}(a)). 

In the MTD model, multiple randomly rotated copies of a signal are embedded at unknown locations within a single noisy observation. 
Formally, let $f: \mathbb{R}^n \to \mathbb{R}$ denote the unknown signal, and consider $N$ transformed instances $f_i(\bm{x}) = \mathcal{R}_i \cdot f(\bm{x}) = f(\mathcal{R}_i^{-1}\bm{x})$, for $i \in \{0,1,\ldots,N-1\}$, where each $\mathcal{R}_i \in \mathsf{SO}(n)$ is drawn i.i.d. from a distribution $\rho$ on $\mathsf{SO}(n)$. 
The observation then takes the form 
\begin{align}
    y(\bm{x}) = \sum_{i=0}^{N-1} \mathcal{T}_i(\mathcal{R}_i \cdot f(\bm{x})) + \xi(\bm{x}), 
    \label{eqn:mtdMainPres}
\end{align}
where $\mathcal{T}_i$ denotes a translation operator and $\xi$ is additive Gaussian white noise with variance $\sigma^2$. 
The goal is to recover the $\mathsf{SO}(n)$-orbit of the underlying signal $f$ under the action of $\mathsf{SO}(n)$. 
A formal definition is given in Problem~\ref{prob:orbitRecoveryMTD}, and a visual illustration appears in Figure~\ref{fig:1}(d). 

Traditional reconstruction pipelines in cryo-EM are based on detecting and extracting signal occurrences from the MTD observation $y$ (defined in~\eqref{eqn:mtdMainPres}), followed by alignment and averaging the signals~\cite{bendory2020single}. However, in high-noise regimes, detection becomes unreliable, and these methods fail. Recent approaches have shifted toward detection-free reconstruction, leveraging the statistical structure of the data to recover the signal directly from MTD observations without explicitly locating signal instances~\cite{bendory2019multi, kreymer2022two, bendory2023toward, kreymer2025expectation}.

\paragraph{The sample complexity problem.}
Despite significant algorithmic progress, a fundamental theoretical question remains open: \begin{quote}
    \emph{How many signal occurrences are required to achieve a target reconstruction accuracy in the orbit recovery problem under SE(n) and the MTD model?}
\end{quote}

This work addresses this question by analyzing the sample complexity of orbit recovery in high-noise regimes, precisely where conventional detection methods fail.
To answer this question, we develop two key reductions. 
First, we reduce the orbit recovery problem under $\mathsf{SE}(n)$ to its compact counterpart over $\mathsf{SO}(n)$ by establishing an explicit algebraic relation between their invariant moments, as presented informally in Theorem~\ref{thm:informalTheorem1}. 
Second, we reduce the MTD problem to the orbit recovery problem under $\mathsf{SE}(n)$, thereby unifying the analysis of both settings. 
These reductions enable us to transfer known sample complexity results from the simpler rotational model to the more challenging scenarios involving rigid motions. 

The main statistical implication of Theorem~\ref{thm:informalTheorem1} is that the sample complexity of the orbit recovery problem in both $\mathsf{SE}(n)$ and the MTD model is tightly coupled to the sample complexity of the simpler rotation-only $\mathsf{SO}(n)$ model.
Specifically, it is lower bounded by the sample complexity of the $\mathsf{SO}(n)$ problem, and in the low-SNR regime ($\sigma \to \infty$), it is upper bounded by a factor of $\sigma^4$ relative to $\mathsf{SO}(n)$ problem. 
An informal statement is given below; formal versions appear in Theorems~\ref{thm:mainTheoremRigidmotion} and~\ref{thm:mainTheorem}.

\begin{thm}[Informal: Sample complexity of the orbit recovery problem under $\mathsf{SE}(n)$ and of the MTD model]
\label{thm:informalMainTheorem}
Assume that the signal in the orbit recovery problem under $\mathsf{SO}(n)$ is determined by $d$-th order moments of the observations. 
Then, the sample complexity of the orbit recovery problem under $\mathsf{SE}(n)$ and in the MTD model is lower bounded by $\omega\!\left(\sigma^{2d}\right)$ and upper bounded by $\omega\!\left(\sigma^{2d+4}\right)$ in the low-SNR regime, as $\sigma,N \to \infty$.
\end{thm}

\paragraph{Provable algorithms.} Finally, this work presents the first complete end-to-end pipeline for rigid-motion orbit recovery in both two and three dimensions, that is, for $\mathsf{SE}(2)$ and $\mathsf{SE}(3)$.
In particular, we demonstrate in the three-dimensional case the reconstruction of a molecular volume directly from its $\mathsf{SE}(3)$ autocorrelations.
These proof-of-concept experiments validate the proposed extraction framework of the $\mathsf{SO}(3)$ moments from the $\mathsf{SE}(3)$ autocorrelations and show that, under a simplified cryo-ET model (see Section~\ref{sec:applications}), accurate recovery is  possible, in principle, at any SNR, provided that a sufficiently large number of observations is available.

\paragraph{Work structure.} 
The next section formally introduces the problems addressed in this paper.
In Section~\ref{sec:reductionFromAutoCorrealtionToTensorMoment} we establish the algebraic extraction of the $d$-th order $\mathsf{SO}(n)$ moment by the order-$(d+2)$ $\mathsf{SE}(n)$ autocorrelation.
Our main statistical results are presented in Section~\ref{sec:sampleComplexityRigidMotion} for the orbit recovery problem under $\mathsf{SE}(n)$ and in Section~\ref{sec:sampleComplexityMTD} for the MTD model.
Importantly, our analysis not only establishes asymptotic bounds but also provides a concrete, explicit algorithmic framework for orbit recovery.
We provide a proof of concept by designing  provable algorithms, as demonstrated in Section~\ref{sec:empirical}, which also presents empirical validation of the underlying algebraic and statistical theory.

\section{Problem statements} \label{sec:problemsStatments}
In this section, we formally introduce the three fundamental problems studied in this work, illustrated in Figure~\ref{fig:1}(b-d). 
We begin with the well-studied case of orbit recovery under the action of the special orthogonal group $G=\mathsf{SO}(n)$ (Problem~\ref{prob:orbitRecoverySOn}), which serves as a baseline for our analysis. 
Building on this foundation, we then consider two related settings: orbit recovery under the rigid motion group (Problem~\ref{prob:orbitRecoverySEn}), and the MTD model (Problem~\ref{prob:orbitRecoveryMTD}). 
These problems are related but not ordered by strict generality: orbit recovery under $\mathsf{SO}(n)$ treats rotations only; orbit recovery under $\mathsf{SE}(n)$ augments the model with translations; and the MTD model addresses detection with multiple rotated copies placed at unknown locations within a single noisy observation. 
Under the separation assumptions specified later, the MTD analysis reduces to the orbit recovery problem under $\mathsf{SE}(n)$, while $\mathsf{SO}(n)$ provides a compact group baseline for moment and sample-complexity comparisons.

\paragraph{Notation.}
We begin with some preliminary definitions and notational conventions that will be used throughout the paper.
Let $\mathcal{B}_R^{(n)}(\bm{x})$ denote the $n$-dimensional Euclidean ball of radius $R$ centered at $\bm{x} \in \mathbb{R}^n$; when centered at the origin, we write $\mathcal{B}_R^{(n)} \triangleq  \mathcal{B}_R^{(n)}(\bm{0})$. The volume of this ball is denoted by $V_n(R)$. The $(n-1)$-dimensional sphere of radius $R$ is denoted by $S^{n-1}(R)$.
We use the notation $\xrightarrow[]{\text{a.s.}}$ to denote almost sure convergence of a sequence of random variables. The index set $\{0, 1, \dots, N-1\}$ is abbreviated as $\pp{N}$. The Frobenius norm is denoted by $\|\cdot\|_F$, and the $L^2$ norm by $\|\cdot\|_{L^2}$. The space $C^0(\mathcal{D})$ denotes the set of continuous functions on a compact domain $\mathcal{D} \subset \mathbb{R}^n$.
Finally, we use $\mathcal{R} \in \mathsf{SO}(n)$ to denote elements of the special orthogonal group (i.e., rotations), $\mathcal{T}$ to denote translation operators, and $g \in \mathsf{SE}(n)$ to denote elements of the rigid motion group (the special Euclidean group).
Throughout this work, we assume the unknown signal $f$ is supported on a ball of radius $R$ (i.e., the closed ball $B_R^{(n)}\subset\mathbb{R}^n$).

\subsection{Orbit recovery under the special orthogonal group SO(n)}
\begin{problem} [Orbit recovery for functions under the special orthogonal group] \label{prob:orbitRecoverySOn}
    Let $V$ be a $p$-dimensional representation of $G=\mathsf{SO}(n)$, where $V\subset L^2(\mathcal{B}_R^{(n)})$ is $G$-invariant and $\dim V=p$. Let $f \in V$ be the unknown signal. For each $i \in \pp{N}$, we observe
    \begin{align}
        y_i = \mathcal{R}_i \cdot f + \xi_i,
    \end{align}
    where the group action is defined by ${\mathcal{R}_i \cdot f (\bm{x}) = f (\mathcal{R}_i^{-1}\bm{x})}$. The rotations $\mathcal{R}_i \in \mathsf{SO}(n)$ are drawn independently from a distribution $\mathcal{R}_i \sim \rho$ and the noise terms $\xi_i \sim \mathcal{N}(0, \sigma^2 I_{p})$ are i.i.d. Gaussian. The sequences $\mathcal{R}_i$ and $\xi_i$ are mutually independent. The goal is to estimate the orbit of $f$ under $\mathsf{SO}(n)$. 
\end{problem} 
Problem~\ref{prob:orbitRecoverySOn} has been extensively studied; see, for example,~\cite{bandeira2023estimation}. A key result for orbit recovery problems over general compact groups, including $\mathsf{SO}(n)$, is that both the mean-squared error (MSE) and the sample complexity in the high-noise regime $\sigma\to\infty$ are governed by the minimal moment that uniquely identifies the orbit of the signal, as defined in Definition~\ref{def:mraTensorMoment}.

\subsection{Orbit recovery under the rigid motion group SE(n)} 

We start with the definition of the special Euclidean group.
\begin{definition} [The Special Euclidean group $\mathsf{SE}(n)$]
The Special Euclidean group $\mathsf{SE}(n)$ is the group of orientation-preserving isometries of $\mathbb{R}^n$, that is, the group of rigid motions that preserve both distances and orientation. 
\end{definition}

It is well-known that $\SE(n)$ decomposes as a semi-direct product $\mathsf{SE}(n) = \mathsf{SO}(n) \ltimes \mathbb{R}^n$. 
This means that as a manifold, the locally compact Lie group $\SE(n)$ can be identified with $\SO(n) \times \R^n$ 
because every element can be uniquely expressed as a pair $(\mathcal{R}, \bm{t})$, where $\mathcal{R} \in \mathsf{SO}(n)$ is a rotation and $\bm{t} \in \mathbb{R}^n$ is a translation. 
The group  operation is given by $(\mathcal{R}_1, \bm{t}_1)\cdot(\mathcal{R}_2, \bm{t}_2) = (\mathcal{R}_1 \mathcal{R}_2, \mathcal{R}_1 \bm{t}_2 + \bm{t}_1)$.

The main challenge posed by translations in the orbit recovery problem under $\mathsf{SE}(n)$ is that the signal is no longer observed at a fixed location but over randomly shifted regions of space. Because translations act non-compactly on $\mathbb{R}^n$, it is necessary to ensure that both the transformed signals and the additive noise are well defined and have finite energy. To this end, we restrict all observations to a fixed compact domain $\mathcal{D} \subset \mathbb{R}^n$ that contains every possible translated and rotated copy of the signal. This domain $\mathcal{D}$ serves a dual purpose: it bounds the support of the translated signals and guarantees that the energy of the white-noise process remains finite. %With this setup, we can formally define the autocorrelations and state our main result. 
With this setup, we can formally define the orbit recovery problem under the rigid motion group.
%\amnon{Note that I changed the location of this paragraph to here, instead of before Definition 4.2.}

\begin{problem}[Orbit recovery under the rigid motion group] \label{prob:orbitRecoverySEn}
Let $V\subset L^2(\mathcal{B}_R^{(n)})$ be as in Problem~\ref{prob:orbitRecoverySOn}, and let $f \in V$ be the unknown signal, viewed as a function on $\mathbb{R}^n$ by zero extension outside $\mathcal{B}_R^{(n)}$ (that is, $f(\bm{x}) = 0$ for $\|\bm{x}\|>R$).
For each $i \in \pp{N}$, we observe a function $y_i : \mathcal{D} \to \mathbb{R}$, defined on a compact support $\mathcal{D}$, and given by
\begin{align}
    y_i = g_i \cdot f + \xi_i, \label{eqn:modelSEnProb}
\end{align}
where $g_i = (\mathcal{R}_i, \bm{t}_i) \in \mathsf{SE}(n)$ acts on $f$ via
\begin{align}
    (g_i \cdot f)(\bm{x}) = f\big(\mathcal{R}_i^{-1}(\bm{x} - \bm{t}_i)\big), \quad \bm{x} \in \mathcal{D}. \label{eqn:groupActionRigidMotion}
\end{align}
The elements $\{g_i\}_{i=0}^{N-1}$ are i.i.d. samples drawn from a distribution $\mu_{\mathsf{SE}(n)}$ on $\mathsf{SE}(n)$,  the noise terms $\{\xi_i\}_{i=0}^{N-1}$ are i.i.d. Gaussian white noise processes with variance $\sigma^2$, satisfying
\begin{align}
    \mathbb{E}[\xi_i(\bm{x}) \xi_i(\bm{x}')] = \sigma^2 \delta(\bm{x} - \bm{x}'), \quad \bm{x}, \bm{x}' \in \mathcal{D},
\end{align}
where $\delta$ denotes the Dirac delta function. The sequences $\{g_i\}$ and $\{\xi_i\}$ are mutually independent.
The goal is to estimate the \em{$\SO(n)$} orbit of $f$ under the group action of $\mathsf{SE}(n)$. 
\end{problem}

\begin{remark}
Note that although the model in Problem~\ref{prob:orbitRecoverySEn} operates under the rigid-motion group $\mathsf{SE}(n)$ (rotations and translations), the reconstruction is determined only up to a global rotation. This follows from our assumption on the signal's support, as explained later.
\end{remark}

\subsection{Orbit recovery of the multi-target detection (MTD) problem}
The orbit recovery problem under $\mathsf{SE}(n)$ has a direct connection to the MTD model, which has numerous applications in cryo-electron microscopy imaging. In the MTD framework, multiple signals $\{f_i\}_{i=0}^{N-1}$, each obtained by applying a random rotation to an underlying signal $f$, are embedded at unknown locations within a large, noisy observation $y$~\cite{bendory2019multi}. %We next provide a formal definition of this problem.

\begin{problem}[Orbit recovery for the multi-target detection model] \label{prob:orbitRecoveryMTD}
Let $V \subset L^2(\mathcal{B}_R^{(n)})$ be as in Problem~\ref{prob:orbitRecoverySOn}, and let $f \in V$ be the unknown signal. Let $G = \mathsf{SO}(n)$ denote the special orthogonal group acting on $\mathbb{R}^n$.
We observe a single MTD measurement $y : [-MR, MR]^n \to \mathbb{R}$ given by
\begin{align}
    y = \sum_{i=0}^{N-1} \mathcal{T}_i(\mathcal{R}_i \cdot f) + \xi,
\end{align}
where $N \in \mathbb{N}$ is the number of signal occurrences; the rotations $\{ \mathcal{R}_i \}_{i=0}^{N-1}$ are i.i.d. samples from a distribution $\rho$ on $\mathsf{SO}(n)$; the translation operators $\{ \mathcal{T}_i \}_{i=0}^{N-1}$ are deterministic but unknown and bounded to the support within the observation domain $[-MR, MR]^n$; and ${\xi}$ is Gaussian white noise process on $[-MR, MR]^n$, that is, a zero-mean Gaussian process with covariance
\begin{align}
    \mathbb{E}[\xi(\bm{x}) \xi(\bm{x}')] = \sigma^2 \delta(\bm{x} - \bm{x}'), \quad \bm{x}, \bm{x}' \in [-MR, MR]^n. \label{eqn:covarianceMatrixMTD}
\end{align}
The sequence $\{\mathcal{R}_i \}_{i \in \pp{N}}$ is independent on $\xi$. The goal is to estimate the orbit of $f$ under the group action of $\mathsf{SO}(n)$. 
\end{problem}

\subsection{Assumptions on the signal} 
\label{sec:assumptionsOnSignal}
Let $f\in \mathbb{R}^n \to \mathbb{R}$ be the signal to be estimated. In spherical coordinates, the signal $f$ can be written as $f(r, \varphi_1, \dots, \varphi_{n-1})$, where the angular coordinates satisfy $\varphi_i \in [0, \pi)$ for $1 \leq i \leq n-2$, and $\varphi_{n-1} \in [0, 2\pi)$. For notational convenience, we denote this as $f(r, \bm{\varphi})$, where $\bm{\varphi} = (\varphi_1, \dots, \varphi_{n-1})$. We define the antipodal point of $(R, \bm{\varphi})$ on the sphere ${S}^{n-1}$ as:
\begin{align}
  (-R,\bm{\varphi}) = (R, \pi - \varphi_1, \pi - \varphi_2, \dots, \pi - \varphi_{n-2}, 2\pi - \varphi_{n-1}).
\end{align}
Then, we have the following assumptions on the signal $f$. 

\begin{assum}[Compact support] \label{assum:support}
    The signal $f \in C^0(\mathcal{B}_R^{(n)})$ is assumed to be continuous and supported on the ball of radius $R$ centered at the origin. In spherical coordinates, this means that $f(r, \varphi_1, \dots, \varphi_{n-1}) = 0$ for all $r > R$.
\end{assum}

\begin{assum}[Non-vanishing antipodal correlation] \label{assum:nonVanishingSupport}
    The following antipodal correlation at the boundary of the support is non-zero:
    \begin{align}
        C_{f} \triangleq \int_{S^{n-1}} f(R, \bm{\varphi}) \, f(-R, \bm{\varphi}) \, d\bm{\varphi} \neq 0, \label{eqn:antipodalProdSum}
    \end{align}
    where $S^{n-1}$ is the $(n-1)$-dimensional unit sphere, and $d\bm{\varphi}$ is its standard surface area element, given in spherical coordinates by
    \begin{align}    
        d\bm{\varphi} = \sin^{n-2}(\varphi_1) \sin^{n-3}(\varphi_2) \cdots \sin(\varphi_{n-2}) \, d\varphi_1 \, d\varphi_2 \cdots d\varphi_{n-1}.
    \end{align}
    We refer to the quantity $C_{f}$ as the \emph{antipodal correlation of $f$}.
\end{assum}

\begin{lem}[Rotational invariance of antipodal correlation under group action] \label{lem:invarianceOfAntipodalProdSum} 
    The antipodal correlation $C_{f}$, defined in~\eqref{eqn:antipodalProdSum}, is invariant under the action of the rotation group $G = \mathsf{SO}(n)$. That is, for all $\mathcal{R} \in \mathsf{SO}(n)$, we have
    \begin{align}
        C_{\mathcal{R} \cdot f} = C_{f}.
    \end{align}
\end{lem}
Lemma~\ref{lem:invarianceOfAntipodalProdSum} is proved in Appendix~\ref{sec:proofOfinvarianceOfAntipodalProdSum}. See Figure~\ref{fig:2}(b) for an illustration of this property.

\begin{figure*}[!t]
    \centering
    \includegraphics[width=0.4 \linewidth]{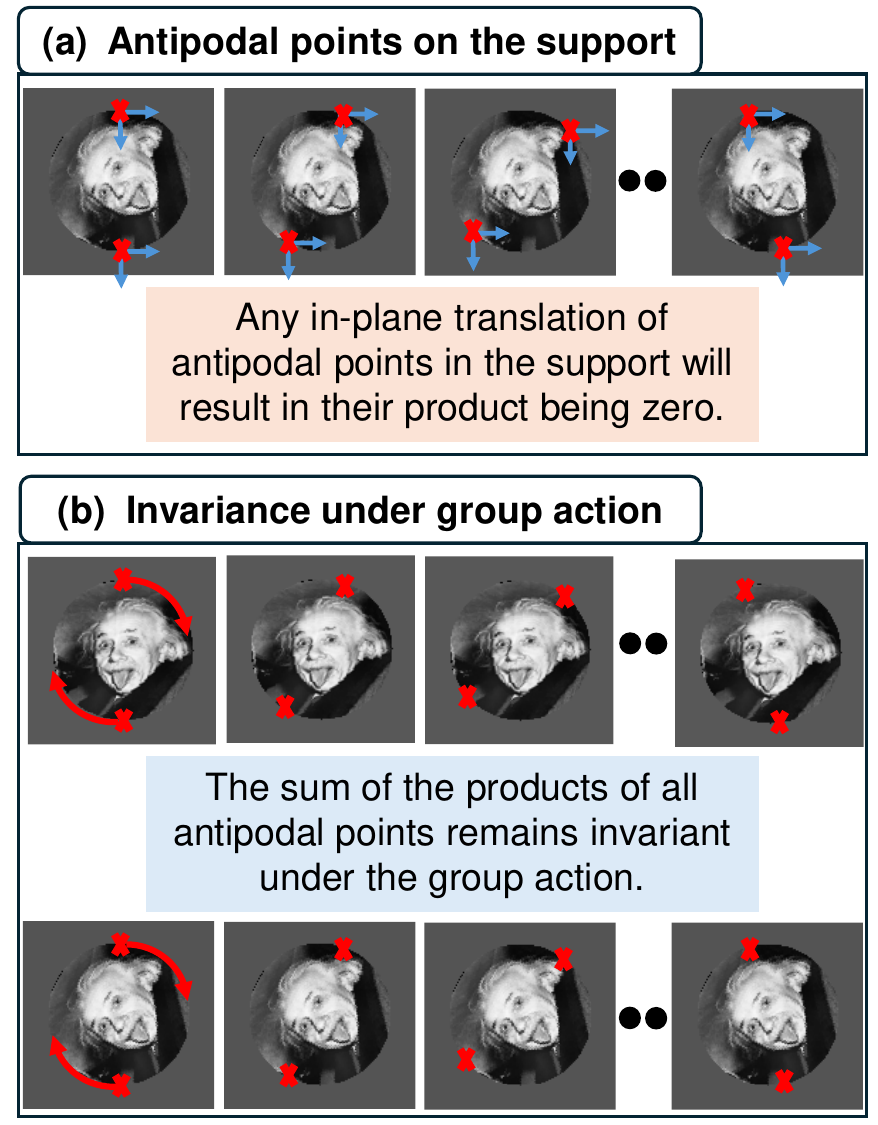}
    \caption{\textbf{Properties of antipodal points.}  
    \textbf{(a)} Applying the same translation to both antipodal points in the support of the signal results in their product being zero, since at least one of the translated antipodal points lies outside the support of the signal $f$.  
    \textbf{(b)} The integral of the product of antipodal points over the support of the signal is invariant under the group action (see Lemma~\ref{lem:invarianceOfAntipodalProdSum}).}
    \label{fig:2}
\end{figure*}

\subsection{Prior work on orbit recovery and synchronization over SE(n)} 
To our best knowledge, the orbit recovery problem under $\mathsf{SE}(n)$ has not been analyzed before. Yet, a substantial literature addresses estimation under rigid motions, but with objectives and statistical models that differ from the orbit recovery setting studied here.

\paragraph{Compactification and invariant construction.}
The authors of~\cite{bendory2022compactification} introduced a compactification of $\mathsf{SE}(2)$ by mapping planar images to a compact domain on which $\mathsf{SO}(3)$-based invariants can be computed.
They prove that the projection is approximately information preserving in the sense that rigid motions map to neighborhoods of rotations on the sphere,  and realize practical approximate $\mathsf{SE}(2)$ invariants via the spherical bispectrum. 
These results provide an effective pipeline for feature construction under rigid motions, but they do not yield sharp sample-complexity guarantees in the high-noise regime, which is the focus of the present work.

\paragraph{Synchronization over $\mathsf{SE}(n)$.}
Synchronization is the problem of estimating absolute poses $\{g_i\in \mathsf{SE}(n)\}$ from noisy pairwise relative ratios $g_{ij}\approx g_i g_j^{-1}$~\cite{singer2011angular}. 
State-of-the-art methods in $\mathsf{SE}(3)$ leverage algebraic and spectral techniques, such as dual-quaternion eigenvector formulations~\cite{hadi2024se}, spectral synchronization~\cite{arrigoni2016spectral}, and robust low-rank~\cite{arrigoni2018robust}, and, more broadly, contraction methods for Cartan motion groups (including $\mathsf{SE}(n)$) that reduce to compact surrogates~\cite{ozyesil2018synchronization}. 
These pipelines are powerful at moderate noise levels, yet they sharply degrade at low SNR: pairwise estimates cease to concentrate, spectral gaps collapse, and robustness guarantees require conditions not met in the high-noise regime. 
Accordingly, these techniques fail to produce accurate estimates in the low-SNR regime.

\paragraph{Our contribution relative to prior art.}
We study orbit recovery rather than synchronization (pairwise relations among many poses), and we provide \emph{high-noise} theory. 
Our main algebraic novelty is an explicit extraction of the $d$-th order $\mathsf{SO}(n)$ moment from the $(d{+}2)$-order $\mathsf{SE}(n)$ autocorrelation, which in turn yields upper bounds on the sample complexity of the orbit recovery under $\mathsf{SE}(n)$ in the low-SNR regime, with downstream implications for MTD and structural biology. 
This fills the gap left by compactification-based invariants and synchronization methods, neither of which supplies sharp, high-noise sample-complexity characterizations.

\section{Extraction of SO(n) moments from SE(n)  autocorrelations} 
\label{sec:reductionFromAutoCorrealtionToTensorMoment}
This section introduces the autocorrelations of the rigid motion group $\mathsf{SE}(n)$ and their connection to the moments of the rotation group $\mathsf{SO}(n)$. 

\subsection{Moments of the special orthogonal group SO(n)}

First, we introduce the invariant features of the special orthogonal group $\mathsf{SO}(n)$.
These features are characterized by the moments of $\mathsf{SO}(n)$, which we define below.

\begin{definition} [Population moments under $\mathsf{SO}(n)$]
\label{def:mraTensorMoment}
    Let $f: \mathcal{B}_R^{(n)} \to \mathbb{R}$, be defined as in Problem \ref{prob:orbitRecoverySOn}.
    The $d$-th order moment of $f$ with respect to the distribution $\rho$ on $\SO(n)$ is defined by
    \begin{align}
        M_{f, \rho}^{(d)}(\bm{\eta}_1, \ldots, \ldots \bmeta_d) \ \triangleq \mathbb{E}_{\mathcal{R} \sim \rho} \left[\prod_{j=1}^{d}f_{\mathcal{R}}(\bmeta_j)\right]
        = \int_{\SO(n)} \prod_{j=0}^{d-1}f_{\mathcal{R}} (\bmeta_j) \, d \rho(\mathcal{R}),
    \label{eqn:sphericalCordinatedTensorMoment}
    \end{align}
where $\bmeta_1, \ldots, \bmeta_d$ lie in the ball of radius $R$.
\end{definition}
Figure \ref{fig:3}(a) illustrates the third-order moment of $\mathsf{SO}(n)$ for $n=2$. It is straightforward to observe that the moment $M_{f, \rho}^{(d)}$ remains invariant under the group action $\mathsf{SO}(n)$, for a Haar-measure of $\rho$ on $\mathsf{SO}(n)$. Many studies have explored the unique recovery of $f$ from its moments \eqref{def:mraTensorMoment}; see, for example, \cite{bandeira2023estimation}.

\subsection{Autocorrelations of the rigid motion group SE(n)} \label{sec:autoCorrealtionInvariant}
Here, we introduce a set of functions, the autocorrelations, which are invariant under rigid group actions, as will be proven later. We then discuss the conditions under which these functions uniquely determine the orbit of $f$.

\begin{figure}
    \centering
    \includegraphics[width=0.85 \linewidth]{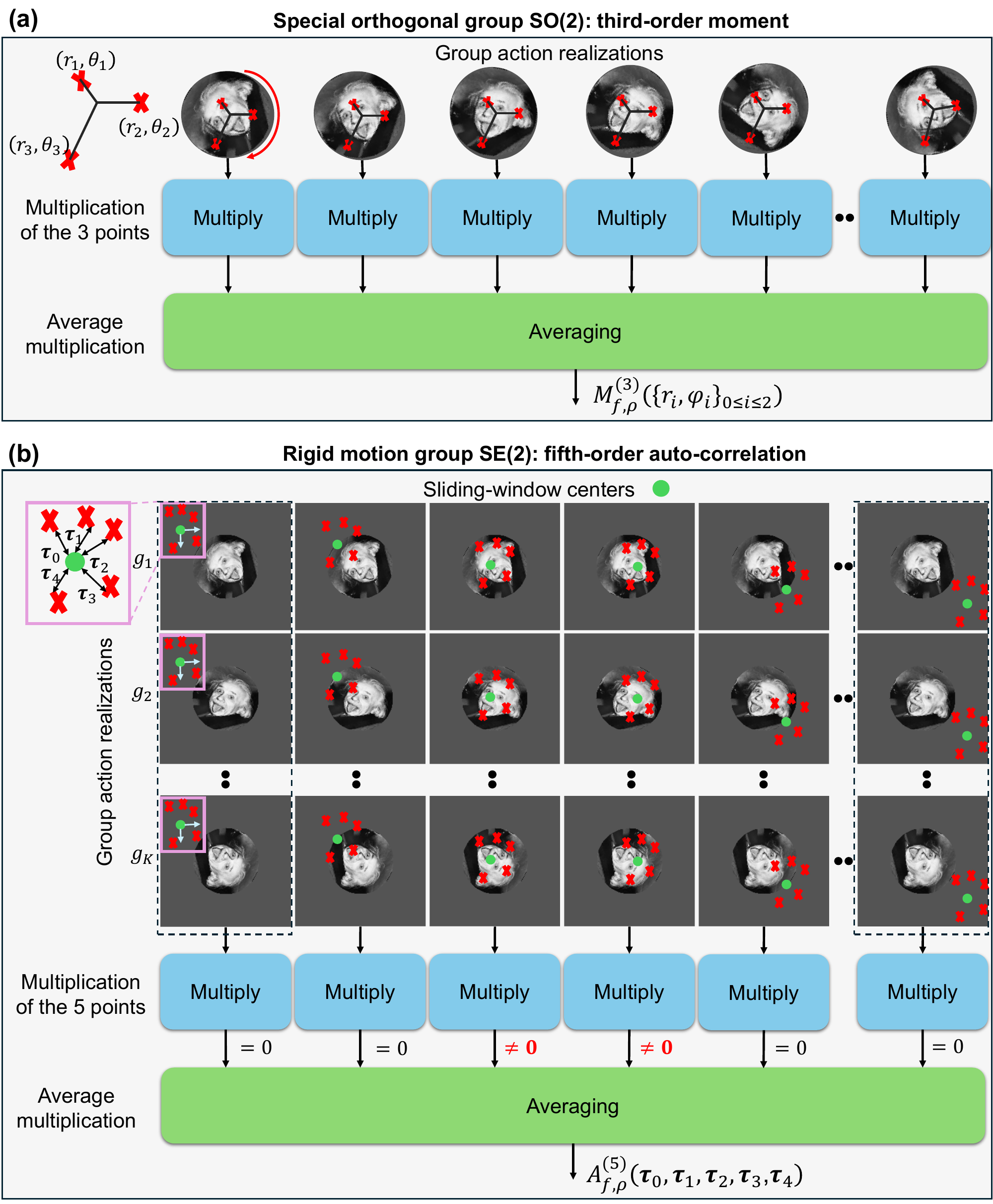}
    \caption{\textbf{Comparing $\mathsf{SO}(2)$ moments and $\mathsf{SE}(2)$ autocorrelations.}
    \textbf{(a) $\mathsf{SO}(2)$ moments:} fix $d=3$ locations in the image and average the product of their intensities over all rotations of the underlying signal (e.g., the Einstein image).
    \textbf{(b) $\mathsf{SE}(2)$ autocorrelation:} Choose five fixed relative offsets $\bm{\tau}_0,\ldots,\bm{\tau}_4$ around a window center. %Each row shows a different rigid-motion realization of the signal (rotation + translation). 
    Each row shows a different rotation of the signal. 
    Each column corresponds to a different sliding-window center (green dot). For every (row, column) pair, sample the five locations at the red offsets $\{\bm{x}_{\text{center}}+\bm{\tau}_j\}_{j=0}^4$ in the transformed image and multiply their intensities. Some products are zero when one or more samples fall outside the object, while others are nonzero when the pattern overlaps the object. Averaging these per-window products across all centers and rigid–motion realizations yields $A_f^{(5)}(\bm{\tau}_0,\bm{\tau}_1,\bm{\tau}_2,\bm{\tau}_3,\bm{\tau}_4)$.}    
    \label{fig:3}
\end{figure}

\begin{definition}[Population autocorrelations under $\mathsf{SE}(n)$]
\label{def:autoCorrelationNoiseFree}
Let $f : \mathbb{R}^n \to \mathbb{R}$ be a continuous function supported on a compact subset of $\mathbb{R}^n$. Suppose that $f$ is acted upon by elements of the rigid motion group $g = (\mathcal{R}, \bm{t}) \in \mathsf{SE}(n)$, drawn from a
distribution $\mu_{\mathsf{SE}(n)}$.  
The \emph{$d$-th order autocorrelation} of $f$ under $\mathsf{SE}(n)$ is defined as: 
\begin{align}
     A_{f}^{(d)}(\bm{\tau}_0, \ldots, \bm{\tau}_{d-1}) 
    &= \int_{\mathbb{R}^n} \mathbb{E}_{g \sim \mu_{\mathsf{SE}(n)}} \left[ \prod_{j=0}^{d-1} f_g(\bm{x} + \bm{\tau}_j) \right] d\bm{x} \nonumber\\
    &= \int_{\mathbb{R}^n} \int_{\mathsf{SE}(n)} \prod_{j=0}^{d-1} f_g(\bm{x} + \bm{\tau}_j) \, d\mu_{\mathsf{SE}(n)}(g) \, d\bm{x},     \label{eqn:autoCorrelationNoiseFreeProduct}
\end{align}
where the group action is defined by $\left((\mathcal{R}, \bm{t}) \cdot f\right)(\bm{x}) = f(\mathcal{R}^{-1}(\bm{x} - \bm{t}))$, and $\bm{\tau}_0, \ldots, \bm{\tau}_{d-1} \in \mathbb{R}^n$.
\end{definition}

In what follows, we adopt a mild assumption on the distribution of rigid motions, namely, that rotations and translations are independent (a condition met in many imaging setups). This factorization simplifies the autocorrelation formula, as stated in Lemma~\ref{lem:autocorrelationinvariant} and proved in Appendix~\ref{sec:proofOfInvariance}.

\begin{assum}[Assumptions on the distribution] \label{assum:rotation_translation_indep}
The group $\SE(n)$ is isomorphic as a manifold to the product $\SO(n) \times \R^n$.
In order to relate $\SE(n)$ autocorrelations to $\SO(n)$ moments, we assume that the probability distribution $\mu_{\SE(n)}$ is given by a product measure $\mu_{\R^n} \times \rho_{\SO(n)}$, where $\mu_{\R^n}$ is a distribution on $\R^n$ and $\rho_{\SO(n)}$ is a distribution on $\SO(n)$. Essentially, this means that the distribution of translations is independent of the distribution of rotations. When we relate between $\mathsf{SO}(n)$ moments to $\mathsf{SE}(n)$ autocorrelations, we will assume that the distribution on $\SO(n)$ is also $\rho_{\SO(n)}$. 
\end{assum}

\begin{lem} \label{lem:autocorrelationinvariant}
Under Assumption~\ref{assum:rotation_translation_indep}, the autocorrelation defined by~\eqref{eqn:autoCorrelationNoiseFreeProduct} simplifies to the following expression:
\begin{align}
    A_{f, \rho}^{(d)}(\bm{\tau}_0, \ldots, \bm{\tau}_{d-1}) 
    =\int_{\mathbb{R}^n} \int_{\mathcal{R} \in \mathsf{SO}(n)} 
    \prod_{j=0}^{d-1} f_{\mathcal{R}}(\bm{x} + \bm{\tau}_j) \, \, d\rho(\mathcal{R}) \, d\bm{x},
    \label{eqn:autoCorrelationNoiseFreeProduct2}
\end{align}
where $f_{\mathcal{R}}(\bm{x}) = f(\mathcal{R}^{-1}\bm{x})$. 
\end{lem}

For any rotation distribution $\rho$ on $\mathsf{SO}(n)$, the $d$-th order autocorrelation is invariant under the simultaneous action of $\mathsf{SE}(n)$ on the signal and on the distribution. In particular, when the rotational distribution $\rho$ is Haar (uniform),  the autocorrelation is invariant under the full action of the rigid motion group $\mathsf{SE}(n)$, as established in the next proposition. 
The proof is given in Appendix~\ref{sec:proofOfinvarianceOfAutoCorrelation}.

\begin{proposition}[Invariance of the autocorrelation under the rigid motion group] 
\label{prop:invarianceOfAutoCorrelation} 
The $d$-th order autocorrelation functions ${A}^{(d)}_{f, \rho}$, as defined in Definition~\ref{def:autoCorrelationNoiseFree}, are invariant under the action of the rigid motion group $\mathsf{SE}(n)$, provided that the rotational distribution $\rho$ is a Haar measure on $\mathsf{SO}(n)$.
\end{proposition}

Figure \ref{fig:3} presents a comparison between the autocorrelations for the rigid motion group and the moments for the special orthogonal group in two dimensions ($n=2$). Specifically, it contrasts the fifth-order autocorrelation of the rigid motion group with the third-order moment of the signal $f$ under the special orthogonal group $\mathsf{SO}(n)$.
In the $\mathsf{SO}(n)$ model, moments are defined on fixed triples of points, incorporating all rotations of the signal $f$. In contrast, the $\mathsf{SE}(n)$ autocorrelation is defined on fixed sets of five points that move transversely to account for both rotation and translation. 
While not immediately apparent, we show next that autocorrelations of order up to $d+2$ suffice to fully and explicitly recover the corresponding $\mathsf{SO}(n)$ moments of order $d$.

\subsection{Extraction of SO(n) moments from SE(n) autocorrelations} 

We now show that the $(d+2)$-order autocorrelation of the rigid motion group, defined in Definition~\ref{def:autoCorrelationNoiseFree}, enables the explicit extraction of the $d$-th order moment.
Specifically, Theorem~\ref{thm:reductionFromAutocorrelationToTensorMoment} establishes that the $d$-th order moment $M_{f, \rho}^{(d)}$ of the $\mathsf{SO}(n)$ model can be extracted from the $(d+2)$-order $A^{(d+2)}_{f, \rho}$ and second-order autocorrelations $A^{(2)}_{f, \rho}$ of the $\mathsf{SE}(n)$ model. The result shows that by evaluating the $(d+2)$-order autocorrelation on a specific slice, where two of the evaluation points are antipodal and near the boundary of the support, and appropriately normalizing it using the second-order autocorrelation, one can isolate and recover the $d$-th order moment invariant to rotations. The proof of Theorem \ref{thm:reductionFromAutocorrelationToTensorMoment} can be found in Appendix \ref{sec:auxiliartForTheoremreductionFromAutocorrelationToTensorMoment}. Figure~\ref{fig:4} provides a visual proof sketch of the theorem.

\begin{thm}[Recovery of the $d$-th order moment of the special orthogonal group from the $(d+2)$-order autocorrelation under the rigid motion group]
\label{thm:reductionFromAutocorrelationToTensorMoment}
Let $f \in C^0 (\mathcal{B}_R^{(n)})$ be a continuous signal satisfying Assumptions~\ref{assum:support} and \ref{assum:nonVanishingSupport}. 
Let $A^{(d+2)}_{f, \rho}$ denote the $(d+2)$-order autocorrelation under the rigid motion group $\mathsf{SE}(n)$ (Problem~\ref{prob:orbitRecoverySEn}), as defined in Definition~\ref{def:autoCorrelationNoiseFree}
where the distribution on $\SE(n)$ satisfies Assumption~\ref{assum:rotation_translation_indep}. Let $M^{(d)}_{f, \rho}$ be the $d$-th order moment associated with the special orthogonal group $\mathsf{SO}(n)$ (Problem~\ref{prob:orbitRecoverySOn}), as defined in Definition~\ref{def:mraTensorMoment}.

Then, for any $\bmeta_1, \ldots, \bmeta_d$ in the ball of radius $R$, the $d$-th order moment $M^{(d)}_{f,\rho}$ is given by the limit:
\begin{align}
    M_{f, \rho}^{(d)}(\bm{\eta}_1, \dots, \bm{\eta}_d) 
    = \lim_{\delta \to 0^+} 
    \frac{\displaystyle \int_{S^{n-1}} A^{(d+2)}_{f, \rho}\big(\bm{\tau}_0^{(\delta)}(\bm{\theta}), \bm{\tau}_1^{(\delta)}(\bm{\theta}), \bm{\eta}_1, \dots, \bm{\eta}_d\big) \, d\bm{\theta}}
         {\displaystyle \int_{S^{n-1}} A^{(2)}_{f, \rho}\big(\bm{\tau}_0^{(\delta)}(\bm{\theta}), \bm{\tau}_1^{(\delta)}(\bm{\theta})\big) \, d\bm{\theta}},
    \label{eqn:mainTheoremExtraction}
\end{align}
where 
\begin{align}
    \bm{\tau}_0^{(\delta)}(\bm{\theta}) = (R(1 - \delta), \bm{\theta}), \quad 
    \bm{\tau}_1^{(\delta)}(\bm{\theta}) = (-R(1 - \delta), \bm{\theta}), \label{eqn:nearlyAntipodalPoints}
\end{align}
for $\bm{\theta} \in S^{n-1}$ are antipodal points.
\end{thm}

% \paragraph{The proof idea.} 
% Figure~\ref{fig:4} provides a visual proof sketch of Theorem~\ref{thm:reductionFromAutocorrelationToTensorMoment}. It illustrates the evaluation of the $(d+2)$-order autocorrelation of the signal under the rigid motion group $\mathsf{SE}(n)$ at pairs of points that are antipodal near the boundary of the support of $f$. By taking the limit as $\delta \to 0^+$, these boundary region autocorrelations can be reduced to the $d$-th order moment under the rotation group $\mathsf{SO}(n)$.

\begin{figure*}[!t]
    \centering
    \includegraphics[width=1.0 \linewidth]{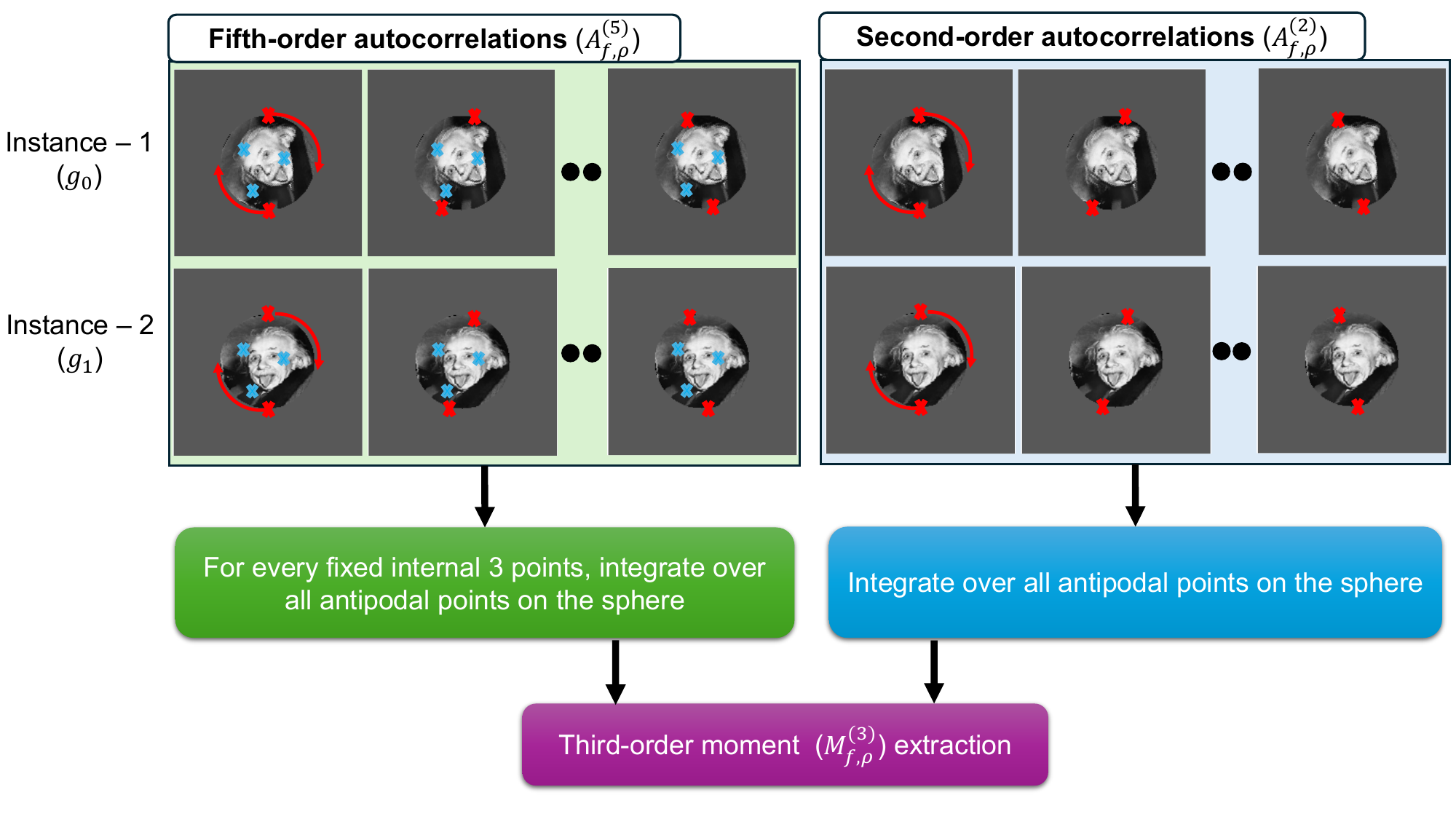}
    \caption{\textbf{Extracting $\mathsf{SO}(2)$ moments from $\mathsf{SE}(2)$ autocorrelations.}
    \textbf{Left panel:} Fifth-order autocorrelation with an antipodal pair. Fix $d=3$ internal points $\bm{\eta}_1,\bm{\eta}_2,\bm{\eta}_3$ inside the support of $f$ (Theorem~\ref{thm:reductionFromAutocorrelationToTensorMoment}). Each row shows a different rigid-motion instance $g_\ell\in\mathsf{SE}(2)$ in which the entire image is rotated and translated. For each transformed image, integrate the product of the intensities at these three fixed points together with the intensities at all pairs of (nearly) antipodal locations on a circle (Definition~\eqref{eqn:nearlyAntipodalPoints}), sweeping over all angles. This yields a directional slice of the $(d+2)$-order autocorrelation $A^{(d+2)}_{f,\rho}$ (the numerator of~\eqref{eqn:mainTheoremExtraction}).
    \textbf{Right panel:} Second-order autocorrelation of antipodal points. For each transformed image, integrate over all antipodal pairs on the boundary, as in Assumption~\ref{assum:nonVanishingSupport} (the denominator of~\eqref{eqn:mainTheoremExtraction}).
    \textbf{Extraction.} The $d$-th order $\mathsf{SO}(2)$ moment $M^{(d)}_{f,\rho}$ is obtained by letting the antipodal pair approach the boundary and normalizing by the second–order autocorrelation, per~\eqref{eqn:mainTheoremExtraction}.}
        \label{fig:4}
\end{figure*}

\section{Sample complexity of the orbit recovery problem under rigid motion} \label{sec:sampleComplexityRigidMotion}

In this section, we harness the algebraic extraction from $\mathsf{SE}(n)$ autocorrelations to $\mathsf{SO}(n)$ moments, developed in the previous section, to derive sample-complexity bounds for the orbit recovery problem under $\mathsf{SE}(n)$ (Problem~\ref{prob:orbitRecoverySEn}). We first set notation and recall the relevant definitions, and then state our main theorem on sample complexity.

\subsection{Performance metric and sample complexity}
A primary goal of this work is to characterize the sample complexity, i.e., the minimum number of observations required to estimate the signal $f$ within a specified mean-squared error (MSE). Since $f$ is only identifiable up to its orbit under the group action, we adopt the standard definition of MSE from~\cite{abbe2018multireference}:
\begin{align}
    \text{MSE}(\hat{f}, f) = \frac{1}{\|f\|_F^2} \mathbb{E} \left[ \min_{g \in G} \| g \cdot \hat{f} - f \|_F^2 \right]. \label{eqn:mseDef}
\end{align}
Here, $\hat{f} = \hat{f}(\{y_i\}_{i=0}^{N-1})$ is any estimator of $f$ based on $N$ noisy observations.
We now formalize the notion of sample complexity.

\begin{definition}[Sample complexity] \label{def:sampleComplexity}
    Let all parameters in the rigid motion orbit recovery problem (Problem~\ref{prob:orbitRecoverySEn}) be fixed, except for the number of observations $N$ and the noise variance $\sigma^2$. Define the MSE as
    \begin{align}
        \mathrm{MSE}^{\ast}_{\mathsf{SE}(n)}(\sigma^2, N) \triangleq  \inf_{\hat{f}} \mathbb{E} \left[ \mathrm{MSE} \left( \hat{f}(\{y_i\}_{i=0}^{N-1}), f \right) \right], \label{eqn:infMtdMSE}
    \end{align}
    and define the sample complexity as
    \begin{align}
        N^{\ast}_{\mathsf{SE}(n)}(\sigma^2, \epsilon) \triangleq  \min \left\{ N : \mathrm{MSE}^{\ast}_{\mathsf{SE}(n)}(\sigma^2, N) \leq \epsilon \right\}. \label{eqn:sampleComplexityMTD}
    \end{align}
\end{definition}

In this definition, $\epsilon > 0$ denotes a target error level, which we consider only in the asymptotic regime $\epsilon \to 0$. %\tamir{Since we use Joao's result, we need to be consistent with his definition that takes the error to zero}. 
Analogously, we define the sample complexity for the orbit recovery problem under the special orthogonal group $\mathsf{SO}(n)$, denoted by $N^{\ast}_{\mathsf{SO}(n)}(\sigma^2, \epsilon)$.

\subsection{Statistical analysis of SE(n) autocorrelations} \label{sec:autoCorrealtionAnalysisRigidMotion}
Here we introduce the autocorrelation analysis framework for the rigid motion problem and present several key results that will be instrumental in deriving upper bounds on the sample complexity. 

\begin{definition}[Empirical autocorrelations under the rigid motions group] \label{def:autoCorrelationMomentsRigidMotion}
Recall that $f \in C^0(\mathcal{B}_R^{(n)})$ is the unknown signal, and let $y_i : \mathcal{D} \to \mathbb{R}$ for $i \in \{0, \dots, N-1\}$ denote $N$ observations generated according to the model in Problem~\ref{prob:orbitRecoverySEn}. 
We define the padded signal $\tilde{y}_i : \bar{\mathcal{D}} \to \mathbb{R}$, where $\bar{\mathcal{D}} \supset \mathcal{D}$ is an extended domain such that, for every $\bm{x} \in \mathcal{D}$ and every shift $\bm{\tau} \in \mathcal{B}_R^{(n)}$, the shifted point $\bm{x} + \bm{\tau} \in \bar{\mathcal{D}}$, that is,
\begin{align}
    \bar{\mathcal{D}}
     = 
    \{\, \bm{x}+\bm{\tau} \mid \bm{x}\in\mathcal{D},\ \bm{\tau}\in\mathcal{B}_R^{(n)} \,\}.
\end{align}
The padded signal $\tilde{y}_i$ is defined as:
\begin{align}
    \tilde{y}_i(\bm{x}) = 
        \begin{cases}
        g_i \cdot f(\bm{x}) + \bar{\xi}_i(\bm{x}), & \bm{x} \in \mathcal{D}, \\
        \bar{\xi}_i(\bm{x}), & \bm{x} \in \bar{\mathcal{D}} \setminus \mathcal{D},
        \end{cases}
\end{align}
where $g_i = (\mathcal{R}_i, \bm{t}_i) \in \mathsf{SE}(n)$ and $\bar{\xi}_i : \bar{\mathcal{D}} \to \mathbb{R}$ is a realization of a white noise process with the same statistical properties as defined in Problem~\ref{prob:orbitRecoverySEn}. 
Then, the empirical $d$-th order autocorrelation of the observations $\{y_i\}_{i=0}^{N-1}$ is defined as:
\begin{align}
a_y^{(d)} (\bm{\tau}_0, \bm{\tau}_1, \dots, \bm{\tau}_{d-1}) 
    = \frac{1}{N} \sum_{i=0}^{N-1} \left( \int_{\bm{x} \in \mathcal{D}} \tilde{y}_i(\bm{x} + \bm{\tau}_0) \tilde{y}_i(\bm{x} + \bm{\tau}_1) \cdots \tilde{y}_i(\bm{x} + \bm{\tau}_{d-1}) \, d\bm{x} \right),
\label{eqn:autoCorrealtionMomentsRigidMotion}
\end{align}
for all $\bm{\tau}_0, \bm{\tau}_1, \dots, \bm{\tau}_{d-1} \in \mathcal{B}_R^{(n)}$, the radius-$R$ ball in $\mathbb{R}^n$.
\end{definition}

\begin{remark}
    Padding the observations $y_i$ to the extended domain $\bar{\mathcal{D}}$ with white noise ensures that all shifted points $\bm{x} + \bm{\tau}$ near the boundary of $\mathcal{D}$ are well-defined. This padding preserves statistical consistency and maintains translation invariance in the analysis.
\end{remark}

In essence, the empirical autocorrelation is computed by first evaluating the $d$-th order autocorrelation for each individual observation $y_i$, and then averaging the result over all $N$ observations. The stochastic integral in \eqref{eqn:autoCorrealtionMomentsRigidMotion} is defined formally in Appendix \ref{sec:preliminariesToStatisticalPart}.

The following proposition shows that, as $N \to \infty$, the empirical autocorrelation converges almost surely to a deterministic limit consisting of the population autocorrelation and a correction term. This correction term is composed of lower-order population autocorrelations and noise contributions, as described in \eqref{eqn:noiseCorrectionFunctionMain}. The proof is provided in Appendix~\ref{sec:proofOfpropRigidMotionEquivalence}.

\begin{proposition} \label{thm:propRigidMotionEquivalence}
Let $y_i : \mathcal{D} \to \mathbb{R}$, for $i \in \{0, \dots, N-1\}$, be i.i.d. observations from the model described in Problem~\ref{prob:orbitRecoverySEn}. Let $a^{(d)}_y$ denote the empirical $d$-th order autocorrelation of the observations, as in~\eqref{eqn:autoCorrealtionMomentsRigidMotion}, and let $A^{(d)}_{f, \rho}$ denote the population $d$-th order autocorrelation, as defined in Definition~\ref{def:autoCorrelationNoiseFree}. Then, for every  $\bm{\tau}_0, \bm{\tau}_1, \dots, \bm{\tau}_{d-1} \in \mathcal{B}_R^{(n)}$, as $N \to \infty$, 
\begin{align}
    \lim_{N \to \infty} a^{(d)}_y(\bm{\tau}_0, \bm{\tau}_1, \dots, \bm{\tau}_{d-1})
    \ \equalityAS \ 
    A^{(d)}_{f, \rho}(\bm{\tau}_0, \bm{\tau}_1, \dots, \bm{\tau}_{d-1}) 
    + P^{(d)}_{f, \rho}(\bm{\tau}_0, \dots, \bm{\tau}_{d-1}), \label{eqn:propRigidMotionEquivalence}
\end{align}
where the correction term $P^{(d)}_{f, \rho}$ accounts for noise contributions, is determined by lower-order population autocorrelations and the noise level $\sigma$, and is given explicitly by
\begin{align}
    P^{(d)}_{f, \rho}(\bm{\tau}_0, \dots, \bm{\tau}_{d-1}) 
    \triangleq \sum_{\substack{S \subsetneq \{0, \dots, d-1\} \\ |S^c| \text{ even}}} 
    \left( \sigma^{|S^c|} \cdot A^{(|S|)}_{f, \rho}(\{\bm{\tau}_j\}_{j \in S}) 
    \cdot \sum_{\text{pairings of } S^c} \prod_{(i,j)} \delta(\bm{\tau}_i - \bm{\tau}_j) \right). \label{eqn:noiseCorrectionFunctionMain}
\end{align}
Here, $S \subsetneq \{0, \dots, d-1\}$ is a proper subset whose complement $S^c$ has even cardinality, $A^{(|S|)}_{f, \rho}$ denotes the lower-order population autocorrelation corresponding to the indices in $S$, and the inner sum is taken over all pairings (i.e., partitions into unordered pairs) of the indices in $S^c$. The product over delta functions $\prod_{(i,j)} \delta(\bm{\tau}_i - \bm{\tau}_j)$ enforces that noise terms only contribute when their arguments coincide. 

\end{proposition}

\subsection{Sample complexity bounds for SE(n)}
In the method of moments, it is well known that, in the high-noise regime, estimating $d$-th–order moments requires $N = \omega(\sigma^{2d})$ samples. The same scaling carries over to autocorrelations: as we show next, if the
population autocorrelations up to order $\bar d$ uniquely determine the orbit of the signal, then the sample complexity of recovering the orbit is upper bounded by $\omega(\sigma^{2\bar d})$,
as proved in Appendix~\ref{sec:proofOfsampleComplexityRigidMotion}.

\begin{proposition}
\label{thm:sampleComplexityRigidMotion}
Assume the conditions of Proposition \ref{thm:propRigidMotionEquivalence} hold. In addition, assume the parameter space $\Theta$ of the unknown signals $f \in \Theta$ is compact. Then, if the population autocorrelations $\{A_{f, \rho}^{(d)}\}_{d=0}^{\bar{d}}$ up to order $\bar{d}$, as defined in  Definition~\ref{def:autoCorrelationNoiseFree}, uniquely determines the $\mathsf{SO}(n)$-orbit of the signal $f$%under the group action of $\mathsf{SE}(n)$
, then the sample complexity of recovering the $\mathsf{SO}(n)$-orbit is upper bounded by $\omega (\sigma^{2 \bar{d}})$.
\end{proposition}

A direct consequence of Theorem \ref{thm:reductionFromAutocorrelationToTensorMoment}, and Proposition \ref{thm:sampleComplexityRigidMotion} is Theorem \ref{thm:mainTheoremRigidmotion}, which relates to the sample complexity of the rigid motion problem (Problem \ref{prob:orbitRecoverySEn}). 

\begin{thm}[Sample complexity of orbit recovery under rigid motion] 
\label{thm:mainTheoremRigidmotion}
Consider the orbit recovery problem under $\mathsf{SO}(n)$ (Problem~\ref{prob:orbitRecoverySOn}) with rotation distribution $\rho$ supported on $G=\mathsf{SO}(n)$. 
Let $d$ denote the minimal order such that the $d$-th order moment 
$M_{f, \rho}^{(d)}$ (see~\eqref{eqn:sphericalCordinatedTensorMoment}) uniquely determines the orbit of a signal $f$ under the group action of $\mathsf{SO}(n)$.

Now, consider the orbit recovery problem under $\mathsf{SE}(n)$ (Problem~\ref{prob:orbitRecoverySEn}), where the marginal distribution on the rotational component is $\rho$ (Assumption~\ref{assum:rotation_translation_indep}). 
Assume that $f$ satisfies Assumptions~\ref{assum:support} and \ref{assum:nonVanishingSupport}. 
Then, the $(d+2)$ order autocorrelation $A_{f, \rho}^{(d+2)}$ in the $\mathsf{SE}(n)$ model (as defined in~\eqref{eqn:autoCorrelationNoiseFreeProduct2}) uniquely determines the $\mathsf{SO}(n)$-orbit of $f$. 
Moreover, the sample complexity of orbit recovery in this setting satisfies
\begin{align}
    \omega\!\left(\sigma^{2d}\right) 
     \le N^{\ast}_{\mathsf{SE}(n)}(\sigma^2) 
     \le \omega\!\left(\sigma^{2d+4}\right),
    \qquad \text{as } \sigma,N \to \infty.
\end{align}
\end{thm}

In other words, the sample complexity for the rigid motion group is lower bounded by that of the special orthogonal group, and upper bounded by the sample complexity of the corresponding orbit recovery problem under $\mathsf{SO}(n)$ multiplied by a factor of $\sigma^4$. The upper bound is established via an explicit recovery procedure in the limit as $N \to \infty$, as detailed in Appendix~\ref{sec:proofOfmainTheoremRigidmotion}, and demonstrated numerically in Section~\ref{sec:empirical}.

%\subsection{Discussion on the tightness of the results}

\section{Sample complexity of the multi-target detection (MTD) problem} \label{sec:sampleComplexityMTD}
%\amnon{Maybe to change the title of the section to "Reduction of the MTD problem to the rigid motion orbit recovery problem.} \amnon{I prefer to keep the name symmetric with the previous section, and just to explain the reduction here in the text.}
In this section, we derive lower and upper bounds on the sample complexity required to estimate the orbit of the unknown signal $f$ in the MTD model (Problem~\ref{prob:orbitRecoveryMTD}). Our main interest is the low-SNR regime, where, as illustrated in Figure~\ref{fig:1}(d), it is infeasible to reliably estimate the locations of the individual occurrences.
Our approach is to relate the MTD problem to the orbit recovery problem under $\mathsf{SE}(n)$ studied in Section~\ref{sec:sampleComplexityRigidMotion}. In particular, we show in Proposition~\ref{thm:prop0} that the analysis of the MTD model can be reduced to an instance of orbit recovery under $\mathsf{SE}(n)$, which enables us to transfer sample-complexity bounds between the two settings.

%\amnon{To add here a connection of this section to the results obtained in the previous Section, and to refer to the proposition 5.6 that connects to the MRA-SE(n) problem.}

\subsection{Continuous MTD model}
We begin by presenting the continuous MTD model together with the assumptions that underlie its formulation. For clarity, we restate the model from Problem~\ref{prob:orbitRecoveryMTD} in a form that emphasizes the separation conditions imposed on the signal instances. 

\paragraph{Model description.}
Let $f \in C^0 (\mathcal{B}_R^{(n)})$ be an unknown continuous signal supported on the $n$-dimensional ball of radius $R$. The observed signal is a function $y: [-MR, MR]^n \to \mathbb{R}$ given by
\begin{align}
    y = \sum_{i=0}^{N-1} s_i \ast f_i + \xi, \label{eqn:MTDmodelGeneral}
\end{align}
where $f_i = \mathcal{R}_i \cdot f$ for $\mathcal{R}_i \in \mathsf{SO}(n)$, the symbol $\ast$ denotes a linear convolution, $N$ is the total number of signal occurrences in the observation $y$, and $M \in \mathbb{N}^{+}$ is a positive integer.
%denotes the ratio between the radius $R$ of the signal $f$ and the dimensions of the MTD observation $y$ (assumed to be an integer). 
The noise term $\xi$ is a Gaussian white noise process on the domain $[-MR, MR]^n$, i.e., a zero-mean Gaussian process with covariance defined in~\eqref{eqn:covarianceMatrixMTD}. The group elements $\{\mathcal{R}_i\}_{i \in \{0, \ldots, N-1\}}$ are i.i.d. samples drawn from a distribution $\rho$, and are independent of the noise $\xi$.
The positions of the embedded signals are encoded by the location indicator functions $\{s_i\}_{i=0}^{N-1}$, which are also unknown. The central goal in the MTD model is to recover the $\mathsf{SO}(n)$ orbit of $f$ from the noisy observation $y$. The assumptions on the signal $f$ are the same as those stated in Section~\ref{sec:assumptionsOnSignal}.

We define the signal density by
\begin{align}
    \gamma \triangleq \frac{N \cdot V_n(2R)}{(2MR)^n}, \label{eqn:densityMTDmodel}
\end{align}
where $V_n(2R) = \frac{\pi^{n/2}(2R)^n}{\Gamma(n/2+1)}$ denotes the volume of an $n$-dimensional ball of radius $2R$. We assume that the noise level $\sigma$ and the density parameter $\gamma$ are known, noting that prior work has shown that $\gamma$ can be reliably estimated from the data~\cite{bendory2019multi}. 
We consider the asymptotic regime where $\sigma, N, M \to \infty$, while keeping the density $\gamma < 1$ fixed.

\begin{remark}
    In contrast to previously studied MTD models (see, e.g.,~\cite{bendory2019multi}), the formulation in~\eqref{eqn:MTDmodelGeneral} is posed in the continuous domain. 
    This choice reflects the fact that translations are naturally expressed in Cartesian coordinates, while group actions of $\mathsf{SO}(n)$ are defined in spherical coordinates. By working in the continuous setting, the model avoids discretization and sampling artifacts.
\end{remark}

\paragraph{Assumptions on the unknown locations.}
Each location indicator $s_i(\bm{x})$ is a Dirac delta function centered at a point $\bm{x}_i$, given by
\begin{align}
    s_i(\bm{x}) = \delta(\bm{x} - \bm{x}_i),
\end{align}
which specifies that the center of the signal instance $f_i$ is located at position $\bm{x}_i$ within the observation $y$. The positions $\{\bm{x}_i\}_{i=0}^{N-1}$ are assumed to be deterministic but unknown. %Under these definitions, the MTD observation $y$ also has compact support: $y : [-MR, MR]^n \to \mathbb{R}$.

Throughout, we assume that signal occurrences are non-overlapping. 
Formally, the centers $\bm{x}_i$ and $\bm{x}_j$ of any two signals $f_i$ and $f_j$ must be separated by a distance of at least $2R$, which ensures that their supports are disjoint and their contributions to the observation $y$ do not interfere. 
For the purpose of establishing the upper bound on sample complexity, we further impose a stronger condition, referred to as the \emph{well-separated} case: the minimal distance between any two centers is greater than $4R$. 
This guarantees a buffer zone of width at least $2R$ between the supports of all signal instances.
%\dan{I don't think we need this formal definition. Also it's not exactly the same as the one written in words since the former says at least while the formal definition say more than.} \amnon{I prefer to keep the definition, as we refer to it in the proofs.}
\begin{definition} [Separation] \label{def:wellSeperatedModel}
Consider the following separation condition:
\begin{align} \label{eq:sep}
  \nonumber \text{If} \ s_{i_1}\p{\bm{x}} = \delta\p{\bm{x} - \bm{x}_{i_1}} \ \text{and} \ s_{i_2}\p{\bm{x}} = \delta\p{\bm{x} - \bm{x}_{i_2}} \ \text{for} \ i_1 \neq i_2, \\
   \text{then} \ \norm{\bm{x}_{i_1} - \bm{x}_{i_2}}_F \geq R_{\mathsf{sep}} + \epsilon,
\end{align}
for $\epsilon > 0$.
If the MTD model satisfies~\eqref{eq:sep} with $R_{\mathsf{sep}}\geq 2R$, we say it is the non-overlapping; if $R_{\mathsf{sep}}\geq 4R$, we say it is well-separated. 
\end{definition}

Although the well-separation assumption is essential for our analysis, previous work has demonstrated that the effect of adjacent signals that fail to meet this condition, but still obey the non-overlapping constraint, remains limited, albeit not negligible~\cite{kreymer2022two}.

\subsection{Statistical analysis of the autocorrelations in MTD} \label{sec:autoCorrealtionAnalysis}
In analogy to the autocorrelation analysis developed for the rigid motion problem in Section~\ref{sec:autoCorrealtionAnalysisRigidMotion}, this section presents the autocorrelation framework for the MTD problem. We derive several key results that are instrumental in establishing upper bounds on the sample complexity.

To address the challenge of localizing signal instances in high-noise settings~\cite{aguerrebere2016fundamental,dadon2024detection}, previous works have proposed using the autocorrelations, which are invariant to both unknown signal locations and unknown group actions~\cite{bendory2019multi,bendory2023toward,kreymer2022two,lan2020multi,bendory2023multi}, as we define next.

\begin{definition} [Empirical autocorrelations of the MTD model] \label{def:autoCorrelationMomentsEmpirical}
    Let $y \in [-MR, MR]^n \to \mathbb{R}$ be an MTD observation following the MTD model in \eqref{eqn:MTDmodelGeneral}. The empirical $d$-th order autocorrelation of $y$ is defined by: 
\begin{align}
    \nonumber {a}^{(d)}_{y} & (\bm{\tau}_0, \bm{\tau}_1, \bm{\tau}_2, ..., \bm{\tau}_{d-1}) 
    \\ & = \frac{1}{(2MR)^n} \int_{\bm{x} \in [-MR, MR]^n} { \tilde{y}(\bm{x} + \bm{\tau}_0)\tilde{y}(\bm{x}+\bm{\tau}_1) ... \tilde{y}(\bm{x}+\bm{\tau}_{d-1}) \ d \bm{x} }, \label{eqn:autoCorrealtionMoments}
\end{align}
where $\tilde{y}$ is the padding with zeros of $y$ to size of $\pp{-(M+1)R, \p{M+1}R}^n$, for $\bm{\tau}_0, \bm{\tau}_1, \bm{\tau}_2, ..., \bm{\tau}_{d-1} \in \mathcal{B}_R^{(n)}$. 
\end{definition}

Definition~\ref{def:autoCorrelationMomentsEmpirical} describes the empirical autocorrelation computed over the full MTD observation defined in~\eqref{eqn:MTDmodelGeneral}. Next, we introduce the autocorrelation ensemble mean, which represents the expected autocorrelation assuming each signal instance in the MTD model is spatially well-separated (Definition~\ref{def:wellSeperatedModel}); see Figure~\ref{fig:5} for an illustration.

\begin{definition} [Autocorrelation ensemble mean] \label{def:autoCorrelationMomentsEnsembele}
    Let $f_{\mathcal{R}} = \mathcal{R} \cdot f$, where $\mathcal{R} \in \mathsf{SO}(n)$ is a random group element drawn from the distribution $\mathcal{R} \sim \rho$.
    Define the random field $Y_{\mathcal{R}}: \mathcal{B}_{3R}^{(n)} \to \mathbb{R}$ by, 
    \begin{align}
        Y_{\mathcal{R}} (\bm{x}) = \tilde{f}_{\mathcal{R}} (\bm{x}) + \xi (\bm{x}), \label{eqn:ensembeleRV}
    \end{align}
    where $\tilde{f}_{\mathcal{R}}: \mathcal{B}_{3R}^{(n)} \to \mathbb{R}$ is given by,
    \begin{align}
        \tilde{f}_{\mathcal{R}}(\bm{x}) = f_{\mathcal{R}}(\bm{x}) \mathbf{1}_{\|\bm{x}\| \leq R} + 0 \cdot \mathbf{1}_{R < \bm{\|x\|} \leq 3R}, 
    \end{align}
    that is, the padding of $f_{\mathcal{R}}$ with zeros to an $n$-dimensional ball with radius $3R$, where the signal $f_{\mathcal{R}}$ is located at the center of the ball. The term $\xi$ is a white Gaussian noise with a variance $\sigma^2$. 
    Then, the $d$-order autocorrelation ensemble mean of $Y_{\mathcal{R}}$ is defined by,
    \begin{align}
        \nonumber \bar{a}^{(d)}_{Y,\rho} & \p{\bm{\tau}_0, \bm{\tau}_1, \bm{\tau}_2, ..., \bm{\tau}_{d-1}} 
        \\ & = \frac{1}{V_n\p{2R}} \int_{\bm{x} \in \mathcal{B}_{2R}^{(n)}} \mathbb{E}_{\mathcal{R} \sim \rho, \xi} \ppp{Y_{\mathcal{R}} (\bm{x} + \bm{\tau}_0, )Y_{\mathcal{R}} (\bm{x} + \bm{\tau}_1) ... Y_{\mathcal{R}}(\bm{x}+\bm{\tau}_{d-1})}  d \bm{x}, \label{eqn:autoCorrealtionMomentsEnsembele}
    \end{align}
    for $\bm{\tau}_0, \bm{\tau}_1, \bm{\tau}_2, ..., \bm{\tau}_{d-1} \in \mathcal{B}_R^{(n)}$, where $\mathbb{E}_{\mathcal{R} \sim \rho, \xi} \ppp{\cdot}$ is an expectation  over the group actions $\mathcal{R} \in \mathsf{SO}(n)$ following a distribution $\rho$, and on the additive noise $\xi$, and $V_n\p{2R}$ is the volume of an $n$-dimensional ball with a radius $2R$. 
\end{definition}

\begin{remark}
\label{remark:5.5} 
    The autocorrelation ensemble mean, $\bar{a}^{(d)}_{Y,\rho}$, defined in Definition~\ref{def:autoCorrelationMomentsEnsembele} and equation~\eqref{eqn:autoCorrealtionMomentsEnsembele} matches the right-hand side of equation~\eqref{eqn:propRigidMotionEquivalence}, i.e., the limiting form of the rigid motion autocorrelation as $N \to \infty$, when the integration domain is taken to be $\mathcal{D} = \mathcal{B}_{2R}^{(n)}$. 
\end{remark}
%\dan{Is this a remark or a corollary?}

The following proposition establishes that, in the well-separated regime, the empirical autocorrelations converge almost surely to their ensemble means. This result, illustrated in Figure~\ref{fig:5}, extends the findings of~\cite{balanov2025note}. The proof is provided in Appendix~\ref{sec:proofOfProp0}.

\begin{proposition} \label{thm:prop0}
Let $y: [-MR, MR]^n \to \mathbb{R}$ be an MTD observation following the model in~\eqref{eqn:MTDmodelGeneral}. 
Assume the model is well-separated, 
%as defined in Definition~\ref{def:wellSeperatedModel}, 
and that the group elements are drawn according to a distribution $\rho$.
Let $a^{(d)}_y$ denote the empirical $d$-th order autocorrelation as defined in~\eqref{eqn:autoCorrealtionMoments}, and let $\bar{a}^{(d)}_{Y,\rho}$ denote the $d$-th order autocorrelation ensemble mean as defined in~\eqref{eqn:autoCorrealtionMomentsEnsembele}. Then, for every $\bm{\tau}_0, \bm{\tau}_1, \dots, \bm{\tau}_{d-1} \in \mathcal{B}_R^{(n)}$, we have
\begin{align}
    \nonumber \lim_{N, M \to \infty} a^{(d)}_y (\bm{\tau}_0, \bm{\tau}_1, \dots, \bm{\tau}_{d-1})
    \equalityAS \ & \gamma \cdot \bar{a}^{(d)}_{Y,\rho} (\bm{\tau}_0, \bm{\tau}_1, \dots, \bm{\tau}_{d-1}) \\
    & + (1 - \gamma) \cdot \chi (\bm{\tau}_0, \bm{\tau}_1, \dots, \bm{\tau}_{d-1}), \label{eqn:asymptoticAutoCorrelationEquivalence}
\end{align}
where the noise term is given by
\[
    \chi(\bm{\tau}_0, \bm{\tau}_1, \dots, \bm{\tau}_{d-1}) = \mathbb{E} \left[ \xi(\bm{\tau}_0) \xi(\bm{\tau}_1) \cdots \xi(\bm{\tau}_{d-1}) \right].
\]
\end{proposition}

\begin{figure*}[!t]
    \centering
    \includegraphics[width=1.0 \linewidth]{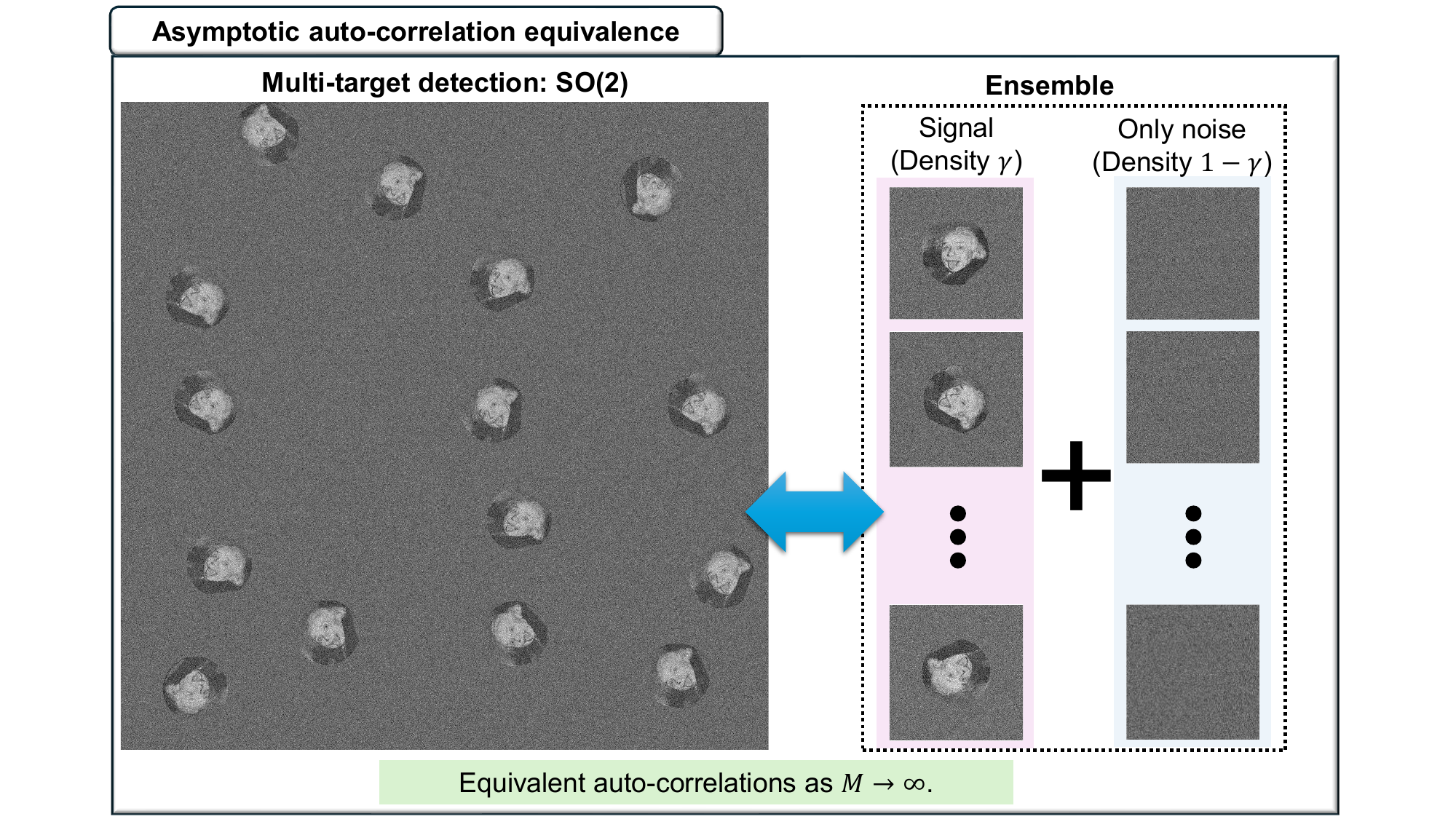}
    \caption{\textbf{Asymptotic equivalence between empirical and ensemble mean autocorrelations.} This figure illustrates the convergence of the empirical autocorrelation (Definition~\ref{eqn:autoCorrealtionMoments}) to the ensemble mean autocorrelation (Definition~\ref{eqn:ensembeleRV}) as the number of samples $N \to \infty$, as formalized in Proposition~\ref{thm:prop0}. In the ensemble model, each observation is drawn either from the signal distribution with probability $\gamma$ or from pure noise with probability $1 - \gamma$. The empirical autocorrelation thus asymptotically matches the expected autocorrelation under this probabilistic mixture.}
    \label{fig:5}
\end{figure*}

\subsection{Sample complexity bounds for MTD}
We now state the corresponding sample complexity theorem for the MTD model, analogous to Theorem~\ref{thm:mainTheoremRigidmotion}, which addressed the rigid motion model.

\begin{thm}[Sample complexity of MTD orbit recovery] 
\label{thm:mainTheorem}
Consider the orbit recovery problem under $\mathsf{SO}(n)$ (Problem~\ref{prob:orbitRecoverySOn}) with rotation distribution $\rho$ supported on $G=\mathsf{SO}(n)$. 
Let $d$ denote the minimal order such that the $d$-th order moment 
$M_{f, \rho}^{(d)}$ uniquely determines the orbit of a signal $f$ under the group action of $\mathsf{SO}(n)$.

Now, consider the MTD model~\eqref{eqn:MTDmodelGeneral} with the same distribution $\rho$, consisting of $N$ occurrences of the signal $f$ located at unknown positions $\{s_i\}_{i=0}^{N-1}$. 
Assume the well-separated case (Definition~\ref{def:wellSeperatedModel}) and that $f$ satisfies Assumptions~\ref{assum:support} and~\ref{assum:nonVanishingSupport}. 
Then the sample complexity of the $\mathsf{SO}(n)$ orbit recovery in the MTD model satisfies
\begin{align}
    \omega\!\left(\sigma^{2d}\right) 
     \le  N^{\ast}_{\mathrm{MTD}}(\sigma^2) \le \omega\!\left(\sigma^{2d+4}\right),
    \qquad \text{as } \sigma,N,M \to \infty.
\end{align}
\end{thm}

Theorem~\ref{thm:mainTheorem} is established via an explicit algorithmic reduction, described in Appendix~\ref{sec:proofSampleComplexityMTD}, which reduces the MTD empirical autocorrelations into the $\mathsf{SO}(n)$ moments. The proof unfolds in three main steps. First, Proposition~\ref{thm:prop0} shows that as the number of signals and observation length $N, M \to \infty$, the empirical autocorrelations converge to their ensemble means. Next, by Remark~\ref{remark:5.5} and equation~\eqref{eqn:propRigidMotionEquivalence}, these ensemble-mean autocorrelations yield the population autocorrelations of the underlying signal. Finally, Theorem~\ref{thm:reductionFromAutocorrelationToTensorMoment} demonstrates that the $d$-th order $\mathsf{SO}(n)$ moment can be recovered from the second- and $(d+2)$-order population autocorrelations. Since these suffice for orbit recovery, Proposition~\ref{thm:sampleComplexityRigidMotion} implies a sample complexity of $\omega(\sigma^{2d+4})$, yielding the upper bound for the MTD model.

\begin{remark}
The sample complexity bounds in Theorem~\ref{thm:mainTheorem} also depend on the signal density parameter~$\gamma$, defined in \eqref{eqn:densityMTDmodel}, which measures the average number of signal occurrences per MTD observation length. 
In particular, the lower bound scales as $\omega(\sigma^{2d} \cdot  \gamma^{-1})$ and the upper bound as $\omega(\sigma^{2d+4} \cdot\gamma^{-1})$. 
This dependence reflects the fact that smaller values of $\gamma$ correspond to sparser signal placements, which effectively decrease the SNR and therefore require larger observation sizes for accurate recovery. 
For clarity of presentation, Theorem~\ref{thm:mainTheorem} states the bounds only in terms of the noise level, with the factor $\gamma$ absorbed into the asymptotic notation. 
\end{remark}

\section{Algorithmic proof-of-concept and empirical validation} \label{sec:empirical}

This section provides empirical validation of the algebraic extraction established in Theorem~\ref{thm:reductionFromAutocorrelationToTensorMoment} and of the sample-complexity scaling predicted by Theorem~\ref{thm:mainTheoremRigidmotion}. In all experiments, we first compute the rigid-motion autocorrelations $A_{f, \rho}^{(d)}$ (Definition~\ref{def:autoCorrelationNoiseFree}), then extract the rotational moments via the boundary-limit expression~\eqref{eqn:mainTheoremExtraction}, and finally use these moments for signal reconstruction. To our knowledge, this constitutes the \emph{first complete implementation} of a three-dimensional reconstruction pipeline that directly parallels the cryo-electron tomography setting. Section~\ref{sec:applications} clarifies how this pipeline relates to the MTD formulation and discusses implications for cryo-ET and cryo-EM.

\paragraph{Scope of the validation.}
The validation framework in this section has two stages: 
(i) extraction of $\mathsf{SO}(n)$ moments from the $\mathsf{SE}(n)$ autocorrelations, and 
(ii) an inversion from $\mathsf{SO}(n)$ moments to the orbit of the underlying signal. 
Our emphasis is on the extraction side, where for the inversion, we adopt specific and convenient implementations, as we detail in the sequel.
We use the present inversion approaches purely for computational convenience. 
Accordingly, the reconstructions should be viewed as a \emph{proof of concept}: we do not claim optimality, noise robustness, or numerical stability of the inversion routines. 
The main contribution here is to establish and validate the extraction itself; once a stable inversion algorithm is employed, the overall pipeline naturally extends to noisy settings.

\subsection{Pipeline overview and implementation details}
Although the theoretical analysis is continuous, numerical evaluation requires discretization on finite grids. This discretization introduces finite-resolution errors that can affect both the estimated moments and the reconstructed signals. Our experiments, therefore, quantify how sampling resolution influences the accuracy of extracted rotational moments and the quality of the resulting reconstructions. In addition, we examine the impact of additive noise on the stability and accuracy of the recovery process. Specifically, we focus on two primary sources of deviation from the continuum theory:  
(i) discretization errors arising from finite radial and angular resolution, including the dependence on the boundary parameter $\delta > 0$ in~\eqref{eqn:nearlyAntipodalPoints}; and  
(ii) noise and finite-sample effects introduced by the empirical estimators described in Section~\ref{sec:autoCorrealtionAnalysisRigidMotion}.

Simulations are performed in both 2D and 3D settings, with the rotational component of $\mathsf{SE}(n)$ sampled from the uniform Haar measure. Consequently, for notation brevity, we omit the dependence on $\rho$ in the notation and denote the corresponding autocorrelations and moments simply by $A_f^{(d)}$ and $M_f^{(d)}$, respectively.
Under Haar-uniform rotations, the extracted moments admit a natural representation in terms of relative angular coordinates, $M_f^{(2)}(r_1, r_2, \Delta\varphi)$ and $M_f^{(3)}(r_1, r_2, r_3, \Delta\varphi_1, \Delta\varphi_2)$, which depend only on angular differences $(\Delta\varphi_1, \Delta\varphi_2)$. Detailed formulations and implementation procedures for both the 2D and 3D cases are provided in Appendices~\ref{app:experimentalMethods2D}–\ref{app:experimentalMethods3D}, and the code implementation is available at 
\begin{center}
\href{https://github.com/AmnonBa/rigid-motion-orbit-recovery}{https://github.com/AmnonBa/rigid-motion-orbit-recovery}.    
\end{center}

\paragraph{Notation.}
For $d \in \{2,3\}$, let $M_{f, \text{true}}^{(d)}$ denote the ground-truth $\SO(n)$ moments computed directly from the underlying volume $f$ (Definition~\ref{def:mraTensorMoment}). The extraction~\eqref{eqn:mainTheoremExtraction} yields discretized moments $M_{f, \text{ext}}^{(d)}(\delta; R, N_\varphi)$,
where $R$ is the number of radial rings (in 2D) or shells (in 3D), $N_\varphi$ is the angular sampling density, and $\delta$ is the boundary parameter.  
In the presence of $N$ noisy observations, we denote by $\widehat{M}_{f, \text{ext}}^{(d)}$ the empirical extracted moments obtained by substituting the empirical autocorrelations $a_y^{(d)}$ (Definition~\ref{def:autoCorrelationMomentsRigidMotion}) into~\eqref{eqn:mainTheoremExtraction}.  
Moment-level errors are reported as normalized MSE,
\begin{align}
    \mathrm{MSE}(M_1, M_2)
    = \frac{\|M_1 - M_2\|_F^2}{\|M_1\|_F^2}, \label{eqn:mse-of-moments}
\end{align}
where $\|\cdot\|_F$ denotes the Frobenius norm evaluated over the discretized grids used in the experiments.

\subsection{Discretization study: boundary-limit and resolution effects}
Figure~\ref{fig:6} illustrates how discretization affects the accuracy of the extraction~\eqref{eqn:mainTheoremExtraction} in the 2D setting (without noise). Two main sources of error are analyzed: the finite boundary parameter $\delta>0$ and the limited resolution of the radial and angular sampling grids.  

Figure~\ref{fig:6}(a) shows the effect of the boundary parameter $\delta$. For fixed $(R, N_\varphi)$, decreasing $\delta$ leads to a monotonic decrease in $\mathrm{MSE}\left(M_{f, \text{true}}^{(d)}, M_{f, \text{ext}}^{(d)}\right)$, consistent with the limit in~\eqref{eqn:mainTheoremExtraction}. Empirically, this decay scales proportionally to $\delta$ as $\delta \to 0$.
Figure~\ref{fig:6}(b) examines the influence of radial resolution. With $\delta$ fixed, increasing the number of rings $R$ reduces radial discretization error and improves the numerical approximation of~\eqref{eqn:mainTheoremExtraction}.  
Figure~\ref{fig:6}(c)-(d) compare $M_{f, \text{true}}^{(d)}$ and $M_{f, \text{ext}}^{(d)}$ for fixed $(\delta, R, N_\varphi)$. In both cases, the extracted moments follow the ground-truth moments closely but exhibit a mild low-pass behavior due to finite angular resolution. This effect is most noticeable at sharp angular transitions, where coarse sampling introduces smoothing artifacts. 

\begin{figure*}[!t]
    \centering
    \includegraphics[width=0.85\linewidth]{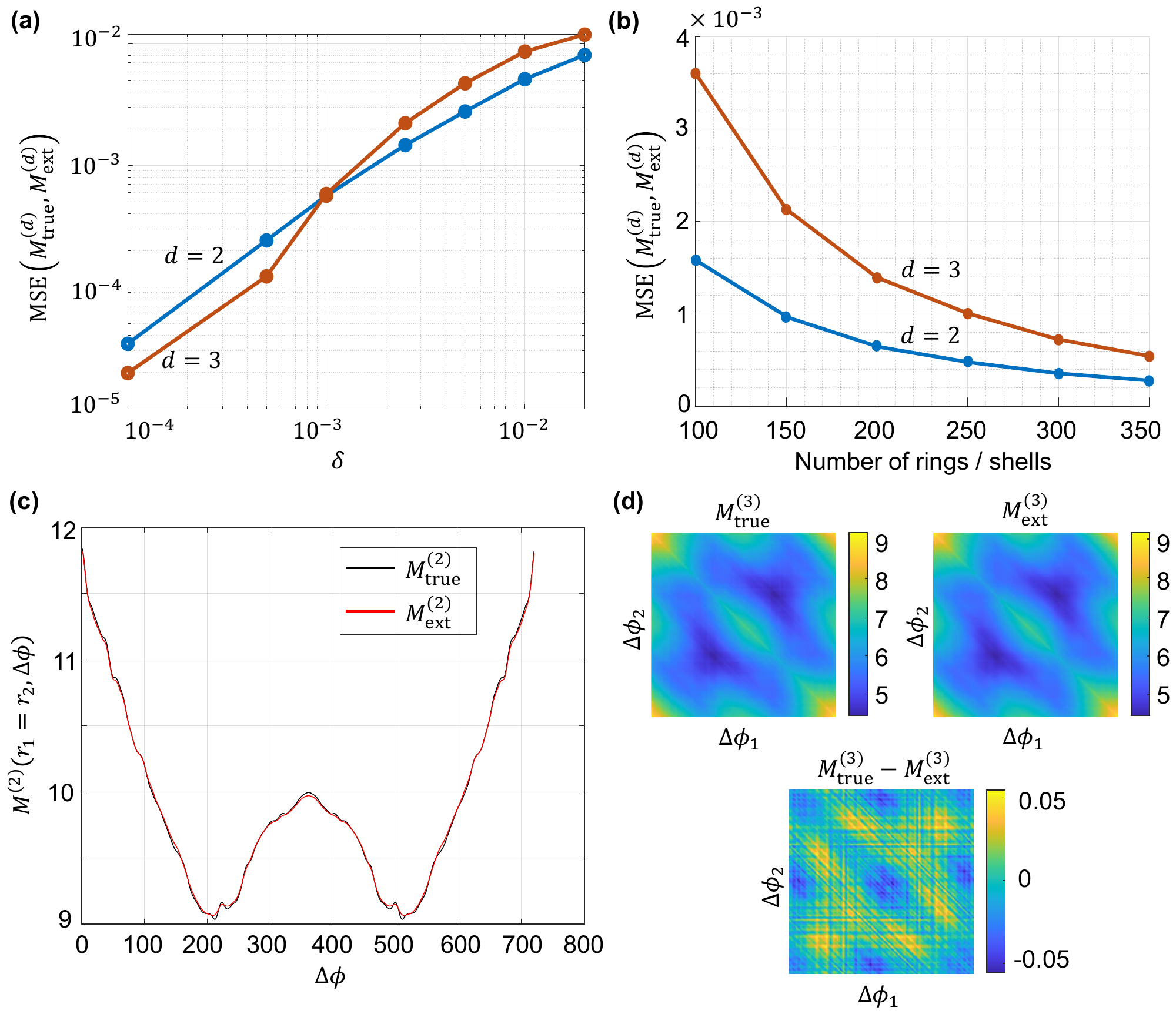}
    \caption{\textbf{Impact of discretization on extraction accuracy.}
   \textbf{(a)}~Effect of the boundary parameter $\delta$: normalized MSE between $M_{f, \text{true}}^{(d)}$ and $M_{f, \text{ext}}^{(d)}(\delta; R, N_\varphi)$ for $d\in\{2,3\}$. As $\delta\to 0$, the error decreases, consistent with Theorem~\ref{thm:reductionFromAutocorrelationToTensorMoment}. Parameters: $R=100$, $N_\varphi=720$.  
    \textbf{(b)}~Effect of radial resolution: normalized MSE versus the number of rings $R$ (with $\delta=0.002$, $N_\varphi=720$). Increasing $R$ reduces radial discretization error.  
    \textbf{(c)}~Second moment at $(r_1, r_2) = (99, 99)$ as a function of relative angle $\Delta\varphi$: $M_{f, \text{ext}}^{(2)}$ follows $M_{f, \text{true}}^{(2)}$ closely, with a mild low-pass effect.  
    \textbf{(d)}~Representative slice of the third moments $M_{f, \text{true}}^{(3)}$ and $M_{f, \text{ext}}^{(3)}$ at $(r_1, r_2, r_3) = (99, 99, 99)$ as a function of relative angles $(\Delta\varphi_1, \Delta\varphi_2)$. The same low-pass effect appears as in panel~(c), with sharper angular transitions producing larger errors. Parameters: $R=200$, $N_\varphi=720$.}
    \label{fig:6}
\end{figure*}

\subsection{Noise and finite-sample scaling}
We next empirically examine the statistical predictions of Section~\ref{sec:sampleComplexityRigidMotion}. These experiments quantify how additive noise and the number of observations $N$ affect the accuracy of the empirical extraction $\widehat{M}_{f, \text{ext}}^{(d)}$. Following the procedure in Section~\ref{sec:autoCorrealtionAnalysisRigidMotion}, we compute $a_y^{(d)}$ from $N$ i.i.d. noisy observations and substitute it into~\eqref{eqn:mainTheoremExtraction} to obtain empirical extracted moments. %Proposition~\ref{thm:propRigidMotionEquivalence} characterizes the asymptotic $N \to \infty$ limit and the leading noise-induced corrections, and the results below provide empirical support for these theoretical predictions.  

Figure~\ref{fig:7} plots the normalized MSE~\eqref{eqn:mse-of-moments} between the noise-free extracted moment $M_{f,\text{ext}}^{(d)}$ (computed without additive noise), and its noisy empirical estimate $\widehat{M}_{f,\text{ext}}^{(d)}$ for $d\in\{2,3\}$, as functions of SNR and the number of observations $N$.
Figure~\ref{fig:7}(a) shows the dependence on SNR. In the high-SNR regime ($\mathrm{SNR}\gg 1$), the errors decay approximately as $\mathrm{SNR}^{-1}$, whereas in the low-SNR regime ($\mathrm{SNR}\ll 1$), the slopes steepen to $\mathrm{SNR}^{-4}$ for $d=2$ and $\mathrm{SNR}^{-5}$ for $d=3$. These observations are consistent with the upper-bound scaling $\sigma^{2d+4}$ from Theorem~\ref{thm:mainTheoremRigidmotion}. We emphasize, however, that while this result validates the upper bound, the true sample complexity required for accurate signal recovery may be lower.  

Figure~\ref{fig:7}(b) examines the dependence on the number of observations $N$ at fixed SNR. For both $d=2$ and $d=3$, the MSE decreases proportionally to $N^{-1}$, in agreement with the law of large numbers expected for empirical estimators (see Appendix~\ref{sec:statistical-analysis-orbit-recoverry}). Taken together, these results show that the empirical behavior of the extraction follows the theoretical scaling prediction in Section~\ref{sec:sampleComplexityRigidMotion}.

\begin{figure*}[!t]
    \centering
    \includegraphics[width=0.9\linewidth]{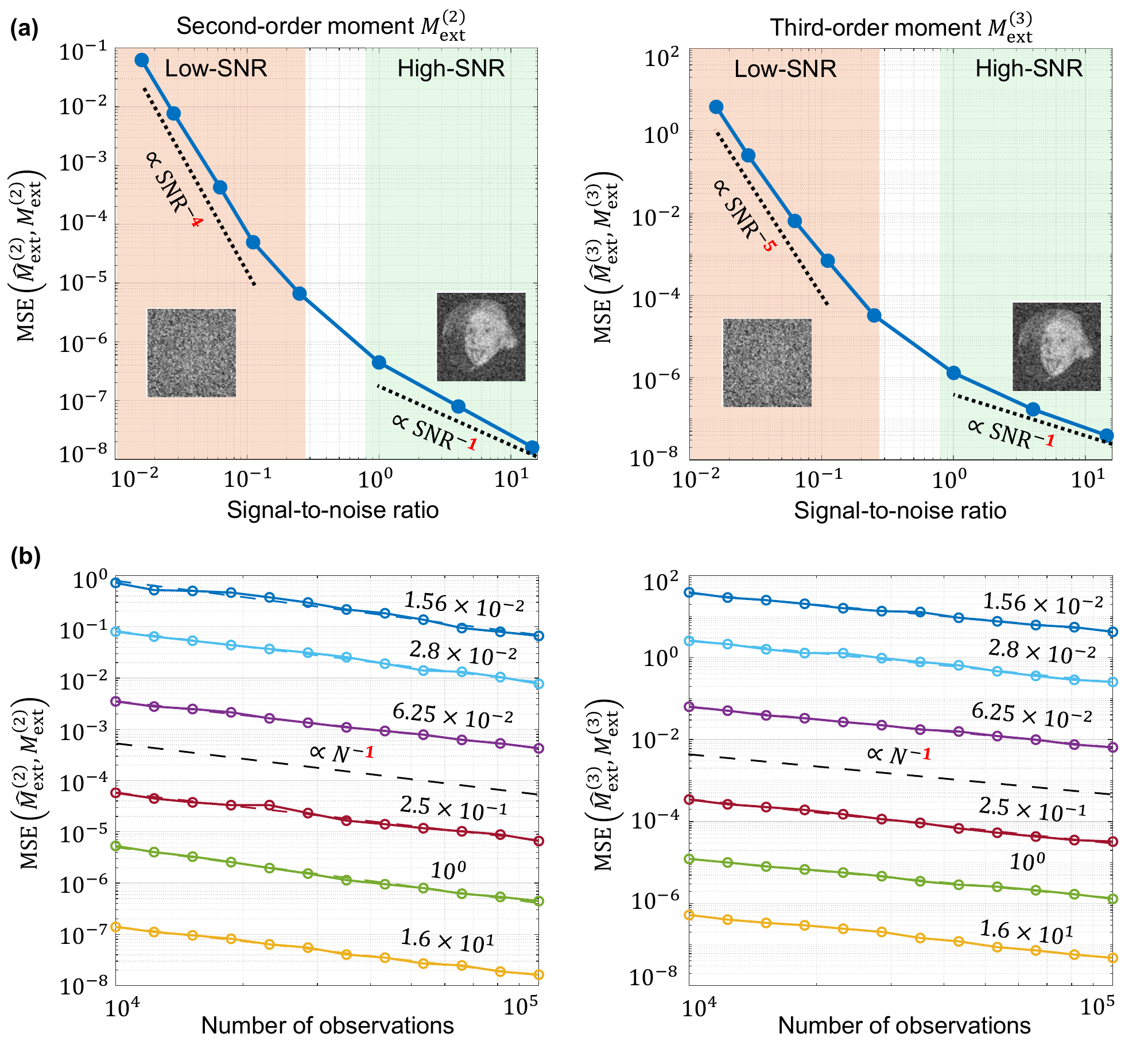}
    \caption{\textbf{Impact of noise and sample size on the empirical moment estimation.}
    \textbf{(a)}~Normalized MSE between $M_{f, \text{ext}}^{(d)}$ (extraction without additive noise) and $\widehat{M}_{f, \text{ext}}^{(d)}$ (empirical extraction from $N$ noisy observations) versus SNR, for $d\in\{2,3\}$. In the high-SNR regime, errors scale as $\mathrm{SNR}^{-1}$, while at low SNR, slopes steepen to $\mathrm{SNR}^{-4}$ ($d=2$) and $\mathrm{SNR}^{-5}$ ($d=3$), consistent with the $\sigma^{2d+4}$ upper-bound scaling predicted by the analysis in~\eqref{eqn:noiseCorrectionFunctionMain}. Parameters: $N=10^5$ observations; 200 Monte Carlo trials per point; Einstein image; $R=30$, $N_\varphi=180$.  
    \textbf{(b)}~Normalized MSE versus $N$ for several fixed SNR levels (left: $d=2$; right: $d=3$). The MSE decreases proportionally to $N^{-1}$, in agreement with the law of large numbers (Appendix~\ref{sec:statistical-analysis-orbit-recoverry}). Each curve corresponds to a different SNR level} %(labels above)}
    \label{fig:7}
\end{figure*}

\subsection{End-to-end reconstructions}
We conclude by demonstrating an end-to-end pipeline that validates the complete extraction–reconstruction framework. Starting from a discretized, noise-free input signal, we compute its rigid-motion autocorrelations (Definition~\ref{def:autoCorrelationNoiseFree}), extract the $\mathsf{SO}(n)$ moments using Theorem~\ref{thm:reductionFromAutocorrelationToTensorMoment}, and then reconstruct the signal from these extracted moments. All reconstructions rely on algebraic simplifications that substantially reduce the computational complexity (see Appendices~\ref{app:experimentalMethods2D}–\ref{app:experimentalMethods3D}).

In the 2D setting (Figure~\ref{fig:8}(a)), we process test images (e.g., Mona Lisa and Einstein portraits). We compute $A_f^{(2)}, A_f^{(4)}, A_f^{(5)}$, extract the $M_{f}^{(2)}$ and $M_{f}^{(3)}$ moments, and apply a two-stage inversion:  
(i) per-ring bispectrum inversion via a frequency-marching solver~\cite{kakarala2012bispectrum,bendory2017bispectrum}, followed by  
(ii) spectral angular synchronization to consistently align all rings~\cite{singer2011angular}.  
After a global alignment to the ground truth, reconstructions achieve MSE $\approx 1\%$, using $R=100$ rings and $N_\varphi=720$ angular samples.  

In the 3D setting (Figure~\ref{fig:8}(b)), we reconstruct a Ribosome-S80 molecular volume~\cite{wong2014cryo}. As in the 2D case, we compute $A_f^{(2)}, A_f^{(4)}, A_f^{(5)}$, extract the $M_f^{(2)}$ and $M_f^{(3)}$ moments, and perform inversion using a frequency-marching solver in the spherical harmonic domain~\cite{bendory2025orbit}. The reconstruction uses angular bandlimit $L_{\max}=12$, $R=32$ concentric shells, and $N_\varphi=12{,}288$ spherical samples. The recovered volume captures the expected molecular structure at low resolution, with fidelity limited by the shell resolution and the chosen bandlimit $L_{\max}$.

\begin{figure*}[!t]
    \centering
    \includegraphics[width=0.95\linewidth]{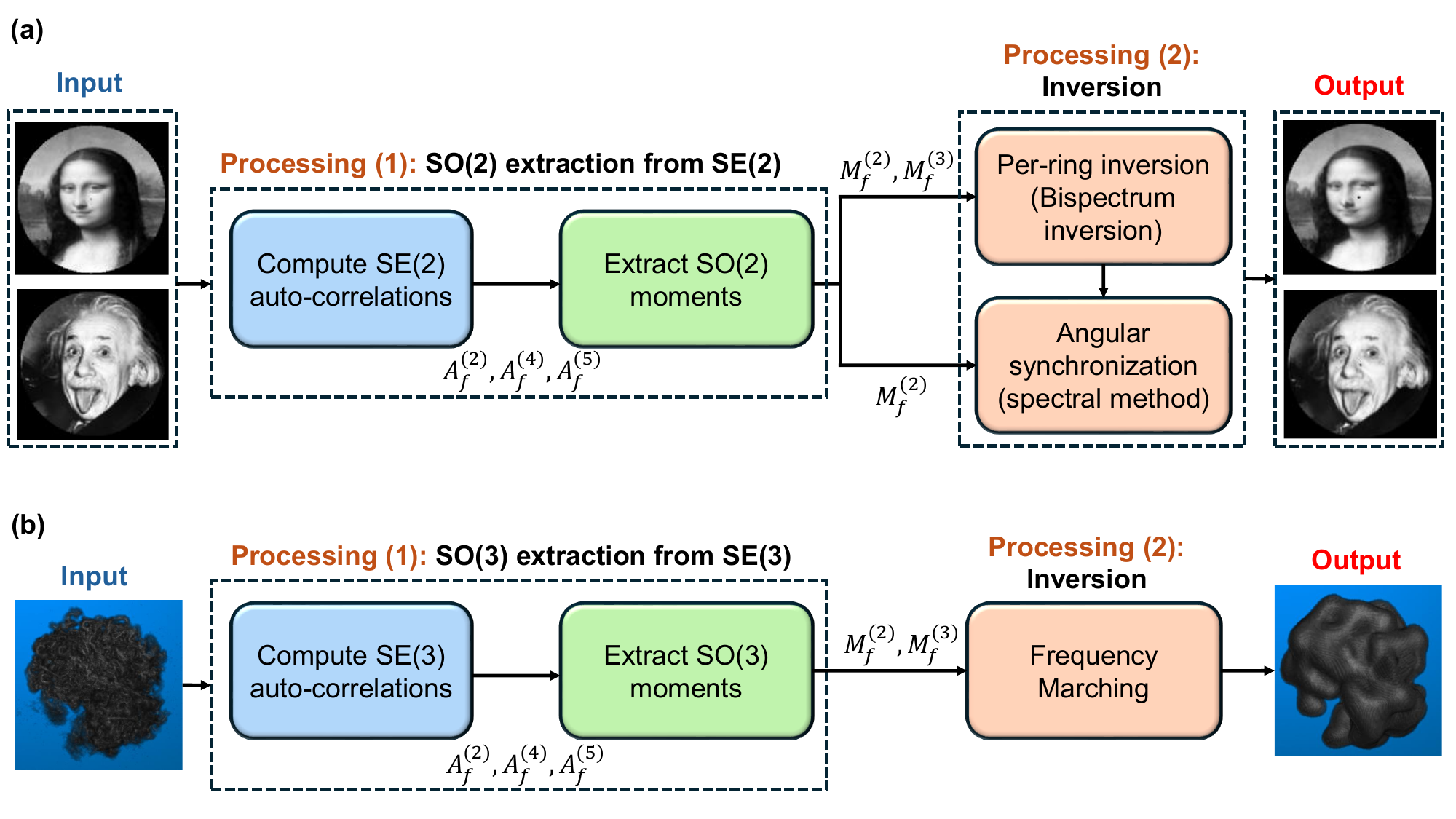}
    \caption{\textbf{End-to-end validation pipeline.} 
    \textbf{(a)}~2D case: a test image (Mona Lisa or Einstein) is processed by reducing $\SE(2)$ autocorrelations to $\SO(2)$ moments, followed by per-ring bispectrum inversion~\cite{kakarala2012bispectrum} and spectral angular synchronization~\cite{singer2011angular}. The reconstruction matches the ground truth with MSE $\approx 1\%$ and SSIM $\approx 0.9$, using $R=100$ rings and $N_\varphi=720$ samples.  
    \textbf{(b)}~3D case: the autocorrelations of a Ribosome-S80 volume~\cite{wong2014cryo} are computed, the $\SO(3)$ moments are extracted, and reconstructed via frequency-marching inversion~\cite{bendory2025orbit} with $L_{\max}=12$, $R=32$ shells, and $N_\varphi=12{,}288$ samples, producing a low-resolution recovery consistent with the imposed bandlimit.}
    \label{fig:8}
\end{figure*}

\section{Applications and outlook} \label{sec:applications}

One of the primary motivations for this work comes from foundational problems in cryo-EM~\cite{milne2013cryo, bai2015cryo, nogales2016development, renaud2018cryo, singer2020computational, yip2020atomic} and cryo-ET~\cite{chen2019complete, schaffer2019cryo, zhang2019advances, turk2020promise, watson2024advances}, where the central goal is to reconstruct a high-resolution 3D molecular structure from a large number of highly noisy 2D projection images. While the main focus of this work is the theoretical analysis of the orbit recovery problem under $\mathsf{SE}(n)$, these results naturally imply quantitative statements about the sample complexity of cryo-EM and cryo-ET reconstruction.
%In this section, we examine these implications in light of the theoretical framework developed earlier.

\subsection{Applications to cryo-ET} 
Cryo-ET enables the visualization of cellular structures such as organelles, viruses, and macromolecular complexes in situ~\cite{chen2019complete, schaffer2019cryo, zhang2019advances, turk2020promise, watson2024advances}. 
Its core objective is to reconstruct the three-dimensional structure of a whole cell or a cellular region by acquiring a tilt series of two-dimensional projections at known angles (typically ranging from $-60^\circ$ to $+60^\circ$). 
These projections are computationally combined to form a tomogram, a 3D reconstruction that preserves the spatial context of macromolecules inside the cell. 
From such tomograms, smaller sub-volumes centered on macromolecules of interest are extracted and referred to as \emph{subtomograms}.

Although in practice each tomogram contains many different macromolecules, it is instructive to consider a simplified statistical model in which the tomogram consists of multiple copies of the same structure $f$, placed at unknown positions and orientations. 
This abstraction fits naturally within the three-dimensional MTD framework ($n=3$): given an observation $y:\mathbb{R}^3 \to \mathbb{R}$, we write
\begin{align}
    y = \sum_{i=0}^{N-1} \mathcal{T}_i(\mathcal{R}_i \cdot f) + \xi,
    \label{eqn:cryoET}
\end{align}
where $\mathcal{R}_i \in \mathsf{SO}(3)$ are unknown rotations, $\mathcal{T}_i$ are translations, and $\xi$ is additive Gaussian white noise. 
In this setting, each transformed copy $f_i = \mathcal{R}_i \cdot f$ corresponds to a subtomogram, while the full observation $y$ models the entire tomogram. 
We emphasize that~\eqref{eqn:cryoET} is a simplified model of cryo-ET, not a full account of real data where macromolecular heterogeneity is common. 
Nevertheless, it serves as a mathematically tractable surrogate that links cryo-ET to the MTD framework and enables a rigorous sample-complexity analysis.

Traditional cryo-ET pipelines detect and extract subtomograms, then align and average them to improve resolution~\cite{zhang2019advances, watson2024advances}. 
However, at low SNR, especially for small or low-contrast complexes, detection becomes unreliable, and these approaches break down. 
Our analysis suggests that in principle, one can bypass explicit detection by working directly with tomograms via invariant statistics. 
It has been proved~\cite{bendory2025orbit} that the minimum moment order for the orbit recovery problem under $\mathsf{SO}(3)$ is $d=3$, which implies a sample complexity of $\omega(\sigma^6)$. 
By transferring this known result from the orbit recovery problem under $\mathsf{SO}(3)$ to the MTD setting, our bounds yield the following implication. 

\begin{corollary}[Sample complexity for the simplified cryo–ET model]
\label{cor:cryoET-bounds}
Under the simplified cryo–ET model~\eqref{eqn:cryoET} and the assumptions of Theorem~\ref{thm:mainTheorem}, suppose that the $\mathsf{SO}(3)$-orbit of a generic signal is identified by third-order moments, as proved in~\cite{bendory2025orbit}. Then, in the low-SNR regime, the sample complexity of the simplified cryo-ET model is bounded as follows: 
\begin{align}
    \omega(\sigma^6) \le  N^{\ast}_{\mathrm{cryo-ET}}(\sigma^2)  \le \omega(\sigma^{10}).
\end{align}
\end{corollary}

An important message of this work is that, under the simplified model~\eqref{eqn:cryoET}, recovery of a 3D macromolecule is possible at any SNR, provided that sufficiently many tomograms are available.
This suggests that, in principle, even small macromolecules that induce very low SNR may still be recoverable if enough data is collected. 
This stands in contrast to the prevailing belief in the cryo-EM community that small molecules cannot be reconstructed~\cite{henderson1995potential}. 

Finally, our sample complexity framework can incorporate additional structural priors. 
For example, Romanov et al.~\cite{romanov2021multi} analyzed the orbit recovery problem under $\mathsf{SO}(n)$ under a Gaussian prior in high dimensions, while Bendory et al.~\cite{bendory2024transversality} studied signals restricted to low-dimensional semi-algebraic sets. 
In such cases, the sample complexity of the orbit recovery problem under $\mathsf{SO}(3)$ drops to $\omega(\sigma^4)$, and by our bounds the corresponding cryo-ET MTD model has sample complexity between $\omega(\sigma^4)$ and $\omega(\sigma^8)$.

\subsection{Extensions toward cryo-EM} 
Our analysis focuses on orbit recovery models where a group acts directly on a signal $f$, without projection operations. 
In contrast, single-particle cryo-EM involves imaging frozen 3D macromolecules to generate noisy 2D projection images, where both the orientation and in-plane position of each particle are unknown (see Figure~\ref{fig:1}(a)). 
This process can be modeled for $y \in \mathbb{R}^2 \to \mathbb{R}$ as
\begin{align} 
    y = \sum_{i=0}^{N-1} \mathcal{T}_i \left( \Pi\left( \mathcal{R}_i \cdot f \right) \right) + \xi, \label{eqn:cryoEM} 
\end{align}
where $\mathcal{R}_i \in \mathsf{SO}(3)$ are unknown rotations, $\Pi$ is the tomographic projection operator mapping a 3D volume to a 2D image, $\mathcal{T}_i$ are random in-plane translations, and $\xi$ is additive noise. 
Thus, the observation $y \in \mathbb{R}^2 \to \mathbb{R}$ is a noisy superposition of randomly rotated and translated 2D projections of the underlying volume $f$.

Our main results do not directly extend to this full cryo-EM model because of the projection step, which violates the structure used in the extraction of $\mathsf{SO}(3)$ moments from $\mathsf{SE}(3)$ autocorrelations. 
However, in a special case where the structural envelope of the volume $f$ is preserved across all projection directions, the assumptions underlying our analysis remain valid. 
Specifically, if $f \in C^0(\mathcal{B}_R^{(3)})$ satisfies the radial boundary condition
\begin{align}
    f(R,\bm{\varphi}) = C \neq 0, \qquad \text{for all } \bm{\varphi} \in S^2,
    \label{eqn:cryoEMrequirement}
\end{align}
then every projection $\Pi(\mathcal{R}\cdot f)$ shares the same spherical envelope at radius $R$. 
Equivalently, if $h_{\mathcal{R}} = \Pi(\mathcal{R}\cdot f)$, then for any $\mathcal{R}_1,\mathcal{R}_2 \in \mathsf{SO}(3)$,
\begin{align}
    h_{\mathcal{R}_1}(R,\theta) = h_{\mathcal{R}_2}(R,\theta), \qquad \forall\, \theta \in [0,\pi].
\end{align}
That is, the projected images agree in their  boundary regardless of orientation. 
This radial invariance preserves the antipodal correlation structure (as in Lemma~\ref{lem:invarianceOfAntipodalProdSum}), thereby allowing the extraction and sample-complexity bounds to extend to this simplified cryo-EM model.

\begin{corollary}
[Sample complexity for cryo-EM under radial-envelope invariance]
\label{conj:cryoEM} 
If the volume $f$ satisfies condition~\eqref{eqn:cryoEMrequirement}, then the sample complexity bounds established in Theorem~\ref{thm:mainTheorem} for the orbit recovery problem under $\mathsf{SE}(3)$ extend to the cryo-EM observation model~\eqref{eqn:cryoEM}.
\end{corollary}

We emphasize that while the boundary condition~\eqref{eqn:cryoEMrequirement}, that the volume $f$ attains a constant value on the its  envelope, is an unrealistic idealization, we adopt it as a tractable surrogate that approximates common envelopes: many macromolecules have roughly spherical outer boundaries with near-uniform density (e.g., globular proteins via hydrophobic packing; viral capsids via icosahedral symmetry)~\cite{harrison2007principles,rossmann1989icosahedral}, and practical cryo-EM pipelines routinely apply spherical masks to approximate molecular support~\cite{scheres2012relion}. 
The analysis under this idealized boundary offers a clean baseline; in future work we aim to relax this assumption, accounting for non-spherical or spatially varying envelopes.

\subsection{Open questions and future work}
In what follows, we discuss the open questions remaining for future work.

\paragraph{Tightness of the bounds.}  
This work establishes both lower and upper bounds on the sample complexity of the orbit recovery problem under $\mathsf{SE}(n)$ and the MTD problem. However, the exact sample complexity remains unresolved. Notably, the upper bound is derived by extracting the $d$-th order $\mathsf{SO}(n)$ moment from only a limited subset (or slice) of the $(d+2)$-order $\mathsf{SE}(n)$ autocorrelation. As such, this bound may not be tight, and the true sample complexity could be substantially lower.
Interestingly, recent studies suggest that the same number of autocorrelation moments may be necessary for solving both the $\mathsf{SO}(2)$-orbit recovery and MTD problems for some specific cases~\cite{kreymer2022approximate, balanov2025note}.

\paragraph{Extension to the non-well-separated MTD model.} In the MTD model, a key open question is whether the current analysis can be extended to the non-well-separated regime, where signal occurrences are no longer guaranteed to be non-overlapping. In this setting, interactions between nearby occurrences introduce dependencies among the autocorrelations, complicating both the statistical modeling and the recovery process. It remains unclear whether the upper bound on sample complexity and its reduction to the rigid motion orbit recovery problem continue to hold in this more challenging setting.
A recent paper proposes analyzing this case through the lens of Markov chains and hard-core models, but it does not provide definitive sample complexity bounds~\cite{abraham2025sample}.

Notably, empirical studies suggest that violating the well-separation condition may have only a limited effect in practice. For example, \cite{kreymer2022two} observed that autocorrelations, and hence signal recovery performance, often remain robust even when signal occurrences are not well-separated, provided they do not overlap. These findings indicate that rigorous guarantees in the non–well-separated regime may still be achievable and merit further investigation.

\paragraph{Provable algorithms for the MTD problem.}
While prior work on the MTD model has largely relied on heuristic methods, this paper provides a proof-of-concept, provable extraction and recovery pipeline for orbit recovery in the MTD setting. A direct strategy is to estimate the orbit of $f$ from the empirical autocorrelations (Definition~\ref{def:autoCorrelationMomentsEmpirical}), but existing algorithms for MTD typically lack formal guarantees of stability or exact recovery. Proposition~\ref{thm:prop0} offers a more principled route by reducing MTD to a well-studied orbit recovery framework, thereby enabling the use of algorithms with established guarantees. Looking ahead, a natural direction is to design algorithms that are both provably accurate and computationally efficient, for instance by formulating orbit estimation as a semidefinite program (SDP) that leverages the polynomial structure of the autocorrelations; SDP relaxations have already shown strong theoretical and empirical performance in related orbit recovery problems~\cite{bandeira2014multireference,bendory2017bispectrum}, suggesting a promising pathway to scalable, provable algorithms for MTD.

\paragraph{Sampling in the MTD model.}
While our analysis of the MTD model is formulated in the continuous domain, practical imaging systems such as cryo-EM and cryo-ET operate on discretized micrographs sampled over finite Cartesian grids.
This raises fundamental questions about how sampling density influences the model’s statistical behavior and whether the theoretical sample-complexity guarantees persist under sampling. 

Our empirical results (Section~\ref{sec:empirical}) indicate that the scaling laws predicted by the continuous theory remain valid when the sampling resolution is adequate. 
In addition, previous studies~\cite{bendory2022signal,bendory2022super} have shown that invariant moments and bispectrum-based methods exhibit notable robustness to discretization, but these works only begin to address the intricate relationship between sampling resolution, identifiability, and invariant-based recovery, an area that warrants substantial further investigation.
Thus, a rigorous theoretical understanding of sampling artifacts remains largely open.

\paragraph{Heterogeneous MTD.}  
This work focuses on the homogeneous MTD model, where a single signal $f$ (together with its rigid-motion transformations) is embedded in the observation. In the heterogeneous setting, which is more faithful to the cryo-ET problem, multiple distinct signals appear, each undergoing its own rotations and translations. 
Here, invariants must be defined jointly with respect to $\mathsf{SO}(n)$ or $\mathsf{SE}(n)$ actions and the permutation group acting on the signal set. This greatly increases the algebraic and statistical complexity, mirroring the challenges observed in heterogeneous variants of the MRA problem. 
While some progress has been made in related orbit-recovery models~\cite{bandeira2023estimation,boumal2018heterogeneous}, establishing rigorous sample-complexity bounds for heterogeneous MTD remains a significant open problem.

%\amnon{We can add that the algebraic Theorem applied also to other groups except for the rigid-motion group (any semi-product group with "rotational" group $K \subset O(n)$ I think. In addition, we can add also linear operators.}

% \section*{Data Availability}
% The detailed implementation and code are available at \href{https://github.com/AmnonBa/rigid-motion-orbit-recovery}{https://github.com/AmnonBa/rigid-motion-orbit-recovery}.
% Tamir: I commented it since you put a link in the text

\section*{Acknowledgment}
T.B. and D.E. are supported in part by BSF under Grant 2020159. T.B. is also supported in part by NSF-BSF under Grant 2024791, in part by ISF under Grant 1924/21, and in part by a grant from The Center for AI and Data Science at Tel Aviv University (TAD).

\bibliographystyle{plain}
%\bibliography{refs} 

\begin{appendices}

{\centering{\section*{Appendix}}}
\paragraph{Appendix organization.} 
This appendix is organized as follows.
Appendix~\ref{sec:reduction-from-SEn-to-SOn} develops the algebraic foundations of our analysis.
Appendices~\ref{sec:preliminariesToStatisticalPart}--\ref{sec:statistical-analysis-orbit-recoverry} treat the statistical aspects of orbit recovery under $\mathsf{SE}(n)$.
Appendix~\ref{sec:MTD-statistical-analysis} presents the statistical analysis of the MTD model.
Finally, Appendices~\ref{app:experimentalMethods2D}--\ref{app:experimentalMethods3D} detail the empirical setups and simulations for the 2D and 3D problems, respectively.

\section{Extraction of SO(n) moments from SE(n) autocorrelations}
\label{sec:reduction-from-SEn-to-SOn}

\subsection{Proof of Lemma \ref{lem:invarianceOfAntipodalProdSum}} \label{sec:proofOfinvarianceOfAntipodalProdSum}
In Cartesian coordinates, the antipodal correlation of a function $f$ can be written as 
\begin{align}
    C_{f} = \int_{S^{n-1}(R)} f(\bm{x}) \, f(-\bm{x}) \, d\bm{x}, \label{eqn:app_A0_1}
\end{align}
where the integration is taken over the sphere of radius $R$. In particular,
if $\mathcal{R}$ is any rotation, then 
\begin{align}
    C_{\mathcal{R} \cdot f} = \int_{S^{n-1}(R)} f(\mathcal{R}^{-1}\bm{x}) \, f(-\mathcal{R}^{-1}\bm{x}) \, d\bm{x}. \label{eqn:app_A0_2}
\end{align}

Letting $\bm{y} = \mathcal{R}^{-1} \bm{x}$, we note that since $\mathcal{R} \in \mathsf{SO}(n)$ is a rotation, i.e. a continuous, orthogonal transformation with determinant one, the change of variables $\bm{x} \mapsto \bm{y}$ preserves volume and maps the sphere $S^{n-1}(R)$ onto itself bijectively. Therefore, we can rewrite the integral as  
\begin{align}
    \int_{S^{n-1}(R)} f(\mathcal{R}^{-1}\bm{x}) \, f(-\mathcal{R}^{-1}\bm{x}) \, d\bm{x} 
    = \int_{S^{n-1}(R)} f(\bm{y}) \, f(-\bm{y}) \, d\bm{y}. \label{eqn:app_A0_3}
\end{align}
Since this holds for all $\mathcal{R} \in \mathsf{SO}(n)$, we conclude that the antipodal correlation is invariant under the action of $\mathsf{SO}(n)$.

\subsection{Proof of Lemma~\ref{lem:autocorrelationinvariant}}\label{sec:proofOfInvariance}
Let $f : \mathbb{R}^n \to \mathbb{R}$ be a continuous function supported on a compact subset of $\mathbb{R}^n$. The $d$-th order autocorrelation of $f$, with respect to a rotational distribution $\rho$ (as defined in Definition~\ref{def:autoCorrelationNoiseFree}), is given by:
\begin{align}
    \nonumber A_{f}^{(d)}(\bm{\tau}_0, \ldots, \bm{\tau}_{d-1}) 
    &= \int_{\mathbb{R}^n} \mathbb{E}_{g \sim \mu_{\mathsf{SE}(n)}} \left[ \prod_{j=0}^{d-1} f_g(\bm{x} + \bm{\tau}_j) \right] d\bm{x} \\
    \nonumber
    &= \int_{\mathbb{R}^n} \int_{\mathsf{SE}(n)} \prod_{j=0}^{d-1} f_g(\bm{x} + \bm{\tau}_j) \, d\mu_{\mathsf{SE}(n)}(g) \, d\bm{x} \\
    &= \int_{\mathbb{R}^n} \int_{\mathcal{R} \in \mathsf{SO}(n)} \int_{\bm{t} \in \mathbb{R}^n} 
    \prod_{j=0}^{d-1} (\mathcal{R}, \bm{t}) \cdot f(\bm{x} + \bm{\tau}_j) \, d\mu(t)\, d\rho(\mathcal{R}) \, d\bm{x},
    \label{eqn:autoCorrelationNoiseFreeProduct3}
\end{align}
where the second equality follows from Assumption~\ref{assum:rotation_translation_indep} that the rotations and translations are drawn independently.

We now show that \eqref{eqn:autoCorrelationNoiseFreeProduct3} is equivalent to \eqref{eqn:autoCorrelationNoiseFreeProduct2}. By definition
\begin{align}
    A_{f, \rho}^{(d)}(\bm{\tau}_0, \ldots, \bm{\tau}_{d-1}) 
    &= \int_{\mathbb{R}^n} \int_{\mathcal{R} \in \mathsf{SO}(n)} \int_{\bm{t} \in \mathbb{R}^n} 
    \prod_{j=0}^{d-1} f((\mathcal{R}, \bm{t}) \cdot (\bm{x} + \bm{\tau}_j)) \, d\mu(\bm{t})\, d\rho(\mathcal{R}) \, d\bm{x} \nonumber \\
    &= \int_{\mathbb{R}^n} \int_{\mathcal{R} \in \mathsf{SO}(n)} \int_{\bm{t} \in \mathbb{R}^n} 
    \prod_{j=0}^{d-1} f(\mathcal{R}^{-1}(\bm{x} + \bm{\tau}_j - \bm{t}))  d\mu(\bm{t})\, d\rho(\mathcal{R}) \, d\bm{x}. 
    \label{eqn:app_B04}
\end{align}    
Applying the change of variables $\bm{z} = \bm{x} - \bm{t}$, we can rewrite~\eqref{eqn:app_B04} as
\begin{align}
    &= \int_{\mathbb{R}^n} \int_{\mathcal{R} \in \mathsf{SO}(n)} \int_{\bm{t} \in \mathbb{R}^n} 
    \prod_{j=0}^{d-1} f(\mathcal{R}^{-1}(\bm{z} + \bm{\tau}_j)) \, d\mu(\bm{t})\, d\rho(\mathcal{R}) \, d\bm{z}.
    \label{eqn:app_B05}
\end{align}

Since $\mu(\bm{t})$ is a probability measure, it satisfies
$\int_{\bm{t} \in \mathbb{R}^n} d\mu(\bm{t})= 1$.
Substituting this into~\eqref{eqn:app_B05}, we obtain
\begin{align}
    A_{f, \rho}^{(d)}(\bm{\tau}_0, \ldots, \bm{\tau}_{d-1}) 
    = \int_{\mathbb{R}^n} \int_{\mathcal{R} \in \mathsf{SO}(n)} 
    \prod_{j=0}^{d-1} f_\mathcal{R}(\bm{x} + \bm{\tau}_j) \, d\rho(\mathcal{R}) \, d\bm{x},
    \label{eqn:autoCorrelationNoiseFreeProduct4}
\end{align}
where $f_\mathcal{R}(\bm{x}) = f(\mathcal{R}^{-1} \bm{x})$. This concludes the proof.

\subsection{Proof of Proposition \ref{prop:invarianceOfAutoCorrelation}} \label{sec:proofOfinvarianceOfAutoCorrelation}

To establish invariance under rigid motions, it suffices to show that 
\begin{align}
    A^{(d)}_{h, \rho} = A^{(d)}_{f, \rho},
\end{align}
for every transformation $g = (\mathcal{R}_0, \bm{t}_0) \in \SE(n)$, where 
\begin{align}
    h(\bm{x}) = (g \cdot f)(\bm{x}) = f\!\left(\mathcal{R}_0^{-1}(\bm{x} - \bm{t}_0)\right).
\end{align}
By Lemma~\ref{lem:autocorrelationinvariant}, the autocorrelation of $h$ is given by
\begin{align}
    A^{(d)}_{h, \rho}(\bm{\tau}_0, \ldots, \bm{\tau}_{d-1}) 
    &= \int_{\mathbb{R}^n} \int_{\mathsf{SO}(n)} \prod_{i=0}^{d-1} h_{\mathcal{R}}(\bm{x} + \bm{\tau}_i) \, d\rho({\mathcal{R}}) \, d\bm{x} \nonumber\\
    &= \int_{\mathbb{R}^n} \int_{\mathsf{SO}(n)} \prod_{i=0}^{d-1} f\left({\mathcal{R}}^{-1} {\mathcal{R}}_0^{-1}(\bm{x} + \bm{\tau}_i - \bm{t}_0)\right) \, d\rho({\mathcal{R}}) \, d\bm{x}.
\end{align}

Applying the change of variables $ \bm{z} = \bm{x} - \bm{t}_0 $ yields
\begin{align} \label{eq:rho0_shift}
    A^{(d)}_{h, \rho}(\bm{\tau}_0, \ldots, \bm{\tau}_{d-1}) 
    &= \int_{\mathbb{R}^n} \int_{\mathsf{SO}(n)} \prod_{i=0}^{d-1} f\left({\mathcal{R}}^{-1} {\mathcal{R}}_0^{-1}(\bm{z} + \bm{\tau}_i)\right) \, d\rho({\mathcal{R}}) \, d\bm{z}. 
\end{align}
Since we assume that $d\rho(\mathcal{R})$ is a Haar measure on $ \mathsf{SO}(n) $, it is left-invariant, that is, for every $Q \in \mathsf{SO}(n)$, $\rho (Q\mathcal{R}) = \rho(\mathcal{R})$,
we can rewrite~\eqref{eq:rho0_shift} as
\begin{align}
    A^{(d)}_{h, \rho}(\bm{\tau}_0, \ldots, \bm{\tau}_{d-1}) 
    &= \int_{\mathbb{R}^n} \int_{\mathsf{SO}(n)} \prod_{i=0}^{d-1} f\left(\tilde{{\mathcal{R}}}^{-1}( \bm{z} +  \bm{\tau}_i\right) \, d\rho(\tilde{{\mathcal{R}}}) \, d\bm{z} \nonumber \\
    &= A^{(d)}_{f, \rho}(\bm{\tau}_0, \ldots, \bm{\tau}_{d-1}),
\end{align}
where $\tilde{\mathcal{R}} = \mathcal{R}_0 \mathcal{R}$. 
Hence, the autocorrelation function is invariant under rigid motions, as claimed.

\subsection{Proof of Theorem \ref{thm:reductionFromAutocorrelationToTensorMoment}}
\label{sec:auxiliartForTheoremreductionFromAutocorrelationToTensorMoment}

\begin{proposition}   
\label{prop:quivalenceBeweenAutoCorrelationAndSphericalCordinated}
With the hypothesis as in the statement of Theorem~\ref{thm:reductionFromAutocorrelationToTensorMoment}, 
we have: 
    \begin{align}
        \nonumber \lim_{\delta \to 0^+} &\frac{\displaystyle\int_{S^{n-1}} A_{f, \rho}^{(d+2)}\left( \bm{\tau}_0^{(\delta)}(\bm{\theta}), \bm{\tau}_1^{(\delta)}(\bm{\theta}), \bm{\eta}_1, \dots, \bm{\eta}_d \right) \, d\bm{\theta}}{ \displaystyle\int_{S^{n-1}} A_{f, \rho}^{(2)}\left( \bm{\tau}_0^{(\delta)}(\bm{\theta}), \bm{\tau}_1^{(\delta)}(\bm{\theta}) \right) \, d\bm{\theta}} 
        \\ &= \frac{\displaystyle\int_{S^{n-1}} \int_{\SO(n)} f_{\mathcal{R}}(R, \bm{\theta}) f_{\mathcal{R}}(-R, \bm{\theta}) \prod_{i=1}^d f_{\mathcal{R}}(\bm{\eta}_i)\, d\rho(\mathcal{R}) \, d\bm{\theta}}
        {\displaystyle\int_{S^{n-1}} \int_{\SO(n)} f_{\mathcal{R}}(R, \bm{\theta}) f_{\mathcal{R}}(-R, \bm{\theta})  \, d\rho(\mathcal{R}) d\bm{\theta}}.   \label{eqn:quivalenceBeweenAutoCorrelationAndSphericalCordinated}
    \end{align}
\end{proposition} 

\begin{proof}[Proof of Proposition~\ref{prop:quivalenceBeweenAutoCorrelationAndSphericalCordinated}]
By Lemma~\ref{lem:autocorrelationinvariant},
\begin{align}
    & A_{f, \rho}^{(d)}(\bm{\tau}_0^{(\delta)}(\theta) , \tau_1^{(\delta)}(\theta), \eta_1 \ldots, \eta_{d-1}) \notag \\
     & = \int_{\mathbb{R}^n} \int_{\mathcal{R} \in \mathsf{SO}(n)} 
    \left[f_{\mathcal{R}}(\bm{x} + \bm{\tau}_0^{(\delta)}(\theta)) f_{\mathcal{R}}(\bm{x} + \bm{\tau}^{(\delta)}_1(\theta)) \prod_{i=1}^d f_{\mathcal{R}}(\bm{x} + \bm{\eta}_i) \right] d\rho(\mathcal{R}) \, d\bm{x},
    \label{eqn:autoCorrelationOrderDplus2}
\end{align}
and 
\begin{align}
    A_{f, \rho}^{(2)}(\bm{\tau}_0^{(\delta)}(\theta) , \tau_1^{(\delta)}(\theta)) = \int_{\mathbb{R}^n} \int_{\mathcal{R} \in \mathsf{SO}(n)} 
   \left[ f_{\mathcal{R}}(\bm{x} + \bm{\tau}_0^{(\delta)}(\theta)) f_{\mathcal{R}}(\bm{x} + \bm{\tau}^{(\delta)}_1(\theta))\right]
d\rho(\mathcal{R}) \, d\bm{x}. \label{eqn:autoCorrelationOrder2}
\end{align}

By Assumption~\ref{assum:support}, the support of $f$ is contained in the ball $\mathcal{B}_R^{(n)}$, so the integrand in~\eqref{eqn:autoCorrelationOrderDplus2} and~\eqref{eqn:autoCorrelationOrder2} vanishes if $(\bm{\tau}^{(\delta)}(\theta) + \bm{x})$ or $(\bm{\tau}^{(\delta)}(\theta) + \bm{x})$ lie outside the ball of radius $R$. In particular, when $\delta = 0$ the integrand is zero except when $\bm{x} = \bm{x}_0 \triangleq {0}$, that is, it vanishes everywhere except in the origin. 
Thus, both the numerator and the denominator in~\eqref{eqn:quivalenceBeweenAutoCorrelationAndSphericalCordinated} approach zero as $\delta \to 0$, so to evaluate the limiting ratio, we must consider the behavior of the integrands as $\delta \to 0$ rather than their pointwise values at $\delta = 0$. 
By assumption, $f$ is a continuous, bounded function supported on a compact subset of $\mathbb{R}^n$. Thus, by Fubini's Theorem, we may interchange the order of integration in the numerator on the left-hand side of~\eqref{eqn:quivalenceBeweenAutoCorrelationAndSphericalCordinated} as follows:
\begin{align}
    &\int_{S^{n-1}} A_{f, \rho}^{(d+2)}\left( \bm{\tau}_0^{(\delta)}(\bm{\theta}), \bm{\tau}_1^{(\delta)}(\bm{\theta}), \bm{\eta}_1, \dots, \bm{\eta}_d \right) \, d\bm{\theta} \notag \\
    &= \int_{S^{n-1}} \int_{\mathbb{R}^n} \int_{\SO(n)} \left[ f_{\mathcal{R}}(\bm{x} + \bm{\tau}_0^{(\delta)}(\bm{\theta})) f_{\mathcal{R}}(\bm{x} + \bm{\tau}_1^{(\delta)}(\bm{\theta})) \prod_{i=1}^d f_{\mathcal{R}}(\bm{x} + \bm{\eta}_i) \right] \, d\rho(\mathcal{R}) \,d\bm{x} \, d\bm{\theta} \notag \\
    &= \int_{\mathbb{R}^n} \int_{S^{n-1}} \int_{\SO(n)} \left[ f_{\mathcal{R}}(\bm{x} + \bm{\tau}_0^{(\delta)}(\bm{\theta})) f_{\mathcal{R}}(\bm{x} + \bm{\tau}_1^{(\delta)}(\bm{\theta})) \prod_{i=1}^d f_{\mathcal{R}}(\bm{x} + \bm{\eta}_i) \right] \,  d\rho{\mathcal{R}} \, d\bm{\theta} \, d\bm{x}.
    \label{eqn:numeratorSwappedIntegration}
\end{align}
Similarly, for the denominator on the left-hand side of~\eqref{eqn:quivalenceBeweenAutoCorrelationAndSphericalCordinated}, we have
\begin{align}
     &\int_{S^{n-1}} A_{f, \rho}^{(2)}\left( \bm{\tau}_0^{(\delta)}(\bm{\theta}), \bm{\tau}_1^{(\delta)}(\bm{\theta}) \right) \, d\bm{\theta} \notag \\
     &= \int_{\mathbb{R}^n} \int_{S^{n-1}} \int_{\SO(n)} \left[ f_{\mathcal{R}}(\bm{x} + \bm{\tau}_0^{(\delta)}(\bm{\theta})) f_{\mathcal{R}}(\bm{x} + \bm{\tau}_1^{(\delta)}(\bm{\theta})) \right] \,
     d\rho(\mathcal{R}) \, d\bm{\theta} \, d\bm{x}.
     \label{eqn:denominatorSwappedIntegration}
\end{align}
Set
\begin{align}
    F_1({\bm{x}}, \delta) & \triangleq \int_{S^{n-1}} \int_{\SO(n)} \left[ f_{\mathcal{R}}(\bm{x} +\bm{\tau}_0^{(\delta)}\p{\bm{\theta}}) f_{\mathcal{R}}(\bm{x} + \bm{\tau}_1^{(\delta)}\p{\bm{\theta}}) \prod_{i=1}^d f_{\mathcal{R}}(\bm{x} + \bm{\eta}_i) d\rho(\mathcal{R}) \right] \,d {\bm{\theta}}
    \nonumber\\ & =  \int_{\SO(n)} \left[\prod_{i=1}^d f_{\mathcal{R}}(\bm{x} + \bm{\eta}_i) \left(\int_{S^{n-1}}f_{\mathcal{R}}(\bm{x} +\bm{\tau}_0^{(\delta)}\p{\bm{\theta}}) f_{\mathcal{R}}(\bm{x} + \bm{\tau}_1^{(\delta)}{\bm{\theta}})d \bm{\theta} \right) \right] \, d\rho(\mathcal{R}),
\end{align}
where the second equality follows by Fubini’s theorem. Similarly, define
\begin{align}
    F_2(\bm{x},\delta) \triangleq \int_{S^{n-1}} \int_{SO(n) } \left[ f_{\mathcal{R}}(\bm{x} +\bm{\tau}_0^{(\delta)}\p{\bm{\theta}}) f_{\mathcal{R}}(\bm{x} + \bm{\tau}_1^{(\delta)}\p{\bm{\theta}})  d \rho(\mathcal{R})\right] \, d {\bm{\theta}} \nonumber
\\  = 
\int_{S^{n-1}} \int_{SO(n) } \left[ f_{\mathcal{R}}(\bm{x} +\bm{\tau}_0^{(\delta)}\p{\bm{\theta}}) f_{\mathcal{R}}(\bm{x} + \bm{\tau}_1^{(\delta)}\p{\bm{\theta}}) d {\bm{\theta}} \right] d \rho(\mathcal{R}).
\end{align}
With the above definitions, the left-hand side of \eqref{eqn:quivalenceBeweenAutoCorrelationAndSphericalCordinated} can be rewritten as
\begin{align}
    \lim_{\delta \to 0^+} \frac{\int_{\mathbb{R}^n} F_1(\bm{x},\delta) d\bm{x}}{\int_{\mathbb{R}^n} F_2(\bm{x}, \delta) d\bm{x}}. \label{eqn:app_B13}
\end{align}

For $\delta \in [0,1)$,  define
\begin{align}
    \mathcal{A}_\delta \triangleq \left\{ \bm{x} \in \mathbb{R}^n  \middle|  
    \exists \bm{\theta} \in S^{n-1} \text{ such that } 
    \norm{\bm{x} + \bm{\tau}_0^{\delta}(\bm{\theta})} \leq R 
    \text{ and } 
    \norm{\bm{x} + \bm{\tau}_1^{\delta}(\bm{\theta})} \leq R 
    \right\}.
\end{align}
Note that as $\mathcal{A}_{0} = \ppp{\bm{x}_0}$, i.e., the point $\bm{x}_0$ is the only point in the set of $\mathcal{A}_{\delta}$, for $\delta = 0$ and a simple computation using the law of cosines shows that for $0 \leq \delta < 1$, $A_\delta$ is the ball of radius $R \sqrt{1 - (1-\delta)^2}$.

Under the definition of the set $\mathcal{A}_\delta$, and of the assumption that $f$ is defined on $\mathcal{B}_R^{(n)}$, it is clear that $F_1({\bm{x}},\delta) = 0$,  $F_2(\bm{x},\delta) = 0$, for every $\bm{x} \notin \mathcal{A}_\delta$, leading to,
\begin{align}
    \int_{\bm{x} \in \mathbb{R}^n} F_r(\bm{x},\delta) d \bm{x} = \int_{\bm{x} \in \mathcal{A}_\delta} F_r(\bm{x},\delta) d \bm{x},
\end{align}
for $r = 1,2$.
Then, applying the Lebesgue Differentiation Theorem~\cite{folland1999real} to $F_1$ and $F_2$ around $\bm{x}_0$, we obtain:
\begin{align}
    \nonumber \lim_{\delta \to 0^+} \frac{1}{\mu \p{\mathcal{A}_\delta}} & \int_{\bm{x} \in \mathcal{A}_\delta} F_2(\bm{x}, \delta)  d \bm{x} = F_2(\bm{x}_0, 0)
    \\ & = \int_{S^{n-1}}\int_{\SO(n)} \pp{f_{\mathcal{R}} \p{R, \bm{\theta}}f_{\mathcal{R}} \p{-R, \bm{\theta}}} \  d\rho(\mathcal{R}) \, d{\bm{\theta}}, \label{eqn:app_B15}
\end{align}
and,
\begin{align}
    \nonumber \lim_{\delta \to 0^+}&  \frac{1}{\mu(\mathcal{A}_\delta)} \int_{\bm{x} \in \mathcal{A}_\delta} F_1(\bm{x}, \delta) d\bm{x}
    = F_1(\bm{x}_0,0) \
    \\ &= \int_{S^{n-1}} \int_{\SO(n)} \left[ f_{\mathcal{R}}(R, \bm{\theta}) f_{\mathcal{R}}(-R, \bm{\theta}) \prod_{i=1}^d f_{\mathcal{R}}
    (\bm{\eta}_i) d\rho(\mathcal{R}) \, \right] d\bm{\theta}. \label{eqn:app_B16}
\end{align}
Note that the continuity of $f$ ensures that the above limits exist.
Substituting \eqref{eqn:app_B15}, \eqref{eqn:app_B16} into \eqref{eqn:app_B13} results:
\begin{align}
    \lim_{\delta \to 0^+} \frac{\int_{\mathbb{R}^n} F_1(\bm{x}, \delta) d\bm{x}}{\int_{\mathbb{R}^n} F_2(\bm{x},\delta) d\bm{x}} = \frac{\int_{S^{n-1}} \int_{\SO(n)} \left[ f_{\mathcal{R}}(R, \bm{\theta}) f_{\mathcal{R}}(-R, \bm{\theta}) \prod_{i=1}^d f_{\mathcal{R}}(\bm{\eta}_i) d\rho(\mathcal{R}) \right] d\bm{\theta}}{\int_{S^{n-1}}\int_{\SO(n)} \pp{f_{\mathcal{R}} \p{R, \bm{\theta}}f_{\mathcal{R}} \p{-R, \bm{\theta}}} \ d{\bm{\theta}}}. 
\end{align}
Recognizing $F_1$, $F_2$, as the integrands of the left-hand-side numerator and denominator of \eqref{eqn:quivalenceBeweenAutoCorrelationAndSphericalCordinated}, respectively, completes the proof of the proposition.
\end{proof}

We can now complete the proof of Theorem~\ref{thm:reductionFromAutocorrelationToTensorMoment} by showing that the right-hand side of~\eqref{eqn:quivalenceBeweenAutoCorrelationAndSphericalCordinated} equals the $d$-th moment of $\SO(n)$ as defined in Definition~\ref{def:mraTensorMoment}.

As $f$ is a bounded function defined on a compact domain of $\mathbb{R}^n$, we apply Fubini's Theorem to interchange the order of integration 
\begin{align}
    \nonumber \int_{S^{n-1}} & \int_{\SO(n)} f_{\mathcal{R}} (R, \bm{\theta})f_{\mathcal{R}} (-R, \bm{\theta})
    f_{\mathcal{R}} (\bm{\eta}_1)  \dots f_{\mathcal{R}}(\bm{\eta}_d)  d\bm{\theta} \, d\rho(\mathcal{R})\\
     & = \int_{\SO(n)} \left(\int_{S^{n-1}} f_{\mathcal{R}} (R, \bm{\theta})f_{\mathcal{R}} (-R, \bm{\theta}) \ d\bm{\theta} \right) f_{\mathcal{R}}(\bm{\eta}_1) \dots f_{\mathcal{R}} (\bm{\eta}_d) d \rho(\mathcal{R}). \label{eqn:interchangeIntegralOrder}
\end{align}

By Assumption \ref{assum:nonVanishingSupport}, and Lemma \ref{lem:invarianceOfAntipodalProdSum}, the expression,
\begin{align}
    C_{f,{\mathcal{R}}} = \int_{S^{n-1}} f_{\mathcal{R}} \p{R, \bm{\theta}}f_{\mathcal{R}} \p{-R, \bm{\theta}} d\bm{\theta} \neq 0,
\end{align}
and is independent of the group element ${\mathcal{R}} \in \mathsf{SO}(n)$. Therefore, the expression 
\begin{align}
    \notag \int_{\SO(n)} \left(\int_{S^{n-1}} f_{\mathcal{R}} (R, \bm{\theta})f_{\mathcal{R}} (-R, \bm{\theta}) \ d\bm{\theta} \right) f_{\mathcal{R}}(\bm{\eta}_1) \dots f_{\mathcal{R}} (\bm{\eta}_d) d \rho(\mathcal{R}) & = \\ \notag 
    \left(\int_{S^{n-1}} f(R, \bm{\theta}) f(-R, \bm{\theta}) d\bm{\theta}\right) \int_{\SO(n)} f_\mathcal{R}(\bm{\eta}_1) \ldots f_{\mathcal{R}}(\bmeta_d) d\rho(\mathcal{R}) 
    & = \\ 
    \left(\int_{S^{n-1}} f(R, \bm{\theta}) f(-R, \bm{\theta}) d\bm{\theta}\right) M^{(d)}_f(\bm{\eta}_1, \ldots , \bm{\eta}_d), 
\end{align}
and the theorem follows.

\section{Preliminaries for the statistical framework}
\label{sec:preliminariesToStatisticalPart}

\subsection{Problem formulation}
The generative model for the $\mathsf{SE}(n)$ orbit-recovery problem, introduced in Problem~\ref{prob:orbitRecoverySEn}, assumes that each observation $y_i: \mathcal{D} \to \mathbb{R}$ is given by
\begin{align}
    y_i(\bm{x}) = (g_i \cdot f)(\bm{x}) + \xi_i(\bm{x}), 
    \qquad i \in \pp{N}, \bm{x} \in \mathcal{D}, \label{eqn:modelSEnProb2}
\end{align}
where $\mathcal{D} \subset \mathbb{R}^n$ is a compact domain, $\{ g_i \}_{i \in \pp{N}}$ are i.i.d. samples from $\mathsf{SE}(n)$, and $\{\xi_i\}_{i \in \pp{N}}$ are i.i.d. Gaussian white noise processes with variance $\sigma^2$. That is,
\begin{align}
    \mathbb{E}[\xi_i(\bm{x}) \xi_i(\bm{x}')] = \sigma^2 \, \delta(\bm{x} - \bm{x}'), 
    \qquad \bm{x}, \bm{x}' \in \mathcal{D},
\end{align}
and the sequences $\{ g_i \}$ and $\{ \xi_i \}$ are independent across $i$.

A rigorous treatment of this model requires a continuous-space representation of the noise $\xi_i$. In the discrete setting, white noise is modeled as i.i.d. Gaussian random variables. In the continuum, however, such a pointwise definition breaks down: sample paths are almost surely discontinuous, and the variance diverges at each location. Instead, white noise must be understood as a random distribution, i.e., a continuous linear functional acting on Schwartz test functions~\cite{holden1996stochastic,geng2020wiener}. The exposition in this section follows the standard construction of spatial white noise processes presented in~\cite{holden1996stochastic,geng2020wiener,chung1990introduction, peccati2011wiener, adler2007random}.

Formally, the analogue of discrete white noise is a \emph{spatial white noise process}, a random field $\xi(\bm{x})$ indexed by $\bm{x}\in\mathbb{R}^n$ with uncorrelated values at distinct points. The random quantities that arise in our analysis include linear functionals,
\begin{align}
    \int_{\mathcal{D}} F(\bm{x}) \, \xi(\bm{x}) \, d\bm{x},
    \label{eq:stochastic_integral_linear}
\end{align}
and multilinear functionals (corresponding to the autocorrelation quantities as defined in~\eqref{eqn:autoCorrealtionMomentsRigidMotion}),
\begin{align}
    \int_{\mathcal{D}} F(\bm{x}) \,\xi(\bm{x}+\bm{\tau}_0)\cdots \xi(\bm{x}+\bm{\tau}_{d-1}) \, d\bm{x}. \label{eq:stochastic_integral_multilinear}
\end{align}
To carefully define and analyze~\eqref{eq:stochastic_integral_linear}-\eqref{eq:stochastic_integral_multilinear}, we first recall the connection between the spatial white noise process and Brownian sheet. The Brownian sheet generalizes a one-dimensional Brownian motion to multiple parameters, as defined next.

\begin{definition}[Gaussian white-noise random measure and Brownian sheet]
\label{def:nparamBM}
A stochastic process $\{W(A)\}_{A\in\mathcal{B}_b(\mathbb{R}^n)}$ indexed by bounded Borel sets is called a \emph{Gaussian white-noise random measure} if:
\begin{enumerate}
    \item For each bounded Borel set $A\subset\mathbb{R}^n$, $W(A)$ is a centered Gaussian random variable with variance $\mathbb{E}[W(A)^2]=|A|$, where $|A|$ denotes the Lebesgue measure of the set $A$.
    \item For any bounded Borel sets $A,B\subset\mathbb{R}^n$, $\mathbb{E}[W(A)W(B)] = |A\cap B|$.
    \item If $A_1,\dots,A_k$ are pairwise disjoint, then $W(A_1),\dots,W(A_k)$ are independent.
\end{enumerate}
Equivalently (and in the terminology often used in the stochastic partial differential equations literature), the cumulative field
\begin{align}
    \mathcal{W}(t) \triangleq  W\!\big([0,t_1]\times\cdots\times[0,t_n]\big),\qquad t\in\mathbb{R}_+^n,
\end{align}
is the standard $n$-parameter \emph{Brownian sheet} with covariance $\mathbb{E}[\mathcal{W}(t)\mathcal{W}(s)]=\prod_{k=1}^n \min\{t_k,s_k\}$.
\end{definition}

Then, the spatial white noise process at noise level $\sigma>0$ is the generalized derivative of the cumulative Brownian sheet, expressed heuristically as
\begin{align}
    \xi(\bm{x})\,d\bm{x}  =  \sigma\, dW(\bm{x}).
    \label{eqn:app_B6}
\end{align}

\paragraph{White noise as an isonormal Gaussian process.}
Let $\mathcal{S}(\mathbb{R}^n)$ denote the Schwartz space.
For any $\varphi \in \mathcal{S}(\mathbb{R}^n)$, define
\begin{align}
    \xi(\varphi)  \triangleq   \sigma \int_{\mathbb{R}^n} \varphi(\bm{x}) \, dW(\bm{x}),
    \label{eqn:WienerItoIntegral}
\end{align}
which is a centered Gaussian random variable with variance $\sigma^2 \|\varphi\|_{L^2}^2$.
Then, the linear functional in~\eqref{eq:stochastic_integral_linear} corresponds to $\varphi(\bm{x})=F(\bm{x})\mathbf{1}_{\mathcal{D}}(\bm{x})$, while higher-order functionals such as~\eqref{eq:stochastic_integral_multilinear} are defined using multiple Wiener–It$\hat{\text{o}}$ integrals (see Appendix~\ref{sec:multiple-Wiener-ito-integral}).

Equivalently, spatial white noise may be characterized as an isonormal Gaussian process
\begin{align}
    \xi: \mathcal{S}(\mathbb{R}^n) \to \mathbb{R},
\end{align}
with covariance
\begin{align}
    \mathbb{E}[\xi(\varphi)\,\xi(\psi)]
     = \sigma^2 \cdot \langle \varphi,\psi\rangle_{L^2}
     = \sigma^2 \int_{\mathbb{R}^n} \varphi(\bm{x})\,\psi(\bm{x})\,d\bm{x},
    \label{eqn:def-white-noise-cov-structure}
\end{align}
for all $\varphi,\psi\in\mathcal{S}(\mathbb{R}^n)$.
Each realization of white noise is a tempered distribution
$\omega\in\mathcal{S}'(\mathbb{R}^n)$ acting via
$\xi(\varphi)(\omega)=\sigma \langle \omega,\varphi\rangle$.

\paragraph{Observations in differential form.}
We now return to the description of~\eqref{eqn:modelSEnProb2} in a differential form.
Let $W_i(\bm{x})$, for $i \in \pp{N}$, denote independent $n$-parameter Gaussian white-noise random measure. Defining the spatial white noise process as their distributional derivative (see~\eqref{eqn:app_B6}) gives
\begin{align}
    \xi_i(\bm{x}) \, d\bm{x} = \sigma \cdot dW_i(\bm{x}). \label{eq:whitenoise_diff_form}
\end{align}
If $f:\mathcal{B}_R^{(n)} \to \mathbb{R}$ is an unknown bounded signal supported on a ball of radius $R$, then
\begin{align}
    h_{g_i}(\bm{x}) \triangleq (g_i \cdot f)(\bm{x})
\end{align}
denotes the rigid-motion transform of $f$. Each observation can then be expressed as
\begin{align}
    y_i(\bm{x}) \, d\bm{x} = h_{g_i}(\bm{x}) \, d\bm{x} + \sigma \cdot dW_i(\bm{x}), 
    \label{eq:obs_diff_form}
\end{align}
which may be interpreted as a stochastic partial differential equation: a rigidly transformed signal corrupted by spatial white noise.

\subsection{Wick's Theorem}
Let $\xi$ be a centered, zero-mean Gaussian process. A fundamental property of such processes is that all higher-order moments can be expressed entirely in terms of second-order moments. This is formalized by \emph{Wick's theorem} (also known as \emph{Isserlis' theorem}), which characterizes the expectation of products of jointly Gaussian random variables~\cite{isserlis1918formula, janson1997gaussian}. 

\begin{thm}[Wick's Theorem] \label{WicksTheorem}
Let $(X_1, \ldots, X_n)$ be a zero-mean multivariate Gaussian random vector. Then, the expectation of the product $X_1 X_2 \cdots X_n$ is given by
\begin{align}
    \mathbb{E}[X_1 X_2 \cdots X_n] 
    = \sum_{p \in \mathcal{P}_n^2} \prod_{\{i,j\} \in p} \mathbb{E}[X_i X_j] = \sum_{p \in \mathcal{P}_n^2} \prod_{\{i,j\} \in p} \mathrm{Cov}(X_i, X_j),
\end{align}
where the sum is taken over all pairwise partitions $\mathcal{P}_n^2$ of the index set $\{1, \ldots, n\}$, i.e., all distinct ways of dividing the indices into unordered pairs $\{i, j\}$. The product is taken over all such pairs in each partition $p$. This formula holds only when $n$ is even; for odd $n$, the expectation is zero.
\end{thm}

Wick's theorem expresses higher-order moments of a Gaussian vector as sums over products of covariances corresponding to pairings of indices. For odd $n$, no complete pairing exists, so the expectation vanishes. This result extends naturally to Gaussian processes, including Gaussian white noise.

\begin{corollary}
\label{lem:WickLemmaForWhiteNoise}
Let $\mathcal{D} \subset \mathbb{R}^n$ be bounded, and let $\xi$ be Gaussian white noise on $\mathcal{D}$.
Then, for any $k \in \mathbb{N}$, the following identities hold 
\begin{align}
    \mathbb{E}\!\left[\xi(\bm{x}_1)\cdots \xi(\bm{x}_{2k})\right] &= \sigma^{2k} \sum_{p \in \mathcal{P}_{2k}^2} \ \prod_{\{i,j\}\in p} \delta(\bm{x}_i - \bm{x}_j),
    \label{eq:white-noise-even}\\[4pt] \mathbb{E}\!\left[\xi(\bm{x}_1)\cdots \xi(\bm{x}_{2k+1})\right] &= 0.
    \label{eq:white-noise-odd}
\end{align}
Here, $\mathcal{P}_{2k}^2$ denotes the set of all pairings of $\{1,\dots,2k\}$, and $\delta(\bm{x})$ is the Dirac delta function.
\end{corollary}

\subsection{The multiple Wiener-It$\hat{\text{o}}$ integral}\label{sec:multiple-Wiener-ito-integral}
To formally define the integral in~\eqref{eq:stochastic_integral_multilinear}, we first need to introduce the multiple Wiener-It$\hat{\text{o}}$ integral.  
The single Wiener-It$\hat{\text{o}}$ integral $\int_{\mathbb{R}^n} \varphi(\bm{x}) \, dW(\bm{x})$, as defined in~\eqref{eqn:WienerItoIntegral}, admits a natural extension to the multiple Wiener-It$\hat{\text{o}}$ integral, which plays a central role in the Wiener chaos expansion~\cite{holden1996stochastic,geng2020wiener}.

\begin{definition}[$m$-tuple Wiener-It$\hat{\text{o}}$ integral] \label{def:m-tuple-Ito-integral}  
Let $(\mathcal{D}, \mathcal{B}(\mathcal{D}), \mu)$ be a measure space, where $\mathcal{D} \subset \mathbb{R}^n$ is a compact domain, $\mathcal{B}(\mathcal{D})$ denotes its Borel $\sigma$-algebra, and $\mu$ is the Lebesgue measure on $\mathcal{D}$. Let $(\Omega,\mathcal{F},\mathbb{P})$ be a probability space, and let $W$ be an $n$-parameter Gaussian white-noise random measure defined on this space.
For each integer $m \geq 1$, the $m$-tuple Wiener-It$\hat{\text{o}}$ integral is a linear map
\begin{align}
    I_m : L^2(\mathcal{D}^m) \to L^2(\Omega),
\end{align}
defined for any $f \in L^2(\mathcal{D}^m)$ by
\begin{align}
    I_m(f) \triangleq \int_{\mathcal{D}^m} f(\bm{t}_1, \dots, \bm{t}_m) \, dW(\bm{t}_1) \cdots dW(\bm{t}_m) \quad \in L^2(\Omega) .\label{eqn:itoIntegralMulti}
\end{align}
\end{definition}

This integral satisfies the following key properties:
\begin{enumerate}
    \item \textit{Linearity and symmetry:} The map $I_m$ is linear and invariant under permutations of its arguments; that is, for any permutation $\pi$ of $\{1,\dots,m\}$,
    \begin{align}
        I_m(f) = I_m(f^\pi),
    \end{align}
    where $f^\pi(\bm{t}_1, \dots, \bm{t}_m) \triangleq  f(\bm{t}_{\pi(1)}, \dots, \bm{t}_{\pi(m)})$.
    \item \textit{Isometry and orthogonality:} The multiple integrals satisfy the isometry relation
    \begin{align}
        \mathbb{E}[I_m(f)^2] = m! \, \|\tilde{f}\|_{L^2(\mathcal{D}^m)}^2,
    \end{align}
    where $\tilde{f}$ is the symmetrization of $f$. Moreover, for $f \in L^2(\mathcal{D}^{m_1})$ and $g \in L^2(\mathcal{D}^{m_2})$,
    \begin{align}
        \mathbb{E}[I_{m_1}(f) I_{m_2}(g)] = 
        \begin{cases}
            m_1! \, \langle \tilde{f}, \tilde{g} \rangle_{L^2(\mathcal{D}^{m_1})} & \text{if } m_1 = m_2, \\
            0 & \text{if } m_1 \neq m_2.
        \end{cases}
    \end{align}
\end{enumerate}

\subsection{Autocorrelations: Lebesgue construction}
\label{subsec:LebesgueAutocorr}

Recall that in the $\mathsf{SE}(n)$ rigid-motion model, the $d$-th order autocorrelation (Definition~\ref{def:autoCorrelationMomentsRigidMotion}) formally takes the form
\begin{align}
    \int_{\mathcal{D}}
    F(\bm{x}) \,
    \xi(\bm{x}+\bm{\tau}_0)\cdots
    \xi(\bm{x}+\bm{\tau}_{d-1}) \, d\bm{x},
    \label{eqn:app_c6}
\end{align}
for $F \in L^2(\mathcal{D})$ and fixed shifts $\bm{\tau}_0,\dots,\bm{\tau}_{d-1} \in \mathbb{R}^n$.
Here, $\xi$ denotes spatial white noise, satisfying $\mathbb{E}[\xi(\varphi)\,\xi(\psi)] = \sigma^2 \langle \varphi,\psi\rangle_{L^2(\mathbb{R}^n)}$ (as defined in~\eqref{eqn:def-white-noise-cov-structure}).
Since $\xi$ is not a function, pointwise products $\xi(\cdot)^d$ are not defined a priori.
To employ~\eqref{eqn:app_c6} in our statistical analysis, we must: (i) define the integral rigorously; (ii) establish that the integral in~\eqref{eqn:app_c6} is square-integrable with finite variance; (iii) justify exchanging expectation and integration; and (iv) relate the construction to Wick’s theorem for signal-noise decompositions.
This subsection develops the required framework.

The principal analytical difficulty in \eqref{eqn:app_c6} is that the integrand is supported on the shifted diagonal
\begin{align}
    \big\{(\bm{x}+\bm{\tau}_0,\dots,\bm{x}+\bm{\tau}_{d-1})
    : \bm{x}\in\mathcal{D}\big\}
    \subset \mathcal{D}^d,
\end{align}
which is not an $L^2$ kernel.  
To make the expression well defined, we introduce a regularization based on mollifiers, a standard tool for handling products of distributions; see, e.g., \cite{janson1997gaussian,holden1996stochastic,evans2022partial}.

\begin{definition}[Mollifier]
\label{def:mollifier}
A \emph{mollifier} is a smooth, compactly supported function
$\rho \in C_c^\infty(\mathbb{R}^n)$ such that $\int_{\mathbb{R}^n} \rho(\bm x)\, d\bm x = 1$. For $\varepsilon>0$, define the rescaled mollifier
\begin{align}
    \rho_\varepsilon(\bm z) = \varepsilon^{-n}\rho\left(\frac{\bm z}{\varepsilon}\right).
    \label{eqn:rho-epsilon}
\end{align}
Then, $\{\rho_\varepsilon\}_{\varepsilon>0}$ is an \emph{approximate identity}:
for every $\varphi \in C_c^\infty(\mathbb{R}^n)$,
\begin{align}
    \lim_{\varepsilon\to 0}
    \int_{\mathbb{R}^n} \rho_\varepsilon(\bm z)\,\varphi(\bm z)\,d\bm z
    = \varphi(\bm 0),
\end{align}
or equivalently, $\rho_\varepsilon(\bm{z}) \to \delta(\bm{z})$ in the sense of distributions.
\end{definition}

\paragraph{Regularized kernel and normalization.}
Extend $F$ by zero outside $\mathcal{D}$, so $F \in L^2(\mathbb{R}^n)$, and let $W$ be the $n$-parameter Gaussian white-noise random measure from Definition~\ref{def:nparamBM}.
For fixed shifts $\bm\tau_0,\dots,\bm\tau_{d-1}\in\mathbb{R}^n$,
define the \emph{regularized kernel}
\begin{align}
    (J_{\bm{\tau},\varepsilon}F)(\bm t_1,\dots,\bm t_d) \triangleq \int_{\mathbb{R}^n} F(\bm x) \prod_{j=1}^d \rho_\varepsilon \, \!\big(\bm t_j-(\bm x+\bm\tau_{j-1})\big)\, d\bm x.
    \label{eq:regularized-kernel}
\end{align}
Since $F\in L^2$ and $\rho_\varepsilon$ is smooth with compact support,
we have $J_{\bm{\tau},\varepsilon}F \in L^2((\mathbb{R}^n)^d)$.

We then define the regularized multiple Wiener–It$\hat{\text{o}}$ integral
\begin{align}
    I_d^{(\varepsilon)}(F)
    \triangleq \int_{(\mathbb{R}^n)^d} \widetilde{J_{\bm{\tau},\varepsilon}F}(\bm t_1,\dots,\bm t_d)\, dW(\bm t_1)\cdots dW(\bm t_d), \label{eqn:unnormalized-Ito-integral}
\end{align}
where $\widetilde{J_{\bm{\tau},\varepsilon}F}$ denotes the symmetrization of
$J_{\bm{\tau},\varepsilon}F$
(Definition~\ref{def:m-tuple-Ito-integral}).
This quantity is well defined by the standard theory of multiple
Wiener-It$\hat{\text{o}}$ integrals (Definition~\ref{def:m-tuple-Ito-integral}).

Next, to ensure a finite asymptotic variance of the integral in~\eqref{eqn:unnormalized-Ito-integral} as $\varepsilon \to 0$, we must normalize it.
Set $\tilde{\rho}_\varepsilon(\bm z)=\rho_\varepsilon(-\bm z)$ and $k_\varepsilon = \rho_\varepsilon \ast \tilde{\rho}_\varepsilon$.
Define
\begin{align}
    m_d(\varepsilon) \triangleq  \int_{\mathbb{R}^n}[k_\varepsilon(\bm z)]^d\, d\bm z = \varepsilon^{-n(d-1)}\, m_d(1),
\end{align}
where we have defined
\begin{align}
    m_d(1) \triangleq \int_{\mathbb{R}^n} [(\rho \ast \tilde{\rho})(\bm z)]^d\, d\bm z
    \in (0,\infty).
\end{align}
Define the normalized integral
\begin{align}
    \widetilde{I}_d^{(\varepsilon)}(F)
    \triangleq m_d(\varepsilon)^{-1/2}\, I_d^{(\varepsilon)}(F).
    \label{eq:regularized-multi}
\end{align}
The scaling $m_d(\varepsilon)^{-1/2}$ ensures a finite asymptotic variance, as would be shown in the proof of Proposition~\ref{prop:limit-shifted} in Appendix~\ref{sec:proofOfPropLimitShifted}. We are now in position to state the proposition.

\begin{proposition}
\label{prop:limit-shifted}
Let $(\Omega,\mathcal{F},\mathbb{P})$ be the probability space, and let $W$ be an $n$-parameter Gaussian white-noise random measure defined on this space. Let $\mathcal{D} \subset \mathbb{R}^n$ be bounded and $F \in L^2(\mathcal{D})$, extended by zero outside $\mathcal{D}$. Fix shifts $\bm\tau_0,\dots,\bm\tau_{d-1} \in \mathbb{R}^n$.
Then:
\begin{enumerate}
    \item \emph{($L^2$–convergence).}
    The family $\{\widetilde{I}_d^{(\varepsilon)}(F)\}_{\varepsilon>0}$ defined in~\eqref{eq:regularized-multi}
    is Cauchy in $L^2(\Omega)$, hence
    \begin{align}
        I_d(F) \triangleq \lim_{\varepsilon\to 0} \widetilde{I}_d^{(\varepsilon)}(F)
        \label{eqn:limiting-cauchy-sequence}
    \end{align}
    exists in $L^2(\Omega)$.
    
    \item \emph{(Independence of the mollifier).}
    The limit $I_d(F)$ is independent of the choice of mollifier
    $\rho\in C_c^\infty(\mathbb{R}^n)$ satisfying $\int\rho=1$.
    
    \item \emph{(Mean and variance).}
    By the It$\hat{\text{o}}$ isometry (Definition~\ref{def:m-tuple-Ito-integral}),
    \begin{align}
        \mathbb{E}[I_d(F)] &= 0, 
        \\ \mathbb{E}[I_d(F)^2] &= d!\,|F\|_{L^2(\mathcal{D})}^2.
    \end{align}
\end{enumerate}
\end{proposition}

Recalling $\xi(\bm x)\, d\bm x = \sigma\, dW(\bm x)$,
Proposition~\ref{prop:limit-shifted} provides a rigorous meaning to
the formal expression~\eqref{eqn:app_c6}:
\begin{align}
    \int_{\mathcal{D}} F(\bm x)\,
    \xi(\bm x+\bm\tau_0)\cdots \xi(\bm x+\bm\tau_{d-1})\, d\bm x
    \triangleq \sigma^{d}\, \lim_{\varepsilon\to 0} \widetilde{I}_d^{(\varepsilon)}(F) \triangleq \sigma^d I_d(F),
\end{align}
which is understood as a multiple Wiener–It$\hat{\text{o}}$ integral of order~$d$.
This construction connects the formal definition of autocorrelations with the stochastic analytic framework and underlies the proofs of Proposition~\ref{thm:propRigidMotionEquivalence} and subsequent statistical results.
In particular, the proposition establishes three essential properties:
\begin{enumerate}
    \item \emph{Well-defined:}
    $I_d(F)$ is an $L^2(\Omega)$ limit of smooth approximations.
    \item \emph{Finite variance:}
    The It$\hat{\text{o}}$ isometry yields explicit variance control, ensuring that empirical averages satisfy the law of large numbers.
    \item \emph{Gaussian structure:}
    As $I_d(F)$ lies in the $d$-th Wiener chaos, Wick–Isserlis expansions (Theorem~\ref{WicksTheorem},    Corollary~\ref{lem:WickLemmaForWhiteNoise}) decompose mixed signal–noise terms into delta function and lower-order autocorrelations.
\end{enumerate}

\subsection{Proof of Proposition~\ref{prop:limit-shifted}} \label{sec:proofOfPropLimitShifted}

We begin with an auxiliary result for Proposition~\ref{prop:limit-shifted}, which is a standard result adapted to our setting.

Let $\rho,\eta\in C_c^\infty(\mathbb{R}^n)$ be mollifiers (Definition~\ref{def:mollifier}) and define
\begin{align}
    \rho_\varepsilon(\bm z) = \varepsilon^{-n}\,\rho\!\left(\frac{\bm z}{\varepsilon}\right),\qquad
    \tilde{\eta}_{\varepsilon'}(\bm z) = (\varepsilon')^{-n}\,\eta\!\left(-\frac{\bm z}{\varepsilon'}\right). \label{eqn:app_A46}
\end{align}
For $d\in\mathbb{N}$, define the convolution between the mollifiers,
\begin{align}
    k_{\varepsilon,\varepsilon'} \triangleq \rho_\varepsilon \ast \tilde{\eta}_{\varepsilon'} \qquad m_d(\varepsilon,\varepsilon') &\triangleq  \int_{\mathbb{R}^n}\!\big[k_{\varepsilon,\varepsilon'}(\bm z)\big]^d\,d\bm z,
\end{align}
and define
\begin{align}
    H_{\varepsilon,\varepsilon'}(\bm z) & \triangleq  \frac{\big[k_{\varepsilon,\varepsilon'}(\bm z)\big]^d}{m_d(\varepsilon,\varepsilon')}. \label{eqn:H-ee}
\end{align}

\begin{lemma}\label{cl:Hee-scaling}
Recall the definition of $H_{\varepsilon,\varepsilon'}$ in~\eqref{eqn:H-ee}. Then,
\begin{enumerate}
    \item With $\lambda\triangleq \varepsilon/\varepsilon'$, there exists $h_\lambda\in L^1(\mathbb{R}^n)$ with $\int_{\mathbb{R}^n} h_\lambda = 1$, such that
    \begin{align}\label{eq:Hee-ai-statement}
        H_{\varepsilon,\varepsilon'}(\bm z) = (\varepsilon')^{-n}\,h_\lambda\!\left(\frac{\bm z}{\varepsilon'}\right).
    \end{align}

    \item The family $\{H_{\varepsilon,\varepsilon'}\}$ is a standard approximate identity, that is, for every $F\in L^2(\mathcal{D})$,
    \begin{align}
        \big\langle F,\ H_{\varepsilon,\varepsilon'} \ast F\big \rangle_{L^2(\mathcal{D})}     \xrightarrow[\varepsilon,\varepsilon'\to 0]{}  \|F\|_{L^2(\mathcal{D})}^2.
    \end{align}
\end{enumerate}
\end{lemma}

\begin{proof}[Proof of Lemma~\ref{cl:Hee-scaling}]
By definition of convolution,
\begin{align}
    k_{\varepsilon,\varepsilon'}(\bm z)
    &= \int_{\mathbb{R}^n} \rho_\varepsilon(\bm u)\,\tilde{\eta}_{\varepsilon'}(\bm z-\bm u)\,d\bm u .
\end{align}
Using the scaling of $\rho_\varepsilon$ and $\tilde{\eta}_{\varepsilon'}$ in~\eqref{eqn:app_A46},
\begin{align}
    k_{\varepsilon,\varepsilon'}(\bm z)
    &= \varepsilon^{-n}(\varepsilon')^{-n} \int_{\mathbb{R}^n}
    \rho\!\left(\frac{\bm u}{\varepsilon}\right)\,
    \eta\!\left(\frac{\bm u-\bm z}{\varepsilon'}\right) d\bm u .
\end{align}
Set $\bm u = \bm z + \varepsilon'\bm v$ so that $d\bm u = (\varepsilon')^n d\bm v$. Then,
\begin{align}
    k_{\varepsilon,\varepsilon'}(\bm z)
    &= \varepsilon^{-n} \int_{\mathbb{R}^n}
    \rho\!\left(\frac{\bm z + \varepsilon'\bm v}{\varepsilon}\right)\,\eta(\bm v)\, d\bm v .
\end{align}
Let $\lambda \triangleq  \varepsilon/\varepsilon'$. Since $\varepsilon = \lambda \varepsilon'$, we have
\begin{align}
    \frac{\bm z + \varepsilon'\bm v}{\varepsilon} = \frac{1}{\lambda}\left(\frac{\bm z}{\varepsilon'} + \bm v\right).
\end{align}
Therefore,
\begin{align}
    k_{\varepsilon,\varepsilon'}(\bm z)
    &= \varepsilon^{-n} \int_{\mathbb{R}^n}
    \rho\!\left( \frac{1}{\lambda}\left(\frac{\bm z}{\varepsilon'} + \bm v\right) \right)\,\eta(\bm v)\, d\bm v\nonumber \\
    &= (\varepsilon')^{-n}\,\lambda^{-n} \int_{\mathbb{R}^n}
    \rho\!\left( \frac{\bm z/\varepsilon' + \bm v}{\lambda} \right)\,\eta(\bm v)\, d\bm v \nonumber
    \\ & \triangleq (\varepsilon')^{-n}\, q_\lambda\!\left(\frac{\bm z}{\varepsilon'}\right), \label{eq:k-cross-scaled}
\end{align}
where in the last transition we have defined
\begin{align}
    q_\lambda(\bm w)
    &\triangleq  \lambda^{-n} \int_{\mathbb{R}^n}
    \rho\!\left( \frac{\bm w + \bm v}{\lambda} \right)\,\eta(\bm v)\, d\bm v.
\end{align}

From \eqref{eq:k-cross-scaled}, raise to the $d$-th power,
\begin{align}
    \big[k_{\varepsilon,\varepsilon'}(\bm z)\big]^d
    = (\varepsilon')^{-nd}\, \big[q_\lambda\!\left(\tfrac{\bm z}{\varepsilon'}\right)\big]^d . \label{eqn:Hee-num}
\end{align}
Integrating and changing variables $\bm w = \bm z/\varepsilon'$ (so $d\bm z = (\varepsilon')^n d\bm w$),
\begin{align}
    m_d(\varepsilon,\varepsilon')
    &= \int_{\mathbb{R}^n} \big[k_{\varepsilon,\varepsilon'}(\bm z)\big]^d\, d\bm z \nonumber\\
    &= (\varepsilon')^{-nd} \int_{\mathbb{R}^n} \big[q_\lambda(\bm z/\varepsilon')\big]^d\, d\bm z \nonumber\\
    &= (\varepsilon')^{-nd}\,(\varepsilon')^n \int_{\mathbb{R}^n} \big[q_\lambda(\bm w)\big]^d\, d\bm w \nonumber\\
    &= (\varepsilon')^{-n(d-1)}\, M_d(\lambda), \label{eqn:Hee-den}
\end{align}
where we have defined,
\begin{align}
    M_d(\lambda)
    \triangleq  \int_{\mathbb{R}^n} \big[q_\lambda(\bm w)\big]^d\, d\bm w
    \in (0,\infty).
\end{align}
Therefore, by the definition of $H_{\varepsilon,\varepsilon'}$ in~\eqref{eqn:H-ee}, in~\eqref{eqn:H-ee}, and substituting~\eqref{eqn:Hee-num} for the numerator and~\eqref{eqn:Hee-den} for the denominator, we obtain,
\begin{align}
    H_{\varepsilon,\varepsilon'}(\bm z)
    &= \frac{(\varepsilon')^{-nd}\,[q_\lambda(\bm z/\varepsilon')]^d}{(\varepsilon')^{-n(d-1)}\,M_d(\lambda)} \nonumber\\
    &= (\varepsilon')^{-n}\, h_\lambda\!\left(\frac{\bm z}{\varepsilon'}\right),
\end{align}
with
\begin{align}
    h_\lambda(\bm w)
    \triangleq \frac{[q_\lambda(\bm w)]^d}{M_d(\lambda)} .
\end{align}
By construction, $h_\lambda \in L^1(\mathbb{R}^n)$ and
\begin{align}
    \int_{\mathbb{R}^n} h_\lambda(\bm w)\, d\bm w = 1,
\end{align}
which proves \eqref{eq:Hee-ai-statement}.

Fix $\lambda>0$. The family $\{(\varepsilon')^{-n} h_\lambda(\cdot/\varepsilon')\}_{\varepsilon'>0}$ has unit mass and is uniformly $L^1$-bounded. If $\rho,\eta$ are compactly supported, then $h_\lambda$ is compactly supported uniformly in $\lambda$, and the supports shrink with $\varepsilon'$. Hence,  this family is a standard approximate identity, satisfying,
\begin{align}
    (\varepsilon')^{-n} h_\lambda(\cdot/\varepsilon') \ast F    \xrightarrow[\varepsilon'\to 0]{L^2(\mathcal{D})} F .
\end{align}
This completes the proof.
\end{proof}

\begin{proof} [Proof of Proposition~\ref{prop:limit-shifted}]
We prove the proposition in three steps.    

\paragraph{Step 1: Expressing $J_{\bm{\tau},\varepsilon}F$ as a convolution kernel.}
By the multiple Wiener–It$\hat{\text{o}}$ isometry (Definition~\ref{def:m-tuple-Ito-integral}),
\begin{align}
    \mathbb{E} \! \big[\,|I_d^{(\varepsilon)}(F)|^2\,\big]
    = d! \, \|J_{\bm{\tau} ,\varepsilon} F\|_{L^2((\mathbb{R}^n)^d)}^2.
\end{align}
Our aim is to show that as $\varepsilon\to 0$, 
\begin{align}
    \lim_{\varepsilon \to 0} \frac{\|J_{\bm\tau,\varepsilon}F\|_{L^2}^2}{m_d(\varepsilon)}   =   \|F\|_{L^2(\mathcal{D})}^2.
\end{align}
With $F$ extended by $0$ outside $\mathcal{D}$, recall by~\eqref{eq:regularized-kernel} that
\begin{align}
    (J_{\bm\tau,\varepsilon}F)(\bm t_1,\dots,\bm t_d)
    = \int_{\mathbb{R}^n} F(\bm x)\prod_{j=1}^d
    \rho_\varepsilon\!\big(\bm t_j-(\bm x+\bm\tau_{j-1})\big)\,d\bm x .
\end{align}
Then, by Fubini,
\begin{align}
    \nonumber \|J_{\bm\tau,\varepsilon}F\|_{L^2((\mathbb{R}^n)^d)}^2
    &= \int_{(\mathbb{R}^n)^d}\!\left|\int_{\mathbb{R}^n} F(\bm x)\prod_{j=1}^d
    \rho_\varepsilon\!\big(\bm t_j-(\bm x+\bm\tau_{j-1})\big)\,d\bm x\right|^2 d\bm t_1\cdots d\bm t_d \\
    &= \int_{(\mathbb{R}^n)^d}\int_{\mathbb{R}^n}\int_{\mathbb{R}^n}
    F(\bm x)F(\bm y)\prod_{j=1}^d
    \rho_\varepsilon\!\big(\bm t_j-(\bm x+\bm\tau_{j-1})\big)\,
    \rho_\varepsilon\!\big(\bm t_j-(\bm y+\bm\tau_{j-1})\big)\,
    d\bm x\,d\bm y\,d\bm t . \label{eqn:app_A49}
\end{align}
For each $j$, we use the convolution identity
\begin{align}
    \nonumber \int_{\mathbb{R}^n}\rho_\varepsilon(\bm t-\bm a)\,\rho_\varepsilon(\bm t-\bm b)\,d\bm t
    & = (\rho_\varepsilon \ast \tilde{\rho}_\varepsilon)(\bm a-\bm b)
    \\ & \triangleq k_\varepsilon(\bm a-\bm b),
\end{align}
where $\tilde{\rho}_\varepsilon(\bm z)=\rho_\varepsilon(-\bm z)$, with $\bm a=\bm x+\bm\tau_{j-1}$ and $\bm b=\bm y+\bm\tau_{j-1}$.
Since $\bm a-\bm b=\bm x-\bm y$, the product over $j$ yields $[k_\varepsilon(\bm x-\bm y)]^d$.
Thus, the expression in~\eqref{eqn:app_A49} simplifies to,
\begin{align}
    \|J_{\bm\tau,\varepsilon}F\|_{L^2((\mathbb{R}^n)^d)}^2
    = \iint_{\mathbb{R}^n\times\mathbb{R}^n}
    F(\bm x)F(\bm y)\,[k_\varepsilon(\bm x-\bm y)]^d\,d\bm x\,d\bm y .
\end{align}
Set $\bm z=\bm x-\bm y$ (so $\bm y=\bm x-\bm z$) and integrate first in $\bm x$ to get
\begin{align}
    \|J_{\bm\tau,\varepsilon}F\|_{L^2((\mathbb{R}^n)^d)}^2
    & = \int_{\mathbb{R}^n}\left(\int_{\mathbb{R}^n}F(\bm x)F(\bm x-\bm z)\,d\bm x\right)\,[k_\varepsilon(\bm z)]^d\,d\bm z
    \nonumber\\ & = \big\langle F,\ [k_\varepsilon]^d \ast F\big\rangle_{L^2(\mathbb{R}^n)}. \label{eqn:app-A80}
\end{align}

\paragraph{Step (2a): $L^2$-convergence.}
With the definition $\rho_\varepsilon(\bm z)=\varepsilon^{-n}\rho(\bm z/\varepsilon)$, 
we compute the convolution
$k_\varepsilon = \rho_\varepsilon \ast \tilde{\rho}_\varepsilon$,
where $\tilde{\rho}_\varepsilon(\bm z) \triangleq \rho_\varepsilon(-\bm z)$.
By the definition of convolution, and using the explicit form of $\rho_\varepsilon$ and $\tilde{\rho}_\varepsilon$, we obtain
\begin{align}
    k_\varepsilon(\bm x)
    &= \int_{\mathbb{R}^n} 
       \varepsilon^{-n}\rho\!\left(\frac{\bm u}{\varepsilon}\right)\,
       \varepsilon^{-n}\rho\!\left(-\,\frac{\bm x - \bm u}{\varepsilon}\right)\,d\bm u \nonumber\\
    &= \varepsilon^{-2n}\int_{\mathbb{R}^n}
       \rho\!\left(\frac{\bm u}{\varepsilon}\right)\,
       \rho\!\left(\frac{\bm u - \bm x}{\varepsilon}\right)\,d\bm u.
\end{align}
Now apply the change of variables $\bm v = \bm u/\varepsilon$, so that 
$\bm u = \varepsilon \bm v$ and $d\bm u = \varepsilon^n d\bm v$. Then
\begin{align}
    k_\varepsilon(\bm x)
    &= \varepsilon^{-2n}\,\varepsilon^n \int_{\mathbb{R}^n}
       \rho(\bm v)\,
       \rho\!\left(\bm v - \frac{\bm x}{\varepsilon}\right)\,d\bm v \nonumber\\
    &= \varepsilon^{-n}\int_{\mathbb{R}^n}
       \rho(\bm v)\,
       \rho\!\left(\bm v - \frac{\bm x}{\varepsilon}\right)\,d\bm v.
\end{align}
Define
\begin{align}
    k_1(\bm z) & \triangleq (\rho \ast \tilde{\rho})(\bm z) = \int_{\mathbb{R}^n} \rho(\bm v)\,\rho(\bm v - \bm z)\,d\bm v,
\end{align}
so that, for any $\bm x\in\mathbb{R}^n$, we can write
\begin{align}
    k_\varepsilon(\bm x)
    &= \varepsilon^{-n}\,k_1\!\left(\frac{\bm x}{\varepsilon}\right). \label{eqn:app-80}
\end{align}

Next, raising~\eqref{eqn:app-80} to the $d$-th power yields
\begin{align}
    [k_\varepsilon]^d(\bm x) 
    &= \big(\varepsilon^{-n}\,k_1(\bm x/\varepsilon)\big)^d
      =  \varepsilon^{-nd}\,\big[k_1(\bm x/\varepsilon)\big]^d,
\end{align}
that is,
\begin{align}
    [k_\varepsilon]^d 
    = \varepsilon^{-nd}\,[k_1(\cdot/\varepsilon)]^d.
    \label{eqn:app-A82}
\end{align}
Integrating~\eqref{eqn:app-A82} over $\mathbb{R}^n$ and again using the change of 
variables $\bm w = \bm x/\varepsilon$ (so $d\bm x = \varepsilon^n d\bm w$), we obtain
\begin{align}
    m_d(\varepsilon) 
    &\triangleq \int_{\mathbb{R}^n} [k_\varepsilon(\bm x)]^d\,d\bm x \nonumber\\
    &= \varepsilon^{-nd} \int_{\mathbb{R}^n} [k_1(\bm x/\varepsilon)]^d\,d\bm x \nonumber\\
    &= \varepsilon^{-nd}\,\varepsilon^n 
       \int_{\mathbb{R}^n} [k_1(\bm w)]^d\,d\bm w \nonumber\\
    &= \varepsilon^{-n(d-1)}\,m_d(1),
    \label{eqn:app-A83}
\end{align}
where we have defined,
\begin{align}
    m_d(1) 
    &\triangleq \int_{\mathbb{R}^n} [k_1(\bm w)]^d\,d\bm w.
\end{align}

Therefore, combining~\eqref{eqn:app-A82},~\eqref{eqn:app-A83}, results,
\begin{align}
    \frac{[k_\varepsilon]^d}{m_d(\varepsilon)}
    = \varepsilon^{-n}\,h(\cdot/\varepsilon),
\end{align}
where we have defined
\begin{align}
    h \triangleq \frac{[k_1]^d}{m_d(1)}
    = \frac{[(\rho \ast \tilde{\rho})]^d}{\int_{\mathbb{R}^n} [(\rho \ast \tilde{\rho})]^d}
    \in L^1.
\end{align}
In particular, we note that $\int_{\mathbb{R}^n} h = 1$. Thus,  $[k_\varepsilon]^d/m_d(\varepsilon)$ is a standard approximate identity, that is,
\begin{align}
    \lim_{\varepsilon \to 0}     [k_\varepsilon (\cdot)]^d/m_d(\varepsilon) = \delta (\cdot). \label{eqn:app-A87}
\end{align}
Combining~\eqref{eqn:app-A80} with~\eqref{eqn:app-A87} results,
\begin{align}\label{eqn:app_A_54}
    \frac{\|J_{\bm\tau,\varepsilon}F\|_{L^2}^2}{m_d(\varepsilon)}
    = \big\langle F,\ \tfrac{[k_\varepsilon]^d}{m_d(\varepsilon)} \ast F\big\rangle_{L^2(\mathbb{R}^n)} \longrightarrow  \|F\|_{L^2(\mathcal{D})}^2,
\end{align}
as $\varepsilon\to 0$, by the $L^2$–continuity of convolution with $L^1$ kernels and the approximate-identity property.
Consequently, for $\varepsilon\to 0$,
\begin{align}
    \mathrm{Var}\big(\widetilde{I}_d^{(\varepsilon)}(F)\big)
    = d! \, \frac{\|J_{\bm\tau,\varepsilon} F\|_{L^2}^2}{m_d(\varepsilon)}
     \longrightarrow  d!\,\|F\|_{L^2(\mathcal{D})}^2. \label{eqn:app_B68}
\end{align}

\paragraph{Step (2b): $L^2$–Cauchy property.}
Fix $\varepsilon,\varepsilon'>0$. 
By the $d$-fold Wiener–It$\hat{\text{o}}$ isometry (Definition~\ref{def:m-tuple-Ito-integral}),
\begin{align}\label{eq:cauchy-first}
    \mathbb{E}\!\left[\big(\widetilde{I}_d^{(\varepsilon)}(F)-\widetilde{I}_d^{(\varepsilon')}(F)\big)^2\right]
    = d!\,\left\|
    \frac{J_{\bm\tau,\varepsilon}F}{\sqrt{m_d(\varepsilon)}}-\frac{J_{\bm\tau,\varepsilon'}F}{\sqrt{m_d(\varepsilon')}}\right\|_{L^2((\mathbb{R}^n)^d)}^2 .
\end{align}
Define $A_\varepsilon \triangleq  \frac{J_{\bm\tau,\varepsilon}F}{\sqrt{m_d(\varepsilon)}}$.
Then, expanding the square in~\eqref{eq:cauchy-first},
\begin{align}
    \|A_\varepsilon-A_{\varepsilon'}\|_{L^2}^2
    = \|A_\varepsilon\|_{L^2}^2+\|A_{\varepsilon'}\|_{L^2}^2-2\langle A_\varepsilon,A_{\varepsilon'}\rangle_{L^2}.
\end{align}
We already showed in \eqref{eqn:app_A_54} that
\begin{align}\label{eq:var-limit}
    \|A_\varepsilon\|_{L^2}^2
    = \frac{\|J_{\bm\tau,\varepsilon}F\|_{L^2((\mathbb{R}^n)^d)}^2}{m_d(\varepsilon)}
    = \Big\langle F,\,\frac{[k_\varepsilon]^d}{m_d(\varepsilon)} \ast F\Big\rangle_{L^2(\mathbb{R}^n)}
     \longrightarrow  \|F\|_{L^2(\mathcal{D})}^2,
\end{align}
as $\varepsilon\to 0$, and the same holds with $\varepsilon'$ in place of $\varepsilon$.

It remains to handle the mixed inner product $\langle A_\varepsilon,A_{\varepsilon'}\rangle_{L^2}$.
Define the cross-kernel
\begin{align}
    k_{\varepsilon,\varepsilon'} \triangleq  \rho_\varepsilon \ast \tilde{\rho}_{\varepsilon'}.
\end{align}
A derivation analogous to \eqref{eqn:app_A49} yields
\begin{align}\label{eq:mixed-inner-product}
    \langle A_\varepsilon,A_{\varepsilon'}\rangle_{L^2} = \Big\langle F,\ \frac{[k_{\varepsilon,\varepsilon'}]^d}{\sqrt{m_d(\varepsilon) m_d (\varepsilon')}} \ast F \Big\rangle_{L^2(\mathbb{R}^n)} .
\end{align}
Set
\begin{align}
    H_{\varepsilon,\varepsilon'} \triangleq \frac{[k_{\varepsilon,\varepsilon'}]^d}{m_d(\varepsilon,\varepsilon')}, \qquad \text{where} \quad     m_d(\varepsilon,\varepsilon') \triangleq  \int_{\mathbb{R}^n}[k_{\varepsilon,\varepsilon'}(\bm z)]^d\,d\bm z.
\end{align}
By Lemma~\ref{cl:Hee-scaling}, there exists $h_\lambda\in L^1(\mathbb{R}^n)$ with $\int h_\lambda=1$, such that $\lambda = \varepsilon / \varepsilon'$, and
\begin{align}\label{eq:Hee-ai-use}
    H_{\varepsilon,\varepsilon'}(\bm z) = (\varepsilon')^{-n}\,h_\lambda\!\left(\frac{\bm z}{\varepsilon'}\right).
\end{align}
Hence, $\{H_{\varepsilon,\varepsilon'}\}$ is a standard approximate identity, and therefore
\begin{align}\label{eq:AI-limit}
    \Big\langle F,\,H_{\varepsilon,\varepsilon'} \ast F\Big\rangle_{L^2(\mathbb{R}^n)}
     \longrightarrow  \|F\|_{L^2(\mathcal{D})}^2,
\end{align}
as $\big((\varepsilon,\varepsilon')\to (0,0)\big)$.
Moreover, by the scaling relations in Lemma~\ref{cl:Hee-scaling},
\begin{align}\label{eq:mass-ratio}
    R(\varepsilon,\varepsilon') \triangleq  \frac{m_d(\varepsilon,\varepsilon')}{\sqrt{m_d(\varepsilon)\,m_d(\varepsilon')}}  \longrightarrow  1.
\end{align}
Combining \eqref{eq:mixed-inner-product}-\eqref{eq:mass-ratio},
\begin{align}\label{eq:mix-limit}
    \langle A_\varepsilon,A_{\varepsilon'}\rangle_{L^2} = R(\varepsilon,\varepsilon')\,\Big\langle F, \,H_{\varepsilon,\varepsilon'} \ast F\Big\rangle_{L^2} \longrightarrow  \|F\|_{L^2(\mathcal{D})}^2,
\end{align}
as $\big((\varepsilon,\varepsilon')\to (0,0)\big)$. Finally, inserting \eqref{eq:var-limit} and \eqref{eq:mix-limit} into \eqref{eq:cauchy-first} shows that the right-hand side converges to $0$ as $(\varepsilon,\varepsilon')\to (0,0)$.
Therefore $\{\widetilde{I}_d^{(\varepsilon)}(F)\}_{\varepsilon>0}$ is Cauchy in $L^2(\Omega)$.

\paragraph{(2) Independence of the mollifier.}
Let $\rho,\eta\in C_c^\infty(\mathbb{R}^n)$ be two mollifiers with associated kernels $k_\varepsilon^{(\rho)}$, $k_\varepsilon^{(\eta)}$ and normalizers $m_d^{(\rho)}(\varepsilon)$, $m_d^{(\eta)}(\varepsilon)$.
Define $k_{\varepsilon,\varepsilon'}^{(\rho,\eta)}\triangleq \rho_\varepsilon \ast \tilde{\eta}_{\varepsilon'}$ and $m_d^{(\rho,\eta)}(\varepsilon,\varepsilon')\triangleq \int_{\mathbb{R}^n}[k_{\varepsilon,\varepsilon'}^{(\rho,\eta)}]^d$.
By Lemma~\ref{cl:Hee-scaling}, the normalized cross-kernels $[k_{\varepsilon,\varepsilon'}^{(\rho,\eta)}]^d \, / \, {m_d^{(\rho,\eta)}(\varepsilon,\varepsilon')}$ form a standard approximate identity, and
\begin{align}
    \frac{m_d^{(\rho,\eta)}(\varepsilon,\varepsilon')}{\sqrt{m_d^{(\rho)}(\varepsilon)\,m_d^{(\eta)}(\varepsilon')}} \longrightarrow  1,
\end{align}
as $\big((\varepsilon,\varepsilon')\to (0,0)\big)$. The same arguments as in \eqref{eq:mixed-inner-product}-\eqref{eq:mix-limit} shows that the corresponding normalized cross inner products converge to $\|F\|_{L^2(\mathcal{D})}^2$, hence the $L^2$-limits coincide and $I_d(F)$ is independent of the mollifier.

\paragraph{(3) Mean and variance.}
Each $\widetilde{I}_d^{(\varepsilon)}(F)$ is centered (being a $d$-fold Wiener–It$\hat{\text{o}}$ integral), so the limit is centered.
Using the isometry and the limit in~\eqref{eqn:app_B68},
\begin{align}
    \mathbb{E}\big[I_d(F)^2\big] = d!\,\|F\|_{L^2(\mathcal{D})}^2,
\end{align}
which completes the proof of the proposition. 
\end{proof}

\section{Statistical analysis for orbit recovery under SE(n)} \label{sec:statistical-analysis-orbit-recoverry}
\subsection{Proof of Proposition \ref{thm:propRigidMotionEquivalence}}
\label{sec:proofOfpropRigidMotionEquivalence}

\paragraph{Convergence of the empirical autocorrelation.}
By the definition of the empirical $d$-th order autocorrelation in \eqref{eqn:autoCorrealtionMomentsRigidMotion}, we have
\begin{align}
    a^{(d)}_y(\bm{\tau}_0, \dots, \bm{\tau}_{d-1}) 
    &= \frac{1}{N} \sum_{i=0}^{N-1} \left( \int_{\bm{x} \in \mathcal{D}} \prod_{j=0}^{d-1} \tilde{y}_i(\bm{x} + \bm{\tau}_j) \, d\bm{x} \right). \label{eqn:app_B1}
\end{align}
The integral of~\eqref{eqn:app_B1} should be referred to as the limit of the mollified and normalized integral used in Proposition~\ref{prop:limit-shifted}, so that the integral in \eqref{eqn:app_B1} is an $L^2(\Omega)$ random variable, for the probability space $(\Omega, \mathcal{F}, \mathbb{P})$ defined in Appendix~\ref{sec:preliminariesToStatisticalPart}.

From the orbit recovery problem under $\mathsf{SE}(n)$, as described in Problem~\ref{prob:orbitRecoverySEn}, and from the construction of the integral in the right-hand-side of \eqref{eqn:app_B1}, as defined in Proposition~\ref{prop:limit-shifted}, the samples $\{\tilde{y}_i\}_{i=0}^{N-1}$ are i.i.d. with finite mean and variance. Therefore, by the strong law of large numbers (SLLN), the empirical average converges almost surely to its expected value as $N \to \infty$, yielding
\begin{align}
    a^{(d)}_y(\bm{\tau}_0, \dots, \bm{\tau}_{d-1}) 
    &\xrightarrow{\text{a.s.}} \mathbb{E}_{g, \xi} \left[ \int_{\bm{x} \in \mathcal{D}} \prod_{j=0}^{d-1} \tilde{y}_1(\bm{x} + \bm{\tau}_j) \, d\bm{x} \right], \label{eqn:app_B2}
\end{align}
where the expectation is over both the group action $g \in \mathsf{SE}(n)$, distributed according to $\mu_{\mathsf{SE}(n)}$ and over the white noise $\xi$.

\paragraph{Decomposition of the moment into signal and noise terms.}
Given the observation model $y_i = (g_i \cdot f) + \xi_i$, we define $f_{\mathcal{R}} = \mathcal{R} \cdot f$, where ${\mathcal{R}} \in \mathsf{SO}(n)$ is drawn according to the distribution $\rho$, and $f_g = g \cdot f$, where $g \in \mathsf{SE}(n)$ drawn according to the distribution $\mu_{\mathsf{SE}(n)}$. The following property holds:
\begin{align}
    \mathbb{E}_{g, \xi} \left[\int_{\bm{x} \in \mathcal{D}} \prod_{j=0}^{d-1} \tilde{y}_1(\bm{x} + \bm{\tau}_j) \, d\bm{x} \right]
    & =  \mathbb{E}_{g, \xi} \left[\int_{\bm{x} \in \mathcal{D}} \prod_{j=0}^{d-1} \pp{f_g(\bm{x} + \bm{\tau}_j) + \xi(\bm{x} + \bm{\tau}_j)}d\bm{x} \right] \nonumber
    \\ & = \mathbb{E}_{\mathcal{R} \sim \rho, \xi} \left[\int_{\bm{x} \in \mathcal{D}} \prod_{j=0}^{d-1} \pp{f_\mathcal{R}(\bm{x} + \bm{\tau}_j) + \xi(\bm{x} + \bm{\tau}_j)}d\bm{x} \right], \label{eqn:app_G4}
\end{align}
where \eqref{eqn:app_G4} follows from the translation invariance of the integral (as established in Proposition~\ref{prop:invarianceOfAutoCorrelation}) and the fact that Gaussian white noise is statistically invariant under spatial translations. Consequently, we may, without loss of generality, restrict our attention to signals centered at the origin, denoting $h_\mathcal{R} = f_\mathcal{R} + \xi$.
This allows us to write
\begin{align}
    \mathbb{E}_{\mathcal{R} \sim \rho, \xi} \left[\int_{\bm{x} \in \mathcal{D}} \prod_{j=0}^{d-1} \tilde{y}_1(\bm{x} + \bm{\tau}_j)d\bm{x}  \right]
    = \mathbb{E}_{\mathcal{R} \sim \rho, \xi} \left[\int_{\bm{x} \in \mathcal{D}} \prod_{j=0}^{d-1} h_\mathcal{R}(\bm{x} + \bm{\tau}_j)d\bm{x}  \right]. \label{eqn:app_D8}
\end{align}

As the unknown function $f$ is bounded on a bounded domain $\mathcal{D}$ (by its continuity from Assumption~\ref{assum:support}), all products involving $f$ are also bounded, thus the expectation and integration order can be interchanged by Fubini, since
\begin{align}
    \mathbb{E}\!\left[\int_{\mathcal{D}} \Big|\prod_{j=0}^{d-1} \tilde{y}_1(\bm{x}+\bm{\tau}_j)\Big| \, d\bm{x}\right]
    &\leq |\mathcal{D}|^{1/2}       \left(\mathbb{E}\!\int_{\mathcal{D}} \Big|\prod_{j=0}^{d-1} \tilde{y}_1(\bm{x}+\bm{\tau}_j)\Big|^2            d\bm{x}\right)^{1/2}
    < \infty, \label{eqn:app_C5}
\end{align}
where the right-hand side finiteness follows from the $L^2$ construction in Proposition~\ref{prop:limit-shifted}. 
Hence,
\begin{align}
    \mathbb{E}_{g, \xi} \left[ \int_{\bm{x} \in \mathcal{D}} \prod_{j=0}^{d-1} \tilde{y}_1(\bm{x} + \bm{\tau}_j) \, d\bm{x} \right] 
    = \int_{\bm{x} \in \mathcal{D}} \mathbb{E}_{g, \xi} \left[ \prod_{j=0}^{d-1} \tilde{y}_1(\bm{x} + \bm{\tau}_j) \right] d\bm{x}. \label{eqn:app_D5}
\end{align}

By interchanging the order of integration and expectation in~\eqref{eqn:app_D8}, expanding the product using multilinearity of expectation and the independence between the signal and the additive noise, we obtain
\begin{align}
    \mathbb{E}_{\mathcal{R} \sim \rho, \xi} \left[ \prod_{j=0}^{d-1} h_\mathcal{R}(\bm{x} + \bm{\tau}_j) \right]
    &= \mathbb{E}_{\mathcal{R} \sim \rho, \xi} \left[ \prod_{j=0}^{d-1} \left( f_\mathcal{R}(\bm{x} + \bm{\tau}_j) + \xi(\bm{x} + \bm{\tau}_j) \right) \right] \nonumber\\
    &= \sum_{S \subseteq \{0, \dots, d-1\}} \mathbb{E}_{\mathcal{R} \sim \rho} \left[ \prod_{j \in S} f_\mathcal{R}(\bm{x} + \bm{\tau}_j) \right] 
    \cdot \mathbb{E}_\xi \left[ \prod_{j \in S^c} \xi(\bm{x} + \bm{\tau}_j) \right], \label{eqn:app_D9}
\end{align}
where $S^c = \{0, \dots, d-1\} \setminus S$ denotes the complement of $S$.
Separating the purely signal-dependent term in \eqref{eqn:app_D9} from the rest, the expression becomes
\begin{align}
    & \mathbb{E}_{\mathcal{R} \sim \rho}  \left[ \prod_{j \in S} f_\mathcal{R}(\bm{x} + \bm{\tau}_j) \right] 
    \cdot \mathbb{E}_\xi \left[ \prod_{j \in S^c} \xi(\bm{x} + \bm{\tau}_j) \right] 
\nonumber    \\ & \quad = \mathbb{E}_{\mathcal{R} \sim \rho} \left[ \prod_{j=0}^{d-1} f_\mathcal{R}(\bm{x} + \bm{\tau}_j) \right]+ \sum_{S \subsetneq \{0, \dots, d-1\}} \mathbb{E}_{\mathcal{R} \sim \rho} \left[ \prod_{j \in S} f_\mathcal{R}(\bm{x} + \bm{\tau}_j) \right] 
    \cdot \mathbb{E}_\xi \left[ \prod_{j \in S^c} \xi(\bm{x} + \bm{\tau}_j) \right]. \label{eqn:app_D10}
\end{align}

Substituting \eqref{eqn:app_D5}-\eqref{eqn:app_D10} into the convergence statement in \eqref{eqn:app_B2} gives:
\begin{align}
    a^{(d)}_y(\bm{\tau}_0, \dots, \bm{\tau}_{d-1}) 
    &\xrightarrow{\text{a.s.}} \int_{\bm{x} \in \mathcal{D}} \mathbb{E}_{\mathcal{R} \sim \rho} \left[ \prod_{j=0}^{d-1} f_\mathcal{R}(\bm{x} + \bm{\tau}_j) \right] d\bm{x} \label{eqn:app_D12} \\
    &\quad +  \sum_{S \subsetneq \{0, \dots, d-1\}} \int_{\bm{x} \in \mathcal{D}} \mathbb{E}_{\mathcal{R} \sim \rho} \left[ \prod_{j \in S} f_\mathcal{R}(\bm{x} + \bm{\tau}_j) \right] 
    \cdot \mathbb{E}_\xi \left[ \prod_{j \in S^c} \xi(\bm{x} + \bm{\tau}_j) \right] d\bm{x}. \label{eqn:app_D13}
\end{align}

\paragraph{The population autocorrelation.}
Since the function $f_\mathcal{R}$ is compactly supported in $\mathcal{D}$, we may extend the domain of integration to $\mathbb{R}^n$. We identify the first term in \eqref{eqn:app_D12} with the population $d$-th order autocorrelation, as defined in Definition~\ref{def:autoCorrelationNoiseFree}:
\begin{align}
    \int_{\mathcal{D}} \mathbb{E}_{\mathcal{R} \sim \rho} \left[ \prod_{j=0}^{d-1} f_\mathcal{R}(\bm{x} + \bm{\tau}_j) \right] d\bm{x}
    &= \int_{\mathbb{R}^n} \mathbb{E}_{\mathcal{R} \sim \rho} \left[ \prod_{j=0}^{d-1} f_\mathcal{R}(\bm{x} + \bm{\tau}_j) \right] d\bm{x} \nonumber\\
    &= A^{(d)}_{f,\rho}(\bm{\tau}_0, \dots, \bm{\tau}_{d-1}) . \label{eqn:app_D15}
\end{align}

\paragraph{The noise term autocorrelation.}
According to Corollary~\ref{lem:WickLemmaForWhiteNoise}, the expectation of the product of independent, zero-mean Gaussian white noise terms in \eqref{eqn:app_D13} satisfies:
\begin{align}
    \mathbb{E}_\xi \left[ \prod_{j \in S^c} \xi(\bm{x} + \bm{\tau}_j) \right] = 
    \begin{cases}
        \sigma^{|S^c|} \sum_{\text{pairings}} \prod_{(i,j)} \delta(\bm{\tau}_i - \bm{\tau}_j), & \text{if } |S^c| \text{ is even}, \\
        0, & \text{otherwise}, \label{eqn:app_D16}
    \end{cases}
\end{align}
and this expression is independent of $\bm{x}$. Again noticing that the function $f_\mathcal{R}$ is compactly supported in $\mathcal{D}$, we may extend the domain of integration to $\mathbb{R}^n$, and the remaining integral in \eqref{eqn:app_D13} can be identified as a lower-order autocorrelation:
\begin{align}
    \mathbb{E}_{\mathcal{R} \sim \rho} \left[ \int_{\bm{x} \in \mathcal{D}}  \prod_{j \in S} f_\mathcal{R}(\bm{x} + \bm{\tau}_j) \, d\bm{x} \right]
    & = \mathbb{E}_{\mathcal{R} \sim \rho} \left[ \int_{\bm{x} \in \mathbb{R}^n}  \prod_{j \in S} f_\mathcal{R}(\bm{x} + \bm{\tau}_j) \, d\bm{x} \right]\nonumber
    \\ & = A^{(|S|)}_{f,\rho}(\{\bm{\tau}_j\}_{j \in S}). \label{eqn:app_D17}
\end{align}

Substituting \eqref{eqn:app_D16}-\eqref{eqn:app_D17} into \eqref{eqn:app_D13} results,
\begin{align}
    \sum_{S \subsetneq \{0, \dots, d-1\}} & \int_{\bm{x} \in \mathcal{D}} \mathbb{E}_{\mathcal{R} \sim \rho} \left[ \prod_{j \in S} f_\mathcal{R}(\bm{x} + \bm{\tau}_j) \right] 
    \cdot \mathbb{E}_\xi \left[ \prod_{j \in S^c} \xi(\bm{x} + \bm{\tau}_j) \right] d\bm{x}
\nonumber    \\ & = \sum_{\substack{S \subsetneq \{0, \dots, d-1\} \\ |S^c| \text{ even}}} \left( \sigma^{|S^c|} \cdot A^{(|S|)}_{f,\rho}(\{\bm{\tau}_j\}_{j \in S}) 
    \cdot \sum_{\text{pairings of } S^c} \prod_{(i,j)} \delta(\bm{\tau}_i - \bm{\tau}_j) \right). \label{eqn:noiseCorrectionFunction}
\end{align}
The last term in \eqref{eqn:noiseCorrectionFunction} is $ P^{(d)}_{f, \rho}(\bm{\tau}_0, \dots, \bm{\tau}_{d-1})$ by the definition in \eqref{eqn:noiseCorrectionFunctionMain}.

Substituting the expressions for the population \eqref{eqn:app_D15} and noise-induced \eqref{eqn:noiseCorrectionFunction} terms into \eqref{eqn:app_D13}, we conclude:
\begin{align}
    a^{(d)}_y(\bm{\tau}_0, \dots, \bm{\tau}_{d-1}) 
    &\xrightarrow{\text{a.s.}} A^{(d)}_{f, \rho}(\bm{\tau}_0, \dots, \bm{\tau}_{d-1}) + P^{(d)}_{f, \rho}(\bm{\tau}_0, \dots, \bm{\tau}_{d-1}), \label{eqn:empiricalAutoCorrLimit}
\end{align}
which completes the proof of the proposition.

\subsection{Proof of Proposition \ref{thm:sampleComplexityRigidMotion}}
\label{sec:proofOfsampleComplexityRigidMotion}
To establish this proposition, our goal is to show that the autocorrelations up to order $\bar{d}$ uniquely determine the orbit of the signal, with a sample complexity of $\omega(\sigma^{2\bar{d}})$. The key components of the proof are as follows:
\begin{enumerate}
    \item In Proposition~\ref{thm:propRigidMotionEquivalence}, we proved a point-wise convergence of the empirical autocorrelation. We now aim to extend this result to the convergence in $L^2((\mathbb{R}^n)^d)$ of the  empirical autocorrelation for all $\bm{\tau}_0, \dots, \bm{\tau}_{d-1} \in \mathcal{B}_R^{(n)}$.
    \item Even if the empirical autocorrelation converges to its ensemble mean at the rate $\omega(\sigma^{2\bar{d}})$, this alone does not guarantee convergence of the orbit estimator to the true orbit of the signal $f$ at the same rate. To establish this, we require certain additional mild conditions, which are formalized and proved in Lemma~\ref{lemma:D3}.
\end{enumerate}

Formally, let $y_i : \mathcal{D} \to \mathbb{R}$ for $i \in [N]$ denote observations generated according to the rigid motion model (Problem~\ref{prob:orbitRecoverySEn}). Let $a_y^{(d)}$ be the empirical $d$-th order autocorrelation as defined in~\eqref{eqn:autoCorrealtionMomentsRigidMotion}. Then, by Proposition~\ref{thm:propRigidMotionEquivalence}, we have:
\begin{align}
    a^{(d)}_y(\bm{\tau}_0, \dots, \bm{\tau}_{d-1}) 
    &\xrightarrow{\text{a.s.}} A^{(d)}_{f, \rho}(\bm{\tau}_0, \dots, \bm{\tau}_{d-1}) + P^{(d)}_{f, \rho}(\bm{\tau}_0, \dots, \bm{\tau}_{d-1}), \label{eqn:app_C1}
\end{align}
for every $\bm{\tau}_0, \dots, \bm{\tau}_{d-1} \in \mathcal{B}_R^{(n)}$, where $A^{(d)}_{f, \rho}$ is the population $d$-th order autocorrelation.
We define the right-hand-side of \eqref{eqn:app_C1}:
\begin{align}
    \bar{\mu}_{f, \rho}^{(d)}(\bm{\tau}_0, \dots, \bm{\tau}_{d-1}) 
    \triangleq A^{(d)}_{f, \rho}(\bm{\tau}_0, \dots, \bm{\tau}_{d-1}) + P^{(d)}_{f, \rho}(\bm{\tau}_0, \dots, \bm{\tau}_{d-1}). \label{eqn:app_A29}
\end{align}

We now present a series of results leading to the main proposition.

\begin{lem}\label{lemma:D1}
Fix $\bm{\tau}_0, \dots, \bm{\tau}_{d-1} \in \mathcal{B}_R^{(n)}$. Let $a_y^{(d)}(\bm{\tau}_0, \dots, \bm{\tau}_{d-1})$ denote the empirical $d$-th order autocorrelation evaluated at this point, and let $\bar{\mu}_{f, \rho}^{(d)}(\bm{\tau}_0, \dots, \bm{\tau}_{d-1})$ be defined as in~\eqref{eqn:app_A29}. Then, for every $\delta > 0$, 
\begin{align}
    \mathbb{P} \left\{ \left| a_y^{(d)}(\bm{\tau}_0, \dots, \bm{\tau}_{d-1}) - \bar{\mu}_{f, \rho}^{(d)}(\bm{\tau}_0, \dots, \bm{\tau}_{d-1}) \right| \geq \delta \right\}
    \leq O \left( \frac{\sigma^{2d}}{N \delta^2} \right). \label{eqn:A29}
\end{align}
\end{lem}

Before stating the next lemma, we define the $L^2$ norm between the empirical autocorrelation $a_y^{(d)}$ and  $\bar{\mu}_{f, \rho}^{(d)}$:
\begin{align}
    \norm{a_y^{(d)} - \bar{\mu}_{f, \rho}^{(d)}}_{L^2}^2 
    \triangleq \int_{\bm{\tau}_0, \ldots, \bm{\tau}_{d-1} \in \mathcal{B}_R^{(n)}}
    \left( a_y^{(d)}(\bm{\tau}_0, \dots, \bm{\tau}_{d-1}) - \bar{\mu}_{f, \rho}^{(d)}(\bm{\tau}_0, \dots, \bm{\tau}_{d-1}) \right)^2
    d \bm{\tau}_0 \ldots d \bm{\tau}_{d-1}. \label{eqn:app_E4}
\end{align}

\begin{lem}[$L^2$-convergence] \label{lemma:D2}
Let $a_y^{(d)}$ be the empirical $d$-th order autocorrelation, and $\bar{\mu}_{f, \rho}^{(d)}$ as defined in~\eqref{eqn:app_A29}. Then, for every $\delta > 0$,
\begin{align}
    \mathbb{P} \left\{ \norm{a_y^{(d)} - \bar{\mu}_{f, \rho}^{(d)}}_{L^2} \geq \delta \right\}
    \leq O \left( \frac{\sigma^{2d} \cdot (V_n(R))^d}{N \delta^2} \right). \label{eqn:A32}
\end{align}
\end{lem}

We now arrive at the main result of this proposition:

\begin{lem} [Sample-complexity of orbit recovery] \label{lemma:D3}
Assume that the population autocorrelations up to order $\bar{d}$ (Definition~\ref{def:autoCorrelationNoiseFree}) almost everywhere determine the orbit of the signal $f$ in Problem~\ref{prob:orbitRecoverySEn}. Suppose also that the parameter space $f \in \Theta$ is compact. Let $\mathrm{MSE}^{\ast}_{\mathsf{SE}(n)}(\sigma^2, N)$ denote the MSE as defined in~\eqref{eqn:infMtdMSE}. Then, if $N = \omega(\sigma^{2\bar{d}})$,
\begin{align}
    \lim_{N, \sigma \to \infty} \mathrm{MSE}^{\ast}_{\mathsf{SE}(n)}(\sigma^2, N) = 0. \label{eqn:app_E6}
\end{align}
\end{lem}

Lemma~\ref{lemma:D3} concludes the proof of the proposition. It remains to prove Lemma~\ref{lemma:D1}, Lemma~\ref{lemma:D2}, and Lemma~\ref{lemma:D3}.

\begin{proof}[Proof of Lemma~\ref{lemma:D1}]
By Chebyshev's inequality, for a fixed-point of the empirical autocorrelation corresponding to 
$\bm{\tau}_0, \bm{\tau}_1, \dots, \bm{\tau}_{d-1} \in \mathcal{B}_R^{(n)},$ 
we have:
\begin{align}
    \mathbb{P} & \left\{ \left| a^{(d)}_{y}(\bm{\tau}_0, \dots, \bm{\tau}_{d-1}) - \bar{\mu}_{f, \rho}^{(d)}(\bm{\tau}_0, \dots, \bm{\tau}_{d-1}) \right| \geq \delta \right\}
    \leq 
    \\ & \qquad \frac{\mathbb{E} \left[ \left( a^{(d)}_{y}(\bm{\tau}_0, \dots, \bm{\tau}_{d-1}) - \bar{\mu}_{f, \rho}^{(d)}(\bm{\tau}_0, \dots, \bm{\tau}_{d-1}) \right)^2 \right]}{\delta^2}. \label{eqn:A24}
\end{align}

According to Proposition~\ref{thm:propRigidMotionEquivalence}, and the isometry property established in Proposition~\eqref{prop:limit-shifted}, the variance of the empirical autocorrelation is bounded above by $O\left( \tfrac{\sigma^{2d}}{N} \right)$. Therefore, we have:
\begin{align}
    \mathbb{E} \left[ \left( a^{(d)}_{y}(\bm{\tau}_0, \dots, \bm{\tau}_{d-1}) - \bar{\mu}_{f, \rho}^{(d)} (\bm{\tau}_0, \dots, \bm{\tau}_{d-1})\right)^2 \right] = O\left( \frac{\sigma^{2d}}{N} \right). \label{eqn:app_E8}
\end{align}
Substituting the bound from~\eqref{eqn:app_E8} into~\eqref{eqn:A24} yields:
\begin{align}
    \mathbb{P} \left\{ \left| a^{(d)}_{y} - \bar{\mu}_{f, \rho}^{(d)} \right| \geq \delta \right\}
    \leq O\left( \frac{\sigma^{2d}}{N \delta^2} \right), \label{eqn:A26}
\end{align}
which completes the proof.

\end{proof}

\begin{proof}[Proof of Lemma~\ref{lemma:D2}]
To establish the sample complexity for the full empirical autocorrelation over the domain $\bm{\tau}_0, \bm{\tau}_1, \dots, \bm{\tau}_{d-1} \in \mathcal{B}_R^{(n)}$, we begin by applying Markov's inequality to bound the deviation of the empirical autocorrelation from its expectation:
\begin{align}
    \mathbb{P} \left\{ \norm{a^{(d)}_{y} - \bar{\mu}_{f, \rho}^{(d)}}_{L^2}^2 \geq \delta^2 \right\}
    \leq \frac{1}{\delta^2} \ \mathbb{E} \left[ \norm{a^{(d)}_{y} - \bar{\mu}_{f, \rho}^{(d)}}_{L^2}^2 \right]. \label{eqn:app_E10}
\end{align}

From the definition given in~\eqref{eqn:app_E4}, the $L^2$ norm can be expressed as an integral:
\begin{align}
    \mathbb{E} & \left[ \norm{a^{(d)}_{y} - \bar{\mu}_{f, \rho}^{(d)}}_{L^2}^2 \right]
    \nonumber\\ &= \mathbb{E} \left[ \int_{\bm{\tau}_0, \ldots, \bm{\tau}_{d-1} \in \mathcal{B}_R^{(n)}}
    \left( a^{(d)}_y(\bm{\tau}_0, \dots, \bm{\tau}_{d-1}) - \bar{\mu}_{f, \rho}^{(d)}(\bm{\tau}_0, \dots, \bm{\tau}_{d-1}) \right)^2
    \, d\bm{\tau}_0 \cdots d\bm{\tau}_{d-1} \right] \nonumber\\
    &= \int_{\bm{\tau}_0, \ldots, \bm{\tau}_{d-1} \in \mathcal{B}_R^{(n)}}
    \mathbb{E} \left[
    \left( a^{(d)}_y(\bm{\tau}_0, \dots, \bm{\tau}_{d-1}) - \bar{\mu}_{f, \rho}^{(d)}(\bm{\tau}_0, \dots, \bm{\tau}_{d-1}) \right)^2
    \right] d\bm{\tau}_0 \cdots d\bm{\tau}_{d-1}, \label{eqn:app_E12}
\end{align}
where the interchange of the expectation and the integral in~\eqref{eqn:app_E12} is due to the dominated convergence theorem (see~\eqref{eqn:app_C5}).

Now, substituting the variance bound from~\eqref{eqn:app_E8} into~\eqref{eqn:app_E12}, we obtain:
\begin{align}
    \mathbb{E} \left[ \norm{a^{(d)}_{y} - \bar{\mu}_{f, \rho}^{(d)}}_{L^2}^2 \right]
    = O\left( \frac{\sigma^{2d} \cdot \left( V_n(R) \right)^d}{N} \right), \label{eqn:app_E13}
\end{align}
where $V_n(R)$ is the volume of the radius-$R$ ball in $\mathbb{R}^n$.
Finally, substituting the bound~\eqref{eqn:app_E13} into~\eqref{eqn:app_E10} completes the proof.

\end{proof}

\begin{proof}[Proof of Lemma~\ref{lemma:D3}]
We reformulate the lemma as a likelihood estimation problem over the orbit $G \cdot f$, where $G = \mathsf{SO}(n)$, using the empirical autocorrelations up to order $\bar{d}$. This allows us to apply classical results from the literature on likelihood estimation, such as those found in~\cite{newey1994chapter}.

Under the assumptions of the proposition, the population autocorrelations $\ppp{A_{f,\rho}^{(d)}}_{d \leq \bar{d}}$, defined in Definition~\ref{def:autoCorrelationNoiseFree}, determine the orbit $G \cdot f$ almost everywhere.
Let $y_i : \mathcal{D} \to \mathbb{R}$, for $i \in \{0, \dots, N-1\}$, be i.i.d. samples from the model specified in Problem~\ref{prob:orbitRecoverySEn}. Let $\mathcal{Y}_N = \{y_i\}_{i=0}^{N-1}$ denote the dataset of observations. Define $a_y^{(d)}$ as the empirical $d$-th order autocorrelation of the observations, as in~\eqref{eqn:autoCorrealtionMomentsRigidMotion}.

Now, define the mapping $h: \mathcal{D}^N \to U$ from the observations $\mathcal{Y}_N$ to the collection of empirical autocorrelations up to order $\bar{d}$, where $U$ is the codomain:
\begin{align}
    h(\mathcal{Y}_N; G \cdot f) = \left(a_y^{(\bar{d})}, a_y^{(\bar{d}-1)}, \dots, a_y^{(1)}\right),
\end{align}
with the dependence on the orbit $G \cdot f$ made explicit through the distribution of $y$.

Let $\hat{Q}_N(\mathcal{Y}_N; G \cdot f)$ denote the distribution of $h(\mathcal{Y}_N; G \cdot f)$, which depends on the number of samples $N$. Then, the problem of estimating the orbit of $f$ reduces to solving the likelihood maximization problem:
\begin{align}
    \widehat{G \cdot f} = \argmax_{G \cdot f,  f \in \Theta} \hat{Q}_N(\mathcal{Y}_N; G \cdot f),
\end{align}
i.e., selecting the orbit that maximizes the likelihood of the observed autocorrelations.

We now verify the standard conditions for consistency of likelihood estimators:
\begin{enumerate}
    \item \textit{Identifiability}: By assumption, the orbit $G \cdot f$ is uniquely determined by the autocorrelations up to order $\bar{d}$. Hence, for any $f \neq f_0$,
    \begin{align}
        \hat{Q}_N(\mathcal{Y}_N; G \cdot f) \neq \hat{Q}_N(\mathcal{Y}_N; G \cdot f_0),
    \end{align}
    which ensures identifiability of the orbit.

    \item \textit{Convergence}: By Corollary~\ref{lemma:D2}, if $N = \omega(\sigma^{2\bar{d}})$, then as $N \to \infty$,
    \begin{align}
        \hat{Q}_N(\mathcal{Y}_N; G \cdot f) \xrightarrow{\text{a.s.}} Q_0(G \cdot f), \label{eqn:app_A40}
    \end{align}
    where $Q_0(G \cdot f)$ denotes the expected autocorrelations:
    \begin{align}
        Q_0(G \cdot f) = \left(\bar{\mu}_{f,\rho}^{(\bar{d})}, \dots, \bar{\mu}_{f,\rho}^{(1)}\right),
    \end{align}
    with $\bar{\mu}_{f,\rho}^{(d)}$ defined in~\eqref{eqn:app_A29}.

    \item \textit{Continuity}: Since the autocorrelations are polynomial functions of $f$, the mapping $f \mapsto \hat{Q}_N(\mathcal{Y}_N; G \cdot f)$ is continuous.

    \item \textit{Compactness}: The parameter space $\Theta$ is assumed to be compact.
\end{enumerate}

These four properties satisfy the conditions of~\cite[Theorem 2.1]{newey1994chapter}, which guarantees consistency of the maximum likelihood estimator. Therefore,
\begin{align}
    \widehat{G \cdot f} \xrightarrow{\mathcal{P}} G \cdot f. \label{eqn:A_41}
\end{align}
Since convergence in~\eqref{eqn:app_A40} holds for $N = \omega(\sigma^{2 \bar{d}})$, the consistency result~\eqref{eqn:A_41} holds under the same condition, which completes the proof.
\end{proof}

\subsection{Proof of Theorem \ref{thm:mainTheoremRigidmotion}}
\label{sec:proofOfmainTheoremRigidmotion}
The lower bound $\omega\left(\sigma^{2d}\right)$ is established in~\cite{balanov2025note}. We now prove the matching upper bound by presenting an explicit algorithmic reduction, given in Algorithm~\ref{alg:redcutiobFromRigidMotiontoSO}. Specifically, Algorithm~\ref{alg:redcutiobFromRigidMotiontoSO} describes how, in the asymptotic regime $N \to \infty$, one can extract $d$-th order $\mathsf{SO}(n)$ moment (Problem~\ref{prob:orbitRecoverySOn}) from the $\mathsf{SE}(n)$ empirical autocorrelations of orders up to $d+2$ (Problem~\ref{prob:orbitRecoverySEn}).

Specifically, Theorem~\ref{thm:reductionFromAutocorrelationToTensorMoment} (whose assumptions are satisfied under the conditions of the theorem) shows that the $d$-th order  $\mathsf{SO}(n)$ moment, $M_{f,\rho}^{(d)}$, can be recovered from the second-order and $(d+2)$-order population autocorrelations, $A^{(2)}_{f,\rho}$ and $A^{(d+2)}_{f,\rho}$.
By assumption, $d$ is the minimal $\mathsf{SO}(n)$ moment order required to uniquely determine the orbit of the signal $f$. Therefore, Theorem~\ref{thm:reductionFromAutocorrelationToTensorMoment} implies that the $(d+2)$-order population autocorrelation suffices to uniquely determine the orbit of $f$.

Furthermore, Proposition~\ref{thm:sampleComplexityRigidMotion} (which is applicable under our assumptions) asserts that if the population autocorrelations up to order $\bar{d}$ uniquely determine the orbit of the signal (in the sense of Definition~\ref{def:autoCorrelationNoiseFree}), then the sample complexity of orbit recovery in the $\mathsf{SE}(n)$ model is upper bounded by $\omega(\sigma^{2\bar{d}})$.
Since Theorem~\ref{thm:reductionFromAutocorrelationToTensorMoment} guarantees that the $(d+2)$-order population autocorrelation uniquely determines the orbit of $f$, applying Proposition~\ref{thm:sampleComplexityRigidMotion} with $\bar{d} = d+2$ yields the desired upper bound $\omega(\sigma^{2d+4})$.

\begin{algorithm}[t!]
  \caption{\texttt{Extraction of $\mathsf{SO}(n)$ moment from $\mathsf{SE}(n)$ empirical autocorrelation}
  \label{alg:redcutiobFromRigidMotiontoSO}}
  \textbf{Input:} Observations $\mathcal{Y}_N = \{y_i\}_{i \in [N]}$ drawn according to Problem~\ref{prob:orbitRecoverySEn}, where the underlying signal $f$ satisfies Assumptions~\ref{assum:support} and \ref{assum:nonVanishingSupport}. \\
\textbf{Output:} $d$-th order $\mathsf{SO}(n)$ moment, $M_{f, \rho}^{(d)}$ (Definition~\ref{def:mraTensorMoment}).
\begin{enumerate}
    \item Compute the empirical autocorrelations of the observations $\mathcal{Y}_N$ up to order $d+2$:
    \[
        \left\{ a_y^{(1)}, a_y^{(2)}, \dots, a_y^{(d+2)} \right\},
    \]
    as defined in Definition~\ref{def:autoCorrelationMomentsRigidMotion}.
    \item For $N \to \infty$, extract the population autocorrelations, as defined in  Definition~\ref{def:autoCorrelationNoiseFree},
    \[
        \left\{ A_{f, \rho}^{(1)}, A_{f, \rho}^{(2)}, \dots, A_{f, \rho}^{(d+2)} \right\}
    \]
    from the empirical autocorrelations using ~\eqref{eqn:propRigidMotionEquivalence} and~\eqref{eqn:noiseCorrectionFunctionMain} .
    \item Recover the $d$-th order $\mathsf{SO}(n)$ moment $M_{f, \rho}^{(d)}$ from the population autocorrelations $A_{f, \rho}^{(d+2)}$ and $A_{f, \rho}^{(2)}$, following the procedure in Theorem~\ref{thm:reductionFromAutocorrelationToTensorMoment}.
\end{enumerate}
\end{algorithm}

\section{Statistical analysis for MTD}
\label{sec:MTD-statistical-analysis}

\subsection{Proof of Proposition \ref{thm:prop0}} \label{sec:proofOfProp0}
Let $z \in [-MR, MR]^n$ be an observation sampled from the well-separated MTD model~\eqref{eqn:MTDmodelGeneral}. The locations ${\bm{x}_0, \bm{x}_1, \bm{x}_2, \ldots, \bm{x}_{N-1}}$ denote the centers of the signals $\ppp{f_i}_{i=0}^{N-1}$ within $z$, where each translation convolution operator is defined by the Dirac delta function $s_i(\bm{x}) = \delta(\bm{x} - \bm{x}_i)$ for $i \in \pp{N}$. According to the well-separated assumption (Definition~\ref{def:wellSeperatedModel}), the pairwise distances satisfy $|\bm{x}_{m_1} - \bm{x}_{m_2}| \geq 4R + \epsilon$ for all $m_1 \ne m_2$.

Fix $\bm{w} \in \mathcal{B}_{2R}^{(n)}$, i.e., a point in the $n$-dimensional ball of radius $2R$. Define the collection of signals $h_i^{(\bm{w})} : \mathcal{B}_R^{(n)} \to \mathbb{R}$, for $i \in [N]$, which represent disjoint regions of the observation $z$
\begin{align}
    h_i^{(\bm{w})} \p{\bm{x}} \triangleq z\p{\bm{x} - (\bm{x}_i -{\bm{w})}},
    \label{eqn:seqAi}
\end{align}
for $\bm{w} \in \mathcal{B}_R^{(n)}$. That is, each $h_i^{(\bm{w})}$ is supported on an $n$-dimensional ball of radius $R$, centered at $\bm{x}_i - \bm{w}$ in the observation $z$.

Then, we have the following straightforward lemma that characterizes the spatial disjointness of these signals. 
% (which is proved in Section \ref{sec:proofOfLemmaA1}):
%\dan{I don't think we need this Lemma and certainly not its proof, because it's essentially obvious. I think we can just remark this as we use it.}
    
\begin{lem} \label{lemma:A1} 
    Assume the well-separated MTD model (Definition \ref{def:wellSeperatedModel}). Then, for every $i_1 \neq i_2$, and for every $\bm{w} \in \mathcal{B}_{2R}^{(n)}$,  
    \begin{align}
         \mathcal{B}_{R}^{(n)} \p{\bm{x_{i_1}} + \bm{w}} \  \bigcap \ \mathcal{B}_{R}^{(n)} \p{\bm{x}_{i_2}} = \emptyset,  
    \end{align}
    and 
    \begin{align}
         \mathcal{B}_{R}^{(n)} \p{\bm{x_{i_1}} + \bm{w}} \  \bigcap \ \mathcal{B}_{R}^{(n)} \p{\bm{x}_{i_2} + \bm{w}} = \emptyset.
    \end{align}
\end{lem}

Lemma~\ref{lemma:A1} ensures that under the well-separated assumption, the domain of $h_{i_1}^{(\bm{w})}$ contains only the signal instance $f_{i_1}$ and does not overlap with any other signal $f_{i_2}$ for $i_2 \ne i_1$. Furthermore, the domains of $h_{i_1}^{(\bm{w})}$ and $h_{i_2}^{(\bm{w})}$ are disjoint for all $i_1 \ne i_2$. We next show that the sequence $\{h_i^{(\bm{w})}\}_{i \in \pp{N}}$ is i.i.d. for every fixed $\bm{w} \in \mathcal{B}_{2R}^{(n)}$.

\begin{lem} \label{lemma:A2} 
Assume the well-separated MTD model (Definition \ref{def:wellSeperatedModel}). Define the signal $\tilde{f}_\mathcal{R} : \mathcal{B}_{3R}^{(n)} \to \mathbb{R}$, as follows:
\begin{align}
    \tilde{f}_\mathcal{R}(\bm{x}) = f_\mathcal{R}(\bm{x}) \mathbf{1}_{\|\bm{x}\| \leq R} + 0 \cdot \mathbf{1}_{R < \bm{\|x\|} \leq 3R}, \label{eqn:A2}
\end{align}
which is the padding of $f_\mathcal{R} = \mathcal{R} \cdot f$ to a ball with radius $3R$.
Then, for every fixed  $\bm{w} \in \mathcal{B}_{2R}^{(n)}$, the sequence  
of signals ${h_i^{(\bm{w})}}$, for $i \in \pp{N}$, as defined in \eqref{eqn:seqAi}, is i.i.d., and follows the distribution of $h^{(\bm{w})}: \mathcal{B}_R^{(n)} \to \mathbb{R}$,
\begin{align}
    h^{(\bm{w})} \p{\bm{x}} = \tilde{f}_\mathcal{R} \p{\bm{x} + \bm{w}} + \xi \p{\bm{x} + \bm{w}}, \label{eqn:ArV}
\end{align}
where $\mathcal{R} \in \mathsf{SO}(n)$ is drawn according to the distribution law $\rho$, and $\xi : \mathcal{B}_{3R}^{(n)} \to \mathbb{R}$ is a white noise with variance $\sigma^2$. 
\end{lem}

To simplify notation, define for any $\bm{x} \in [-MR,MR]^n$
\begin{align}
    F \p{\bm{x}} 
    = z\p{\bm{x} + \bm{\tau}_0}z\p{\bm{x} + \bm{\tau}_1} \ldots z\p{\bm{x} + \bm{\tau}_{d-1}}
 = \prod_{j=0}^{d-1} z\p{\bm{x} + \bm{\tau}_j}, \label{eqn:app_C0}
\end{align}
where $\bm{\tau}_0, \bm{\tau}_1, \bm{\tau}_2, ..., \bm{\tau}_{d-1} \in \mathcal{B}_R^{(n)}$. With this definition, the $d$-th order empirical autocorrelation  \eqref{eqn:autoCorrealtionMoments} can be expressed as,
\begin{align}
     {a}^{(d)}_{z} (\bm{\tau}_0, \bm{\tau}_1, \bm{\tau}_2, ..., \bm{\tau}_{d-1}) 
    = \frac{1}{(2MR)^n} \int_{\bm{x} \in [-MR, MR]^n} F \p{\bm{x}} d \bm{x}, \label{eqn:autoCorrelationInFrep}
\end{align}
Moreover, from the definition of $h_i^{(\bm{w})}$ in~\eqref{eqn:seqAi}, for any $\bm{x} \in \mathcal{B}_{2R}^{(n)}$ and $\bm{\tau}_j \in \mathcal{B}_R^{(n)}$, we have
\begin{align}
    F \p{\bm{x} - \bm{x}_i} 
    &= \prod_{j=0}^{d-1} z\p{\bm{x} - \p{\bm{x}_i - \bm{\tau}_j}} 
    = \prod_{j=0}^{d-1} h_i^{(\bm{x})} (\bm{\tau}_j). \label{eqn:app_I9}
\end{align}

It is important to note that the values ${F(\bm{x})}$ are not generally i.i.d. over $\bm{x}$, so the strong law of large numbers cannot be applied directly to~\eqref{eqn:autoCorrelationInFrep}. To overcome this, we decompose the integral into a sum of independent integrals over disjoint regions, where the integrands are i.i.d., and apply the strong law of large numbers to each term separately.

\begin{lem} \label{lemma:A3}
    Assume the well-separated MTD model, and let $F(\bm{x})$ be defined as in~\eqref{eqn:app_C0}. Then,  for any fixed $\bm{x} \in \mathcal{B}_{2R}^{(n)}$, the sequence $\{F(\bm{x} - \bm{x}_i)\}_{i=0}^{N-1}$ is i.i.d., where $\{\bm{x}_i\}_{i=0}^{N-1}$ are the signal centers.
\end{lem}

As a consequence of Lemma~\ref{lemma:A3}, we have the following result:

\begin{lem} \label{lemma:A4} 
Assume the well-separated MTD model, and recall the definition of $F\p{\bm{x}}$ in \eqref{eqn:app_C0}. 
Then,  
    \begin{align}
        \frac{1}{(2MR)^n} & \int_{\bm{x} \in [-MR, MR]^n}  F \p{\bm{x}} d \bm{x}
        \xrightarrow{\s{a.s.}} \nonumber
        \\ &\frac{\gamma}{V_n\p{2R}} \int_{\bm{x} \in \mathcal{B}_{2R}^{(n)}} \mathbb{E}_{\mathcal{R} \sim \rho, \xi} \ppp{F\p{\bm{x} - \bm{x}_0}} \ d \bm{x} \label{eqn:app_A10}
         \\ & + \p{1 - \gamma} \mathbb{E} \pp{\xi \p{\bm{\tau}_0} \xi  \p{\bm{\tau}_1} ...\xi \p{\bm{\tau}_{d-1}}}\nonumber, \label{eqn:app_A11}
    \end{align}
as $N, M \to \infty$. 
\end{lem} 

The first term in~\eqref{eqn:app_A10} can be expressed in terms of the autocorrelation ensemble mean (Definition \ref{def:autoCorrelationMomentsEnsembele}):
\begin{lem} \label{lemma:A5} 
Assume the well-separated MTD model. Recall the definition of the autocorrelation ensemble mean $\bar{a}^{(d)}_{Y,\rho}$ (Definition \ref{def:autoCorrelationMomentsEnsembele}). Then, 
\begin{align}
    \frac{\gamma}{V_n\p{2R}} \int_{\bm{x} \in \mathcal{B}_{2R}^{(n)}} \mathbb{E}_{\mathcal{R} \sim \rho, \xi} \ppp{F\p{\bm{x} - \bm{x}_0}} \ d \bm{x}  = \gamma \cdot  \bar{a}^{(d)}_{Y,\rho} (\bm{\tau}_0, \bm{\tau}_1, \bm{\tau}_2, ..., \bm{\tau}_{d-1}).
\end{align}
\end{lem}
Combining~\eqref{eqn:autoCorrelationInFrep} with Lemmas~\ref{lemma:A4} and~\ref{lemma:A5} completes the proof. It remains to prove the lemmas. %It is left to prove the lemmas in the sequel.

\subsubsection{Proof of Lemma~\ref{lemma:A2}}
By the construction of the MTD model in \eqref{eqn:MTDmodelGeneral}, the observed stochastic process is generated by two independent components: (i) the group elements $\{g_i\}_{i=0}^{N-1} \subset G$, which are drawn independently from a common distribution $\rho$ and act on the signal $f$, and (ii) additive Gaussian white noise $\xi$. The signal $f$ and the locations $\{s_i\}_{i=0}^{N-1}$ are deterministic but unknown. Therefore, to analyze the statistical structure of the sequence $\{h_i^{(\bm{w})}\}_{i=0}^{N-1}$ defined in \eqref{eqn:seqAi}, it suffices to show that these are i.i.d. with respect to the randomness introduced by the group actions and the additive noise.

\paragraph{Independence of $h_i^{(\bm{w})}$.}
By Lemma~\ref{lemma:A1}, for every pair $i_1 \neq i_2$, the domains of $h_{i_1}^{(\bm{w})}$ and $h_{i_2}^{(\bm{w})}$ are disjoint. Consequently, these functions depend on disjoint regions of the observation $z$ and are affected by independent realizations of the additive white noise $\xi$ and group elements $g_{i_1}$ and $g_{i_2}$. Since, by definition of the MTD model, the group elements and the noise are independent across instances, it follows that $h_{i_1}^{(\bm{w})}$ and $h_{i_2}^{(\bm{w})}$ are independent.

\paragraph{Identical distribution.}
Let $z \in [-MR, MR]^n$ denote an observation from the well-separated MTD model \eqref{eqn:MTDmodelGeneral}. Let $\{\bm{x}_0, \bm{x}_1, \dots, \bm{x}_{N-1}\}$ denote the centers of the individual signal instances $\{f_i\}_{i=0}^{N-1}$ within $z$, such that $s_i(\bm{x}) = \delta(\bm{x} - \bm{x}_i)$ for each $i \in [N]$. 

Define the shifted functions $\{p_i\}_{i=0}^{N-1} : \mathcal{B}_{3R}^{(n)} \to \mathbb{R}$ by
\begin{align}
    p_i(\bm{x}) \triangleq z(\bm{x} - \bm{x}_i),
\end{align}
so that $p_i$ represents the local patch of the observation centered at the $i$-th instance location $\bm{x}_i$. That is, $p_i$ is supported on a ball of radius $3R$ centered at the origin and corresponds to the portion of $z$ around the $i$-th signal instance $f_i$.

Recall the signal function $\tilde{f}_\mathcal{R}$ defined in \eqref{eqn:A2}:
\begin{align}
    \tilde{f}_\mathcal{R}(\bm{x}) = f_\mathcal{R}(\bm{x}) \cdot \mathbf{1}_{\|\bm{x}\| \leq R} + 0 \cdot \mathbf{1}_{R < \|\bm{x}\| \leq 3R}.
\end{align}
Then, under the well-separated MTD model (Definition~\ref{def:wellSeperatedModel}), each $p_i$ is distributed as
\begin{align}
    p(\bm{x}) \triangleq \tilde{f}_\mathcal{R}(\bm{x}) + \xi(\bm{x}), \qquad \bm{x} \in \mathcal{B}_{3R}^{(n)}, \label{eqn:A22}
\end{align}
where $\mathcal{R} \sim \rho$ is a random group element of $\mathsf{SO}(n)$ and $\xi$ is additive Gaussian white noise. That is, $p_i$ represents a noisy observation of a randomly rotated version of the signal $f$. 

Now, by the definition of $h_i^{(\bm{w})}$ in \eqref{eqn:seqAi}, we have
\begin{align}
    h_i^{(\bm{w})}(\bm{x}) = p_i(\bm{x} + \bm{w}), \qquad \bm{x} \in \mathcal{B}_R^{(n)}. \label{eqn:A23}
\end{align}
Combining \eqref{eqn:A22} and \eqref{eqn:A23}, we see that $h_i^{(\bm{w})}$ is distributed as
\begin{align}
    h^{(\bm{w})}(\bm{x}) \overset{d}{=} \tilde{f}_g(\bm{x} + \bm{w}) + \xi(\bm{x} + \bm{w}), \qquad \bm{x} \in \mathcal{B}_R^{(n)},
\end{align}
which is precisely the distribution defined in \eqref{eqn:ArV}. Therefore, each $h_i^{(\bm{w})}$ is identically distributed.

\begin{remark}
Although the distributions of the centered patches $p_i$ are identical, the functions $\{p_i\}$ are not mutually independent, as they are extracted from overlapping regions of the MTD observation $z$ and thus share common noise components. In contrast, the functions $\{h_i^{(\bm{w})}\}$ are constructed from disjoint regions of $z$, ensuring that they are both identically distributed and mutually independent.
\end{remark}

\subsubsection{Proof of Lemma~\ref{lemma:A3}}
By the definition of $F(\bm{x})$ in \eqref{eqn:app_C0} and of $h_i^{(\bm{w})}$ in \eqref{eqn:seqAi}, we have the following identity, which holds for every $\bm{x} \in \mathcal{B}_{2R}^{(n)}$ and for every $\bm{\tau}_0, \bm{\tau}_1, \ldots, \bm{\tau}_{d-1} \in \mathcal{B}_R^{(n)}$:
\begin{align}
    F (\bm{x} - \bm{x}_i)
    &= \prod_{j=0}^{d-1} z(\bm{x} - (\bm{x}_i - \bm{\tau}_j))
    = \prod_{j=0}^{d-1} h_i^{(\bm{x})} (\bm{\tau}_j). \label{eqn:app_I34}
\end{align}

Lemma~\ref{lemma:A2} establishes that for any fixed $\bm{x} \in \mathcal{B}_{2R}^{(n)}$, the sequence of signals $\ppp{ h_i^{(\bm{x})} }_{i \in \pp{N}}$ is i.i.d. Consequently, since each term in the product on the right-hand side of \eqref{eqn:app_I34} depends on $h_i^{(\bm{x})}$ evaluated at different shifts $\bm{\tau}_j$ (all fixed), the product itself remains an i.i.d. sequence over the index $i$. That is, the mapping $i \mapsto F(\bm{x} - \bm{x}_i)$ defines an i.i.d. sequence for every $\bm{x} \in \mathcal{B}_{2R}^{(n)}$.
This observation directly implies that the left-hand side of \eqref{eqn:app_I34}, namely $\ppp{F(\bm{x} - \bm{x}_i)}_{i\in \pp{N}}$, forms an i.i.d. sequence. This completes the proof.

\subsubsection{Proof of Lemma~\ref{lemma:A4}}
We define the sets for $i \in \pp{N}$,
\begin{align}
    \mathcal{C}_M^{(i)} = \mathcal{B}_{2R}^{(n)}(\bm{x}_i),
\end{align} 
that is, the ball of radius $2R$ centered at $\bm{x}_i$. By the well-separated MTD model assumption, for all $i_1 \neq i_2$,
\begin{align}
    \mathcal{C}_M^{(i_1)} \cap \mathcal{C}_M^{(i_2)} = \emptyset.
\end{align}
Define
\begin{align}
    \mathcal{C}_M = \bigcup_{i=0}^{N-1} \mathcal{C}_M^{(i)},
\end{align}
which is the union of disjoint balls of radius $2R$ centered at the signal locations $\{\bm{x}_i\}_{i=0}^{N-1}$. Since these balls are disjoint, we have
\begin{align}
    \frac{|\mathcal{C}_M|}{(2MR)^n} = \frac{N \cdot V_n(2R)}{(2MR)^n} = \gamma,
\end{align}
where $|\mathcal{C}_M|$ denotes the volume of $\mathcal{C}_M$, and the density $\gamma$ is defined as in \eqref{eqn:densityMTDmodel}.

We decompose the domain of integration as
\begin{align}
    \int_{\bm{x} \in [-MR, MR]^n} F(\bm{x}) \, d\bm{x} 
    = \int_{\bm{x} \in \mathcal{C}_M} F(\bm{x}) \, d\bm{x}
    + \int_{\bm{x} \in \mathcal{C}_M^c} F(\bm{x}) \, d\bm{x}, \label{eqn:A28}
\end{align}
where $\mathcal{C}_M^c = [-MR, MR]^n \setminus \mathcal{C}_M$ is the complement of $\mathcal{C}_M$ in the observation domain.

Since $\mathcal{C}_M$ is the union of disjoint balls $\mathcal{C}_M^{(i)}$, we can write
\begin{align}
    \int_{\bm{x} \in \mathcal{C}_M} F(\bm{x}) \, d\bm{x}
    = \sum_{i=0}^{N-1} \int_{\bm{x} \in \mathcal{B}_{2R}^{(n)}} F(\bm{x} - \bm{x}_i) \, d\bm{x}. \label{eqn:app_I39}
\end{align}
Since $F(\bm{x})$ is a product of bounded functions, and $\bm{x} \mapsto F(\bm{x} - \bm{x}_i)$ is uniformly bounded over compact subsets of $\mathbb{R}^n$, the integrands are uniformly dominated by an integrable function. 
Then, by the dominated convergence theorem, we may interchange summation and integration:
\begin{align}
    \sum_{i=0}^{N-1} \int_{\bm{x} \in \mathcal{B}_{2R}^{(n)}} F(\bm{x} - \bm{x}_i) \, d\bm{x}
    = \int_{\bm{x} \in \mathcal{B}_{2R}^{(n)}} \sum_{i=0}^{N-1} F(\bm{x} - \bm{x}_i) \, d\bm{x}.
\end{align}
By Lemma~\ref{lemma:A3}, the sequence $\{F(\bm{x} - \bm{x}_i)\}_{i \in [N]}$ is i.i.d. for each fixed $\bm{x} \in \mathcal{B}_{2R}^{(n)}$. Hence, using the strong law of large numbers, we have
\begin{align}
    \frac{1}{(2RM)^n} \sum_{i=0}^{N-1} F(\bm{x} - \bm{x}_i)
    & = \frac{N}{(2RM)^n} \cdot \frac{1}{N} \sum_{i=0}^{N-1} F(\bm{x} - \bm{x}_i)
  \nonumber  \\ & \xrightarrow{\text{a.s.}} \frac{\gamma}{V_n(2R)} \, \mathbb{E}_{\mathcal{R} \sim \rho, \xi} \left[ F(\bm{x} - \bm{x}_0) \right]. \label{eqn:A17}
\end{align}

Since \eqref{eqn:A17} holds for all $\bm{x} \in \mathcal{B}_{2R}^{(n)}$, we again apply the dominated convergence theorem to obtain
\begin{align}
    \frac{1}{(2RM)^n} \int_{\bm{x} \in \mathcal{B}_{2R}^{(n)}} \sum_{i=0}^{N-1} F(\bm{x} - \bm{x}_i) \, d\bm{x}
    \xrightarrow{\text{a.s.}} \frac{\gamma}{V_n(2R)} \int_{\bm{x} \in \mathcal{B}_{2R}^{(n)}} \mathbb{E}_{\mathcal{R} \sim \rho, \xi} \left[ F(\bm{x} - \bm{x}_0) \right] \, d\bm{x}. \label{eqn:A30}
\end{align}

We now turn to the second term in \eqref{eqn:A28}:
\begin{align}
    \int_{\bm{x} \in \mathcal{C}_M^c} F(\bm{x}) \, d\bm{x}
    = \int_{\bm{x} \in \mathcal{C}_M^c} \xi(\bm{x} + \bm{\tau}_0) \xi(\bm{x} + \bm{\tau}_1) \cdots \xi(\bm{x} + \bm{\tau}_{d-1}) \, d\bm{x}.
\end{align}
The domain $\mathcal{C}_M^c$ has volume ratio $|\mathcal{C}_M^c| / (2RM)^n = 1 - \gamma$. As $M \to \infty$, using the strong law again:
\begin{align}
    \frac{1}{(2RM)^n} \int_{\bm{x} \in \mathcal{C}_M^c} \xi(\bm{x} + \bm{\tau}_0) \cdots \xi(\bm{x} + \bm{\tau}_{d-1}) \, d\bm{x}
    \xrightarrow{\text{a.s.}} (1 - \gamma) \, \mathbb{E} \left[ \xi(\bm{\tau}_0) \cdots \xi(\bm{\tau}_{d-1}) \right]. \label{eqn:A33}
\end{align}
Substituting \eqref{eqn:A30}, \eqref{eqn:app_I39}, and \eqref{eqn:A33} into \eqref{eqn:A28} concludes the proof.

\subsubsection{Proof of Lemma~\ref{lemma:A5}}
By the definition of $F(\bm{x})$ in \eqref{eqn:app_C0}, and the locations $\{ \bm{x}_i \}_{i=0}^{N-1}$, which denote the centers of the individual signal occurrences $f_i$, it follows that for every fixed $\bm{x} \in \mathcal{B}_{2R}^{(n)}$, the term $F(\bm{x} - \bm{x}_i)$ depends only on the random field $h_i^{(\bm{x})}$, as defined in \eqref{eqn:seqAi}. In particular, following from \eqref{eqn:app_I9}, we have
\begin{align}
    \mathbb{E}_{\mathcal{R} \sim \rho, \xi} \left[ F(\bm{x} - \bm{x}_0) \right] 
    = \mathbb{E}_{\mathcal{R} \sim \rho, \xi} \left[ \prod_{j=0}^{d-1} h_0^{(\bm{x})} (\bm{\tau}_j) \right]. \label{eqn:app_I48}
\end{align}

According to Lemma \ref{lemma:A2} and \eqref{eqn:ArV}, the random field $h_0^{(\bm{x})}$ satisfies
\begin{align}
    h^{(\bm{x})} (\bm{\tau}_j) = \tilde{f}_\mathcal{R} (\bm{\tau}_j + \bm{x}) + \xi (\bm{\tau}_j + \bm{x}), \label{eqn:app_I49}
\end{align}
where $\tilde{f}_\mathcal{R}$ is the padded signal associated with transformation $\mathcal{R} \sim \rho$, as defined in \eqref{eqn:A2}, and $\xi$ denotes additive white noise. 
Substituting \eqref{eqn:app_I49} into \eqref{eqn:app_I48} yields
\begin{align}
    \mathbb{E}_{\mathcal{R} \sim \rho, \xi} \left[ F(\bm{x} - \bm{x}_0) \right]
    &= \mathbb{E}_{\mathcal{R} \sim \rho, \xi} \left[ \prod_{j=0}^{d-1} \left( \tilde{f}_\mathcal{R}(\bm{\tau}_j + \bm{x}) + \xi(\bm{\tau}_j + \bm{x}) \right) \right]. \label{eqn:app_I51}
\end{align}
Similar to \eqref{eqn:ensembeleRV}, we define the random field $Y(\bm{x})$ as:
\begin{align}
    Y(\bm{x}) = \tilde{f}_\mathcal{R}(\bm{x}) + \xi(\bm{x}), \label{eqn:app_I52}
\end{align}
for $\bm{x} \in \mathcal{B}_{2R}^{(n)}$. Then, the expectation in \eqref{eqn:app_I51} can be expressed concisely in terms of $Y$:
\begin{align}
    \frac{\gamma}{V_n(2R)} & \int_{\bm{x} \in \mathcal{B}_{2R}^{(n)}} \mathbb{E}_{\mathcal{R} \sim \rho, \xi} \left[ F(\bm{x} - \bm{x}_0) \right] \, d\bm{x} \nonumber
    \\ &= \frac{\gamma}{V_n(2R)} \int_{\bm{x} \in \mathcal{B}_{2R}^{(n)}} \mathbb{E}_{\mathcal{R} \sim \rho, \xi} \left[ \prod_{j=0}^{d-1} Y(\bm{x} + \bm{\tau}_j) \right] d\bm{x}\nonumber \\
    &= \gamma \cdot \bar{a}^{(d)}_{Y,\rho} (\bm{\tau}_0, \bm{\tau}_1, \dots, \bm{\tau}_{d-1}),
\end{align}
where $\bar{a}^{(d)}_{Y,\rho}$ is the $d$-th order ensemble-mean autocorrelation of the field $Y$, as defined in \eqref{eqn:autoCorrealtionMomentsEnsembele}. This completes the proof.

\subsection{Proof of Theorem \ref{thm:mainTheorem}} \label{sec:proofSampleComplexityMTD}

The proof of this theorem closely follows that of Theorem~\ref{thm:mainTheoremRigidmotion}. The lower bound of $\omega\left(\sigma^{2d}\right)$
was established in~\cite{balanov2025note}. To establish the matching upper bound, we present an explicit algorithmic reduction, as outlined in Algorithm~\ref{alg:redcutiobFromMTDtoMRA}, which recovers the $d$-th order $\mathsf{SO}(n)$ moment $M_{f,\rho}^{(d)}$ (as defined in Definition~\ref{def:mraTensorMoment}) from the empirical autocorrelations of the MTD model up to order $d+2$, $\{a_y^{(k)}\}_{k=1}^{d+2}$ (Definition~\ref{def:autoCorrelationMomentsEmpirical}).

This reduction proceeds in three main stages:
\begin{enumerate}
    \item \textit{Autocorrelation equivalence:}  
    Under the assumption that $\gamma$ and $\sigma^2$ are known, Proposition~\ref{thm:prop0} shows that the ensemble-mean autocorrelations $\bar{a}_{Y,\rho}^{(d)}$ can be recovered from the empirical autocorrelations $a_{y}^{(d)}$ of the MTD observation $y$, in the limit as $M, N \to \infty$.

    \item \textit{Recovery of population autocorrelations:}  
    From the ensemble-mean autocorrelations $\{ \bar{a}_{Y,\rho}^{(k)} \}_{k=1}^d$, we can recover the population autocorrelations $\{ A_{f,\rho}^{(k)} \}_{k=1}^d$, as defined in Definition~\ref{def:autoCorrelationNoiseFree}, using the relation given in~\eqref{eqn:propRigidMotionEquivalence}.

    \item \textit{Extraction of $\mathsf{SO}(n)$ moments:}     Theorem~\ref{thm:reductionFromAutocorrelationToTensorMoment} asserts that, under current assumptions, the $d$-th order $\mathsf{SO}(n)$ moment $M_{f,\rho}^{(d)}$, can be recovered from the autocorrelations of the second and $(d+2)$-order population, $A_{f,\rho}^{(2)}$ and $A_{f,\rho}^{(d+2)}$. Since $d$ is the minimal moment order required for unique orbit recovery of $f$ in the $\mathsf{SO}(n)$ problem, it follows that the population autocorrelations $A_{f,\rho}^{(2)}$ and $A_{f,\rho}^{(d+2)}$ also suffice for orbit recovery in the MTD setting.
\end{enumerate}

Since the $(d+2)$-order autocorrelation suffices for orbit recovery, we invoke Proposition~\ref{thm:sampleComplexityRigidMotion} to obtain the desired upper bound $\omega\left(\sigma^{2d+4}\right)$, which completes the proof.

\begin{algorithm}[t!]
  \caption{\texttt{Extraction of $\mathsf{SO}(n)$ moments from MTD autocorrelations} \label{alg:redcutiobFromMTDtoMRA}}
\textbf{Input:} A well-separated MTD observation $z: \pp{-MR, MR}^n \to \mathbb{R}$ (Definition \ref{def:wellSeperatedModel}), satisfying Assumptions \ref{assum:support} and \ref{assum:nonVanishingSupport}. \\
\textbf{Output:} $d$-th order $\mathsf{SO}(n)$ moment $M_{f, \rho}^{(d)}$ (Definition \ref{def:mraTensorMoment}).
\begin{enumerate}
    \item Compute the empirical MTD autocorrelations of $z$ up to order $d+2$:  
          \begin{align}
            \ppp{{a}^{(1)}_{z}, {a}^{(2)}_{z}, \dots, {a}^{(d+2)}_{z}},
          \end{align}
          as defined in Definition \ref{def:autoCorrelationMomentsEmpirical}.
    \item From $\ppp{{a}^{(1)}_{z}, {a}^{(2)}_{z}, \dots, {a}^{(d+2)}_{z}}$, determine the corresponding ensemble mean MTD autocorrelations up to order $d+2$ (see Definition \ref{def:autoCorrelationMomentsEnsembele}):  
      \begin{align}
        \ppp{\bar{a}^{(1)}_{Y,\rho}, \bar{a}^{(2)}_{Y,\rho}, \dots , \bar{a}^{(d+2)}_{Y,\rho}},  
      \end{align}
    as $M \to \infty$, by applying Proposition \ref{thm:prop0} .
    
    \item Extract the population autocorrelations $\ppp{A^{(1)}_{f,\rho}, A^{(2)}_{f,\rho}, \dots , A_{f,\rho}^{(d+2)}}$ (Definition \ref{def:autoCorrelationNoiseFree}) up to order $d+2$ from the ensemble mean MTD autocorrelations $\ppp{\bar{a}^{(1)}_{Y,\rho}, \bar{a}^{(2)}_{Y,\rho}, \dots , \bar{a}^{(d+2)}_{Y,\rho}}$ (Proposition \ref{thm:propRigidMotionEquivalence}).  

    \item Retrieve the $d$-th order moment $M_{f, \rho}^{(d)}$ from the population autocorrelations $A^{(d+2)}_{f,\rho}$ and $A^{(2)}_{f,\rho}$ as described in Theorem \ref{thm:reductionFromAutocorrelationToTensorMoment}.  
\end{enumerate}
\end{algorithm}

\section{Numerical experiments in two dimensions}
\label{app:experimentalMethods2D}

This appendix outlines the numerical procedure used to reconstruct a band-limited image $f$, supported on a disk, from its invariants under the rigid-motion group $\mathsf{SE}(2)$ (see Section~\ref{sec:empirical} and Figures~\ref{fig:6}-\ref{fig:8} for an overview).  
Throughout, rotations are assumed to be uniformly distributed with respect to the Haar measure on $\SO(2)$, though similar procedures can be formulated for non-uniform orientation distributions as well. Consequently, we omit the dependence on $\rho$ in the notation and denote the corresponding autocorrelations and moments simply by $A_f^{(d)}$ and $M_f^{(d)}$, respectively.

The reconstruction pipeline proceeds through three main stages:
\begin{enumerate}
    \item \emph{Extraction of $\mathsf{SO}(2)$ moments from $\mathsf{SE}(2)$ autocorrelations.}  
    The fourth-order and fifth-order autocorrelation functions, $A_f^{(4)}$ and $A_f^{(5)}$, are first computed, followed by the extraction of the corresponding second-order and third-order moments, $M_f^{(2)}$ and $M_f^{(3)}$, as established in Theorem~\ref{thm:reductionFromAutocorrelationToTensorMoment}.  

    \item \emph{Intra-ring inversion from $\{M_f^{(2)}, M_f^{(3)}\}$.}  
    Each radial ring in the Fourier domain is reconstructed independently from its associated rotation-invariant moments $\{M_f^{(2)}, M_f^{(3)}\}$, yielding an estimate of both the magnitude and the phase structure within that ring, up to a global rotation.

    \item \emph{Inter-ring angular synchronization using $M_f^{(2)}$.}  
    Finally, the relative orientations between rings are recovered through pairwise phase synchronization based solely on the second-order moments $M_f^{(2)}$, producing a globally aligned reconstruction of the image.
\end{enumerate}

While the proposed procedure establishes a provable reconstruction framework, it is not designed to be computationally optimal or robust. The methods presented here serve primarily as a proof of concept, intended to demonstrate and validate the theoretical extraction of $\mathsf{SO}(2)$ moments from $\mathsf{SE}(2)$ autocorrelations, rather than to provide an algorithm optimized for large-scale implementation. In principle, a more direct approach operating on the image bispectrum itself could be developed, potentially offering improved numerical performance, as exemplified in~\cite{ma2019heterogeneous}.

\subsection{Extraction of SO(2) moments from SE(2) autocorrelations}
\label{subsec:numericalImplementation}

Theorem~\ref{thm:reductionFromAutocorrelationToTensorMoment} establishes that the $\SO(2)$ moments of a disk-supported function can be extracted from its $\SE(2)$ autocorrelations via the boundary limit~\eqref{eqn:mainTheoremExtraction}.  
In this subsection, we describe how this theoretical relation is implemented numerically and provide the computational details underlying the experiments presented in Section~\ref{sec:empirical}.

The numerical implementation proceeds in three main steps: (i) We first rewrite the relation of Theorem~\ref{thm:reductionFromAutocorrelationToTensorMoment} in terms of one-dimensional ring signals and their angular correlations.
(ii) Next, we discretize the continuous variables, i.e., translations, radii, and angular coordinates, obtaining discrete analogues of all quantities required for computation.  
(iii) Finally, we transform along the angular coordinate using the discrete Fourier transform (DFT), which yields an efficient representation and substantially reduces computational complexity through the use of the FFT algorithm. 

Concretely, for any collection of points $\bm{\eta}_1,\ldots,\bm{\eta}_d$ in the disk $\mathcal{B}_R \triangleq \mathcal{B}_R^{(2)}$, the $d$th-order $\SO(2)$ moment of $f$ is given according to Theorem~\ref{thm:reductionFromAutocorrelationToTensorMoment} by
\begin{align}
    M_{f}^{(d)}(\bm{\eta}_1, \dots, \bm{\eta}_d) 
    &= \lim_{\delta \to 0^+} \frac{\displaystyle \int_{S^{1}} A^{(d+2)}_{f}\!\left(\bm{\tau}_0^{(\delta)}(\theta), \bm{\tau}_1^{(\delta)}(\theta), \bm{\eta}_1, \dots, \bm{\eta}_d\right)\, d\theta} 
    {\displaystyle \int_{S^{1}} A^{(2)}_{f}\!\left(\bm{\tau}_0^{(\delta)}(\theta), \bm{\tau}_1^{(\delta)}(\theta)\right)\, d\theta},
    \label{eqn:mainTheoremExtraction2}
\end{align}
where the boundary points approaching antipodal positions are defined as
\begin{align}
    \bm{\tau}_0^{(\delta)}(\theta) = (R(1 - \delta), \theta), 
    \qquad
    \bm{\tau}_1^{(\delta)}(\theta) = (-R(1 - \delta), \theta),
\end{align}
with $\theta \in S^1$. In our experiments we restrict attention to the cases $d=2$ and $d=3$, which suffice for the complete recovery pipeline.

\subsubsection{Continuous setting}
We first present the extraction in the continuous domain and subsequently derive its discretized form required for numerical implementation.  
Let $ f : \mathcal{B}_R \subset \mathbb{R}^2 \to \mathbb{R} $ denote a function supported on a disk of radius $R$.  
For each translation center $\bm{t} \in \mathbb{R}^2$, define
\begin{align}
    s(\bm{t}) = \frac{1}{2\pi}\int_{0}^{2\pi} f \bigl( \bm{t} + \bm{\tau}_0(\theta)\bigr) \, f \bigl( \bm{t} + \bm{\tau}_1(\theta) \bigr)\, d\theta,
    \label{eqn:boundary-conditions}
\end{align}
where $\bm{\tau}_0(\theta)$ and $\bm{\tau}_1(\theta)$ are antipodal boundary offsets of radius $R(1-\delta)$.  
The factor $s(\bm{t})$ coincides with the denominator in~\eqref{eqn:mainTheoremExtraction2}, namely,
\begin{align}
    D = \int_{S^{1}} A^{(2)}_{f}\!\left(\bm{\tau}_0^{(\delta)}(\theta), \bm{\tau}_1^{(\delta)}(\theta)\right)\, d\theta = \int_{\mathbb{R}^2} s(\bm{t})\, d\bm{t}.
    \label{eqn:denominator-normalization}
\end{align}
For notational simplicity, we henceforth suppress the explicit dependence on $\delta$.

Fix $r \in [0, R]$. For each translation center $\bm{t} \in \mathbb{R}^2$, define the associated ring signal as the restriction of $f$ to the circle of radius $r$ centered at $\bm{t}$:
\begin{align}
    a^{(\bm{t})}_r(\varphi) = f\bigl(\bm{t} + r(\cos\varphi,\sin\varphi)\bigr),
    \qquad \varphi \in [0, 2\pi). 
    \label{eqn:ring-signals}
\end{align}
For any pair of radii $r_1, r_2 \in [0, R]$, the numerator in~\eqref{eqn:mainTheoremExtraction2} for the case $d = 2$ can be expressed as
\begin{align}
    \Gamma^{(2)}(r_1, r_2, \Delta\varphi) = \int_{\mathbb{R}^2} d\bm{t} \left[s(\bm{t}) \cdot \frac{1}{2\pi} \int_{0}^{2\pi} a^{(\bm{t})}_{r_1}(\varphi)\, a^{(\bm{t})}_{r_2}(\varphi+\Delta\varphi)\, d\varphi \right].
\end{align}
The extracted second-order moment is then obtained as
\begin{align}
    M_{\text{ext}}^{(2)}(r_1, r_2, \Delta\varphi) = \frac{\Gamma^{(2)}(r_1, r_2, \Delta\varphi)}{D},
\end{align}
where $D$ is defined in~\eqref{eqn:denominator-normalization}. The special case $r_1 = r_2 = r$ corresponds to \emph{intra-ring} correlations, whereas $r_1 \neq r_2$ represents \emph{inter-ring} correlations.

Analogously to the second-order case, the third-order numerator in~\eqref{eqn:mainTheoremExtraction2} 
for radii $r_1, r_2, r_3 \in [0, R]$ is given by
\begin{align}
  \nonumber \Gamma^{(3)}&(r_1, r_2, r_3, \Delta\varphi_1, \Delta\varphi_2)
   \\ & =  \int_{\mathbb{R}^2} d\bm{t} 
    \left[s(\bm{t}) \cdot \frac{1}{2\pi} \int_{0}^{2\pi} a^{(\bm{t})}_{r_1}(\varphi)\, a^{(\bm{t})}_{r_2}(\varphi+\Delta\varphi_1)\, a^{(\bm{t})}_{r_3}(\varphi+\Delta\varphi_2)\, d\varphi \right].
\end{align}
The normalization factor $D$ is the same as in~\eqref{eqn:denominator-normalization}, and the corresponding extracted third-order moment is defined as
\begin{align}
    M_{\text{ext}}^{(3)}(r_1, r_2, r_3, \Delta\varphi_1, \Delta\varphi_2)
    = \frac{\Gamma^{(3)}(r_1, r_2, r_3, \Delta\varphi_1, \Delta\varphi_2)}{D}.
\end{align}
By Theorem~\ref{thm:reductionFromAutocorrelationToTensorMoment}, as $ \delta \to 0^+ $, the extracted moments converge to the true rotation-invariant moments on $\mathsf{SO}(2)$, that is, $M_{\text{ext}}^{(d)} \to M_f^{(d)}$.

\subsubsection{Discrete setting}
We now describe the discretization of the continuous formulation used in the numerical implementation.  
In the continuous setting, the expectation in Theorem~\ref{thm:reductionFromAutocorrelationToTensorMoment} is taken over all translations $\bm{t} \in \mathbb{R}^2$.  
In practice, this integral is approximated using a finite translation grid $\Lambda = \{\bm{t}_j\}_{j=1}^{N_t} \subset \mathbb{R}^2$ that covers a region sufficiently large to contain the disk $\mathcal{B}_R$ supporting $f$. 

The variables are discretized as follows. The radius $r$ is sampled on a finite set of concentric rings $\{r_q\}_{q=1}^{R_{\max}} \subset [0, R]$;  the boundary angle $\theta$ is discretized into $N_\theta$ uniformly spaced samples of $[0, 2\pi)$; and the angular coordinate $\varphi$ on each ring is discretized into $N_\varphi$ uniformly spaced samples. The corresponding parameters are summarized in Table~\ref{tab:discretization}. The specific discretization values used in the numerical experiments are reported in Section~\ref{sec:empirical} of the main text.

\begin{table}[h]
    \centering
    \begin{tabular}{ll}
    \toprule
    Symbol & Description / Discretization \\
    \midrule
    $ \Lambda = \{\bm{t}_j\}_{j=1}^{N_t} $, $ N_t $ & Translation grid in $ \mathbb{R}^2 $ covering $ \mathcal{B}_R $ \\
    $ \{r_q\}_{q=1}^{R_{\max}} $ & Discrete set of radii; $ R_{\max} $ denotes the number of rings \\
    $ N_\theta $ & Boundary-angle samples on $ [0, 2\pi) $, $ \theta_k = 2\pi k / N_\theta $ \\
    $ N_\varphi $ & Ring-angle samples on $ [0, 2\pi) $, $ \varphi_\ell = 2\pi \ell / N_\varphi $ \\
    \bottomrule
    \end{tabular}
    \caption{Discretization parameters used in the numerical extraction procedure.}
    \label{tab:discretization}
\end{table}

The continuous boundary factor 
$s(\bm{t})$, defined in~\eqref{eqn:boundary-conditions}, 
is approximated in the discrete setting by
\begin{align}
    s_h(\bm{t}_j) = 
    \frac{2\pi}{N_\theta}\sum_{k=0}^{N_\theta-1} f\bigl(\bm{t}_j+\bm{\tau}_0(\theta_k)\bigr)\, f\bigl(\bm{t}_j+\bm{\tau}_1(\theta_k)\bigr),
    \qquad \theta_k = \tfrac{2\pi k}{N_\theta}.
    \label{eq:disc_s}
\end{align}
Here, the subscript $h$ denotes dependence on the discretizations. The integral over translations in the continuous formulation is replaced by a finite sum over the translation grid $\Lambda = \{\bm{t}_j\}_{j=1}^{N_t}$, which approximates
\begin{align}
  \int_{\mathbb{R}^2} g(\bm{t})\, d\bm{t}
   \ \approx\ 
  \sum_{j=1}^{N_t} g(\bm{t}_j).
\end{align}
Accordingly, the common denominator~\eqref{eqn:denominator-normalization} is approximated by
\begin{align}
    D_h = \sum_{j=1}^{N_t} s_h(\bm{t}_j).
    \label{eq:disc_D}
\end{align}

For each radius $r_q$ and translation center $\bm{t}_j$, the discretized ring signals defined in~\eqref{eqn:ring-signals} are given by
\begin{align}
    a^{(\bm{t}_j)}_{r_q}(\varphi_\ell)
    = f \bigl(\bm{t}_j + r_q(\cos\varphi_\ell,\sin\varphi_\ell)\bigr),
    \qquad
    \varphi_\ell = \tfrac{2\pi \ell}{N_\varphi},\ \ell = 0, \ldots, N_\varphi - 1.
    \label{eq:disc_ring}
\end{align}
Using~\eqref{eq:disc_s}-\eqref{eq:disc_ring}, the discrete second-order numerator is computed as
\begin{align}
    \Gamma^{(2)}_{h}(r_{q_1}, r_{q_2}, \Delta\varphi) = \sum_{j=1}^{N_t} s_h(\bm{t}_j)\,
    \Biggl[\frac{1}{N_\varphi}\sum_{\ell=0}^{N_\varphi-1}
    a^{(\bm{t}_j)}_{r_{q_1}}(\varphi_\ell)\,
    a^{(\bm{t}_j)}_{r_{q_2}}(\varphi_\ell+\Delta\varphi) \Biggr],
    \label{eq:N2_h_spatial}
\end{align}
and the discrete third-order numerator is
\begin{align}
    \nonumber \Gamma^{(3)}_{h}&(r_{q_1}, r_{q_2}, r_{q_3}, \Delta\varphi_1, \Delta\varphi_2) \\ & = \sum_{j=1}^{N_t}
    s_h(\bm{t}_j)\, 
    \Biggl[\frac{1}{N_\varphi}\sum_{\ell=0}^{N_\varphi-1}  a^{(\bm{t}_j)}_{r_{q_1}}(\varphi_\ell)\,   a^{(\bm{t}_j)}_{r_{q_2}}(\varphi_\ell+\Delta\varphi_1)\,   a^{(\bm{t}_j)}_{r_{q_3}}(\varphi_\ell+\Delta\varphi_2) \Biggr].
    \label{eq:N3_h_spatial}
\end{align}
The extracted discrete moments are then obtained as
\begin{align}
    M^{(2)}_{\text{ext}, h}(r_{q_1}, r_{q_2}, \Delta\varphi) = \frac{\Gamma^{(2)}_{h}(r_{q_1}, r_{q_2}, \Delta\varphi)}{D_h},
    \label{eq:M2_h_ratio}
\end{align}
and
\begin{align}
    M^{(3)}_{\text{ext}, h}(r_{q_1}, r_{q_2}, r_{q_3}, \Delta\varphi_1, \Delta\varphi_2)
    = \frac{\Gamma^{(3)}_{h}(r_{q_1}, r_{q_2}, r_{q_3}, \Delta\varphi_1, \Delta\varphi_2)}{D_h}.
    \label{eq:M3_h_ratio}
\end{align}

\subsubsection{Use of the fast Fourier transform}

Direct evaluation of the correlation sums for each angular lag $\Delta\varphi$ is computationally expensive, requiring $\mathcal{O}(N_\varphi^2)$ operations with respect to the number of angular samples $N_\varphi$.  
A more efficient approach exploits the Fourier representation of the ring signals and computes correlations using FFT, which reduces the computational complexity to $\mathcal{O}(N_\varphi \log N_\varphi)$.

Let $\mathcal{F}$ and $\mathcal{F}^{-1}$ denote the unitary DFT and its inverse acting on the angular grid $\{\varphi_\ell\}_{\ell=0}^{N_\varphi-1} = \{2 \pi \ell / N_\varphi\}_{\ell=0}^{N_\varphi -1}$.  
For a fixed translation $ \bm{t} $ and radius $r$, the unitary discrete Fourier coefficients of the ring signal $a^{(\bm{t})}_{r_q}$ are defined as
\begin{align}
    c^{(\bm{t})}_{r_q}[m] 
    &= (\mathcal{F}\,a^{(\bm{t})}_{r_q})[m] = \frac{1}{\sqrt{N_\varphi}}
    \sum_{\ell=0}^{N_\varphi-1}
     a^{(\bm{t})}_{r_q} \!\left(\tfrac{2\pi \ell}{N_\varphi}\right) e^{-i 2\pi m \ell / N_\varphi},
    \qquad m = 0, \ldots, N_\varphi - 1.
    \label{eq:cr_unitary}
\end{align}

By the discrete convolution theorem, the cyclic convolution of two ring signals satisfies
\begin{align}
    \frac{1}{N_\varphi}\sum_{\ell=0}^{N_\varphi-1}  a^{(\bm{t})}_{r_{q_1}}\!\left(\tfrac{2\pi \ell}{N_\varphi}\right) a^{(\bm{t})}_{r_{q_2}}\!\left(\tfrac{2\pi \ell}{N_\varphi} + \Delta\varphi\right)
    = \sum_{m=0}^{N_\varphi-1}
    c^{(\bm{t})}_{r_{q_1}}[m]\,
    c^{(\bm{t})}_{r_{q_2}}[-m]\,
    e^{i m \Delta\varphi}.
    \label{eq:WK_discrete_2}
\end{align}
Substituting~\eqref{eq:WK_discrete_2} into the spatial-domain expression~\eqref{eq:N2_h_spatial} and exchanging the order of summation yields the Fourier-domain form
\begin{align}
  \Gamma_h^{(2)}(r_{q_1}, r_{q_2}, \Delta\varphi)
   = \mathcal{F}^{-1}_{m \to \Delta\varphi}
  \Bigg(\sum_{j=1}^{N_t}
    s_h(\bm{t}_j)\,
    c^{(\bm{t}_j)}_{r_{q_1}}[m]\,
    c^{(\bm{t}_j)}_{r_{q_2}}[-m]\Bigg).
  \label{eq:N2_fft}
\end{align}

Analogously, the third–order correlation admits the discrete spectral representation
\begin{align}
    &\frac{1}{N_\varphi}\sum_{\ell=0}^{N_\varphi-1} a^{(\bm{t})}_{r_{q_1}}\!\left(\tfrac{2\pi \ell}{N_\varphi}\right) a^{(\bm{t})}_{r_{q_2}}\!\left(\tfrac{2\pi \ell}{N_\varphi}+\Delta\varphi_1\right) a^{(\bm{t})}_{r_{q_3}}\!\left(\tfrac{2\pi \ell}{N_\varphi}+\Delta\varphi_2\right)
    \nonumber\\[4pt]
    &\quad =  
    \sum_{m_1,m_2=0}^{N_\varphi-1}
        c^{(\bm{t})}_{r_{q_1}}[m_1]\,
        c^{(\bm{t})}_{r_{q_2}}[m_2]\,
        c^{(\bm{t})}_{r_{q_3}}[-(m_1+m_2)]\,
        e^{i(m_1\Delta\varphi_1 + m_2\Delta\varphi_2)}.
    \label{eq:WK_discrete_3}
\end{align}
Substituting~\eqref{eq:WK_discrete_3} into~\eqref{eq:N3_h_spatial} yields
\begin{align}
    \nonumber \Gamma_h^{(3)} & (r_{q_1}, r_{q_2}, r_{q_3}, \Delta\varphi_1, \Delta\varphi_2)
    \\ & = \mathcal{F}^{-1}_{(m_1,m_2)\to(\Delta\varphi_1,\Delta\varphi_2)}
    \Bigg(\sum_{j=1}^{N_t}
         s_h(\bm{t}_j)\,
         c^{(\bm{t}_j)}_{r_{q_1}}[m_1]\,
         c^{(\bm{t}_j)}_{r_{q_2}}[m_2]\,
         c^{(\bm{t}_j)}_{r_{q_3}}[-(m_1+m_2)]\Bigg).
    \label{eq:N3_fft}
\end{align}

%\begin{comment}   

Algorithm~\ref{alg:se2ToSo2Reduction} summarizes the complete extraction procedure.  
Steps~(1)-(2) compute the boundary function $s_h(\bm{t}_j)$ and their aggregate $D_h$, which provide the common denominator in~\eqref{eq:disc_D}.  
Steps~(3)-(4) construct the discrete ring signals and their Fourier coefficients as defined in~\eqref{eq:disc_ring} and~\eqref{eq:cr_unitary}.  
Steps~(5)-(6) form the spectral accumulators $\Gamma^{(2)}_h$ and $\Gamma^{(3)}_h$ according to~\eqref{eq:N2_fft} and~\eqref{eq:N3_fft}.  
Finally, Step~(7) applies inverse FFTs and normalizes by $D_h$, yielding the extracted discrete moments $M^{(2)}_{\text{ext}, h}$ and $M^{(3)}_{\text{ext}, h}$ as in~\eqref{eq:M2_h_ratio}–\eqref{eq:M3_h_ratio}.

\begin{algorithm}[]
  \caption{\texttt{Extraction of $\mathsf{SO}(2)$ moments from $\mathsf{SE}(2)$ autocorrelations}}
  \label{alg:se2ToSo2Reduction}
  \textbf{Input:} Image $f$ supported on $\mathcal{B}_R$; discretization parameters $(\Lambda = \{t_j \}_{j=1}^{N_t},\{r_q\}_{q=1}^{R_{\max}},N_\theta,N_\varphi)$. \\
  \textbf{Output:} Extracted discrete moments $\{M^{(2)}_{\text{ext}, h},\, M^{(3)}_{\text{ext}, h}\}$.
  \begin{algorithmic}[1]
    \State \emph{(Boundary sampling)} For each $\bm{t}_j \in \Lambda$ and $\theta_k = 2\pi k/N_\theta$, evaluate boundary offsets $\bm{\tau}_0(\theta_k),\,\bm{\tau}_1(\theta_k)$.
    \State \emph{(Boundary weights)} Compute $s_h(\bm{t}_j)$ by~\eqref{eq:disc_s} and  $D_h$ by~\eqref{eq:disc_D}.
    \State \emph{(Ring sampling)} For each $r_q$ and $\bm{t}_j$, sample the ring signals $a^{(\bm{t}_j)}_{r_q}(\varphi_\ell)$ via~\eqref{eq:disc_ring}.
    \State \emph{(Fourier coefficients)} Compute $c^{(\bm{t}_j)}_{r_q}[m] = (\mathcal{F}\,a^{(\bm{t}_j)}_{r_q})[m]$ using~\eqref{eq:cr_unitary}.
    \State \emph{(Spectral accumulator, $d=2$)} For each pair $(r_{q_1}, r_{q_2})$, form
      \begin{align}
        \Gamma^{(2)}_h(r_{q_1}, r_{q_2}, \Delta\varphi) = \mathcal{F}^{-1}_{m \to \Delta\varphi}
        \Biggl(\sum_{j=1}^{N_t}
          s_h(\bm{t}_j)\,
          c^{(\bm{t}_j)}_{r_{q_1}}[m]\,
          c^{(\bm{t}_j)}_{r_{q_2}}[-m] \Biggr),
        \label{eq:alg_N2}
      \end{align}
      cf.~\eqref{eq:N2_fft}.
    \State \emph{(Spectral accumulator, $d=3$)} Analogously, compute $\Gamma^{(3)}_h$ using~\eqref{eq:N3_fft}.
    \State \emph{(Normalization)} Return (cf.~\eqref{eq:M2_h_ratio}–\eqref{eq:M3_h_ratio})
      \begin{align}
        M^{(2)}_{\text{ext}, h} \,=\, \Gamma^{(2)}_h / D_h,
        \qquad
        M^{(3)}_{\text{ext}, h} \,=\, \Gamma^{(3)}_h / D_h.
      \end{align}
  \end{algorithmic}
\end{algorithm}
%\end{comment}

\subsection{Orbit recovery from SO(2) moments via bispectrum inversion}

\subsubsection{Per-ring inversion}
\label{sec:ringInversion}

Given the $\SO(2)$ moments $\{M_f^{(2)}, M_f^{(3)}\}$, we now describe how to recover each individual ring signal $a_r(\varphi_\ell)$, discretized on $N_\varphi$ angular samples by $\varphi_\ell = 2 \pi \ell / N_\varphi$.  
The procedure follows a variant of the classical bispectrum  pipeline~\cite{nikias1993signal,kakarala2012bispectrum,bendory2017bispectrum} and is conceptually analogous to the one-dimensional multi-reference alignment (MRA) method, where the bispectrum uniquely determines the signal up to a global rotation~\cite{bendory2017bispectrum}. In practice, we substitute $\{M_f^{(2)},M_f^{(3)}\}$ with their extracted discrete approximations $\{M^{(2)}_{\text{ext},h},M^{(3)}_{\text{ext},h}\}$ obtained from Algorithm~\ref{alg:se2ToSo2Reduction}, which introduces a small discretization error.

\paragraph{Step 1: Amplitude recovery from $M_f^{(2)}$.}
By construction (Algorithm~\ref{alg:se2ToSo2Reduction}, Steps~(3)-(4)), the coefficients $c_r[m]$ denote the unitary angular Fourier transform of the ring signal $a_r(\varphi_\ell)$ defined in~\eqref{eq:cr_unitary}. Since $M_f^{(2)}(r,\Delta\varphi)$ represents the circular autocorrelation of $a_r(\varphi_\ell)$, the discrete convolution theorem implies
\begin{align}
    \mathcal{F}\!\left\{M_f^{(2)}(r,\Delta\varphi)\right\}[m] = |c_r[m]|^2.
    \label{eq:amp_from_M2}
\end{align}
Thus, the amplitude spectrum $|c_r[m]|$ is obtained directly from the second-order moment.

\paragraph{Step 2: Normalized bispectrum.}
As shown in the third-order spectral representation~\eqref{eq:WK_discrete_3} and~\eqref{eq:N3_fft}, each frequency triplet $(m_1,m_2,-m_1-m_2)$ contributes
\begin{align}
    B_r[m_1,m_2] = c_r[m_1]\,c_r[m_2]\,c_r[-(m_1+m_2)].
    \label{eq:bispec_def}
\end{align}
Normalizing~\eqref{eq:bispec_def} by the amplitudes from~\eqref{eq:amp_from_M2} yields the normalized bispectrum
\begin{align}
    C_r[m_1,m_2] \triangleq \frac{B_r[m_1,m_2]}{|c_r[m_1]||c_r[m_2]||c_r[-(m_1+m_2)]|}
    =  e^{i(\theta_{m_1}+\theta_{m_2}+\theta_{-(m_1+m_2)})},
    \label{eq:norm_bispec}
\end{align}
where $c_r[m]=|c_r[m]|e^{i\theta_m}$.  
Each constraint~\eqref{eq:norm_bispec} is naturally associated with a weight
\begin{align}
    W_r[m_1,m_2] =  |c_r[m_1]||c_r[m_2]||c_r[-(m_1+m_2)]|,
    \label{eq:bispec_weights}
\end{align}
which reflects the relative reliability of the corresponding bispectral measurement.

\paragraph{Step 3: Phase recovery.}
The normalized bispectrum constraints~\eqref{eq:norm_bispec} couple the unknown phases $\{\theta_m\}$ across all frequencies. Since the signal is identifiable only up to a circular shift, we fix $\theta_1=0$.
Fix an integer $K \le \big\lfloor (N_\varphi - 1)/2 \big\rfloor$ and retain only the harmonics $|m| \le K$. For each $m=2,\dots,K$, the relation~\eqref{eq:norm_bispec} with $(m_1,m_2)=(j,m-j)$ implies, up to multiples of $2\pi$,
\begin{align}
    \arg C_r[j,m-j] = \theta_j + \theta_{m-j} - \theta_m,
    \qquad j=1,\dots,m-1.
    \label{eq:phase_constraint}
\end{align}
Each such constraint provides an estimate for $\theta_m$, which we combine using the  weights in~\eqref{eq:bispec_weights}.  
The practical update rule is a weighted circular mean:
\begin{align}
    \theta_m = \arg \, \!\Biggl(\sum_{j=1}^{m-1} W_r[j,m-j]\, \exp \Big(i\big(\theta_j+\theta_{m-j}-\arg C_r[j,m-j]\big)\Big)\Biggr).
    \label{eq:phase_update}
\end{align}
The recursion proceeds sequentially: $\theta_2$ is determined from $j=1$, then $\theta_3$ from $(j=1,2)$, and so on, propagating forward up to $K$.  
This sequential estimator is equivalent to a weighted least-squares fit on the unit circle, where~\eqref{eq:phase_update} provides a circular-mean estimate that performs well under moderate noise~\cite{nikias1993signal}. Finally, Hermitian symmetry imposes $\theta_{-m}=-\theta_m$, ensuring that the reconstructed ring signal is real-valued.

\paragraph{Step 4: Spectrum assembly.}
With amplitudes $|c_r[m]|$ obtained from the second-order moment~\eqref{eq:amp_from_M2} and phases $\theta_m$ recovered from the normalized bispectrum constraints~\eqref{eq:phase_update}, we reassemble the complex Fourier coefficients as
\begin{align}
    c_r[m] = |c_r[m]|\,e^{i\theta_m}, 
    \qquad |m|\leq K,
    \label{eq:coef_reassembly}
\end{align}
enforcing Hermitian symmetry $c_r[-m]=\overline{c_r[m]}$ to ensure a real-valued signal, while setting all higher modes to zero.  
Finally, applying the inverse FFT to~\eqref{eq:coef_reassembly} yields the reconstructed ring samples:
\begin{align}
    a_r(\varphi_\ell) = \frac{1}{\sqrt{N_\varphi}}
    \sum_{m=-K}^{K} c_r[m]\, e^{i m \varphi_\ell},
    \qquad \ell=0,\dots,N_\varphi-1,
    \label{eq:ring_recon}
\end{align}
for $\varphi_\ell = 2 \pi \ell / N_\varphi$.

\subsubsection{Inter-ring angular synchronization}
\label{sec:ringSync}

After per-ring inversion (Section~\ref{sec:ringInversion}), each reconstructed ring $\widehat{a}_r(\varphi)$ is determined only up to an unknown in-plane rotation $\theta_r \in [0,2\pi)$.  
We denote by $a_r(\varphi)$ the ground-truth ring at radius $r$ and $\widehat{a}_r(\varphi)$ for its estimate, related by
\begin{align}
    \widehat{a}_r(\varphi) = a_r(\varphi+\theta_r).
\end{align}
Let $c_r[k]$ and $\widehat{c}_r[k]$ denote the DFTs of $a_r(\cdot)$ and $\widehat{a}_r(\cdot)$, respectively, with $N_\varphi$ angular samples and harmonic indices $k\in\{0,1,\dots,N_\varphi-1\}$.  
The goal is to estimate the set $\{\theta_r\}_{r=1}^{R_{\max}}$, up to a global gauge, using only second-order information.  
Although each ring may be individually rotated, pairwise second-order moments computed from the full image preserve relative orientations across radii, which enables the synchronization between the rings.

\paragraph{Pairwise second-order model.}
For radii $i,j \in \{1 ,2, \ldots, R_{\max} \}$, define the pairwise moment
\begin{align}
    M^{(2)}_{ij}(\Delta\varphi) = \frac{1}{2\pi}\!\int_0^{2\pi} a_i(\varphi)\,a_j(\varphi+\Delta\varphi)\,d\varphi .
    \label{eq:def_M2ij}
\end{align}
Its DFT satisfies the correlation identity
\begin{align}
    \mathcal{F}\!\left\{M^{(2)}_{ij}\right\}[k]
    = c_i[k]\,\overline{c_j[k]},
    \qquad k=0,1,\dots,N_\varphi-1.
    \label{eq:pair_WK_discrete}
\end{align}
A rotation by $\theta_r$ acts diagonally in the Fourier domain:
\begin{align}
    \widehat{c}_r[k] = c_r[k]\,e^{-ik\theta_r}.
    \label{eq:rot_model}
\end{align}
Combining \eqref{eq:pair_WK_discrete} and \eqref{eq:rot_model}, and excluding the DC mode $k=0$, yields the normalized edge measurements
\begin{align}
    q_{ij}[k] \triangleq \frac{\widehat{c}_i[k]\,\overline{\widehat{c}_j[k]}} {\mathcal{F}\!\{M^{(2)}_{ij}\}[k]} = e^{-ik(\theta_i-\theta_j)},
    \qquad k=1,\dots,N_\varphi-1.
    \label{eq:sync_edge_measure}
\end{align}
Thus, $\arg(q_{ij}[k])$ encodes the relative offset $\theta_i-\theta_j$, scaled by $k$. The relations in \eqref{eq:sync_edge_measure} form an overdetermined system across $(i,j)$ pairs and harmonics. 

A convenient way to organize the relations in \eqref{eq:sync_edge_measure} is as a weighted measurement graph: the vertices correspond to the discrete radii, and each edge $(i,j)$ carries multi-frequency orientation information through the values $\{q_{ij}[k]\}_{k\ge 1}$. This graph encodes an overdetermined system for the unknown offsets ${\theta_i}$, and can be solved by standard angular-synchronization methods. Common approaches include semidefinite-programming relaxations~\cite{singer2011angular} and likelihood-refinement schemes initialized from a spectral solution~\cite{boumal2013robust}. In our implementation, we adopt a spectral method with multi-frequency pooling and data-driven edge weights, following~\cite{singer2011angular}.

\paragraph{Maximum-likelihood formulation.}
Let $G=(V,E)$ be the measurement graph with $V=\{1,\dots,R_{\max}\}$ and edges $(i,j)\in E$.  
We estimate $\theta=(\theta_1,\dots,\theta_{R_{\max}})$, up to a global gauge, by
\begin{align}
    \max_{\theta\in\mathbb{R}^{R_{\max}}}  \sum_{(i,j)\in E}  \sum_{k\in\mathcal{K}} w_{ij}[k] \Re \Big\{e^{ik(\theta_i-\theta_j)}\,\overline{q_{ij}[k]}\Big\},
    \qquad \theta_1=0,
    \label{eq:sync_MLE_multik}
\end{align}
where $\mathcal{K}\subset\{1,\dots,N_\varphi-1\}$ is the set of selected harmonics, $q_{ij}[k]$ are given by \eqref{eq:sync_edge_measure}, and $w_{ij}[k]\geq 0$ are weights.  
Writing $q_{ij}[k]=e^{-ik(\theta_i-\theta_j)}\,\eta_{ij}[k]$ with $|\eta_{ij}[k]|=1$ modeling phase noise, \eqref{eq:sync_MLE_multik} is a weighted least-squares estimator on the unit circle~\cite{singer2011angular,boumal2013robust}.  
In practice, we use
\begin{align}
    w_{ij}[k] \propto 
    \bigl|\mathcal{F}\{M^{(2)}_{ij}\}[k]\bigr|\,
    \bigl|\widehat{c}_i[k]\bigr|\,
    \bigl|\widehat{c}_j[k]\bigr|,
    \label{eq:weights_impl}
\end{align}
which assigns larger weights to harmonics with strong cross-power and large amplitudes.

\paragraph{Multi-frequency pooling.}
Because the phase of $q_{ij}[k]$ wraps $k$ times faster than $\theta_i-\theta_j$, harmonics must be pooled with phase de-scaling. We map each measurement to the principal branch by dividing the phase by $k$, such that the angles are reduced to $(-\pi,\pi]$ after the division,
\begin{align}
    \widehat{\delta\theta}_{ij}
    &=\arg\!\left(\sum_{k\in\mathcal{K}} w_{ij}[k]\, \exp\Bigl(-\tfrac{i}{k}\,\arg(q_{ij}[k])\Bigr)\right),
    \label{eq:sync_edge_pool_angle}\\
    W_{ij}
    &=\left|\sum_{k\in\mathcal{K}} w_{ij}[k] \, \exp \Bigl(-\tfrac{i}{k}\,\arg(q_{ij}[k])\Bigr)\right|.
    \label{eq:sync_edge_pool_weight}
\end{align}
We then assemble the Hermitian affinity
\begin{align}
    H_{ij} = 
    \begin{cases}         W_{ij}\,e^{i\widehat{\delta\theta}_{ij}}, & i\neq j,\\[2pt]
    0, & i=j,
    \end{cases}
    \qquad H_{ji}=\overline{H_{ij}},
    \label{eq:sync_H}
\end{align}
compute its leading eigenvector $\widehat{u}$, and extract synchronized angles as
\begin{align}
    \widehat{\theta}_r = \arg(\widehat{u}_r) - \arg(\widehat{u}_1),
    \qquad r=1,\dots,R_{\max},
    \label{eq:sync_spectral_solution}
\end{align}
which fixes the global gauge by setting $\widehat{\theta}_1=0$.  
Each spectrum is then rotated via $\widehat{c}_r[k]\leftarrow \widehat{c}_r[k]\,e^{-ik\widehat{\theta}_r}$, and the aligned rings are obtained by inverse DFT.  

\section{Numerical experiments in three dimensions}
\label{app:experimentalMethods3D}

This section formalizes the recovery of a three-dimensional volume from its $\mathsf{SE}(3)$ autocorrelations. As in the two-dimensional case, we assume throughout that rotations are distributed according to the Haar-uniform measure on $\mathsf{SO}(3)$. For notational simplicity, we omit the dependence on $\rho$ and denote the resulting $\mathsf{SE}(3)$ autocorrelations and $\mathsf{SO}(3)$ moments simply by $A_f^{(d)}$ and $M_f^{(d)}$, respectively.
The input consists of the second, fourth, and fifth-order autocorrelations, $A_f^{(2)}$, $A_f^{(4)}$, and $A_f^{(5)}$.
The reconstruction pipeline then proceeds in two main stages:
\begin{enumerate}
    \item \emph{Extraction of $\mathsf{SO}(3)$ moments from $\mathsf{SE}(3)$ autocorrelations.}
    Using the algebraic relations of Theorem~\ref{thm:reductionFromAutocorrelationToTensorMoment}, the second and third-order $\mathsf{SO}(3)$ moments, $M_f^{(2)}$ and $M_f^{(3)}$, are extracted from the autocorrelations $A_f^{(2)}$, $A_f^{(4)}$, and $A_f^{(5)}$.

    \item \emph{Frequency-marching reconstruction from $\{M_f^{(2)}, M_f^{(3)}\}$.}  
    The $\mathsf{SO}(3)$ moments, expressed in the spherical-harmonic (SH) basis, enable a linear, degree-by-degree recovery of the spherical-harmonic coefficients of $f$ through a frequency-marching procedure~\cite{bendory2025orbit}.
\end{enumerate}

\paragraph{Section structure.} Subsection~\ref{sec:radial-angular-rep-and-SH-form} introduces the representation of volume $f$ in spherical coordinates, and formulates the $\mathsf{SO}(3)$ second-order (Gram matrix) and third-order (bispectrum) moments in spherical-harmonic basis. 
Subsection~\ref{subsec:numericalImplementation3D} describes the practical extraction of the second-order and third-order moments of $\mathsf{SO}(3)$ from the fourth-order and fifth-order autocorrelations.  
Finally, subsection~\ref{subsec:freq_march} details the frequency-marching recovery: given the extracted $\mathsf{SO}(3)$ moments, we reconstruct the spherical-harmonic coefficients by solving iterative linear systems.

\subsection{Radial-angular representation and spherical harmonics form of moments} \label{sec:radial-angular-rep-and-SH-form}

\subsubsection{Radial-angular representation}
\label{subsec:radial_angular}

Let $f:\mathbb{R}^3 \to \mathbb{R}$ be a real-valued function supported in the unit ball $\mathcal{B} \triangleq  \{x \in \mathbb{R}^3 : \|x\| \le 1\}$.
In spherical coordinates $(r,\theta,\varphi) \in [0,1] \times [0,\pi] \times [0,2\pi)$, the volume $f$ admits an expansion in the normalized spherical-harmonic basis $\{Y_{\ell m}\}_{\ell \ge 0,  |m| \le \ell}$:
\begin{align}
    f(r,\theta,\varphi)
    &= \sum_{\ell=0}^{L_{\max}}
    \sum_{m=-\ell}^{\ell}
    F_{\ell m}(r)\,Y_{\ell m}(\theta,\varphi),
    \qquad
    \int_{S^2} \overline{Y_{\ell m}}(\omega)\,Y_{\ell' m'}(\omega)\,d\omega
    = \delta_{\ell\ell'}\delta_{mm'} .
    \label{eq:sph_expansion}
\end{align}
Here, $\ell \in \mathbb{N}$ is the \emph{angular degree}, $m \in \{-\ell,\ldots,\ell\}$ the \emph{azimuthal index}, and $L_{\max}$ the \emph{bandlimit}, i.e., the maximal angular degree present in the expansion. For each $(\ell,m)$, the function $F_{\ell m}:[0,1]\to\mathbb{R}$ specifies the corresponding radial profile.

\paragraph{Radial discretization by Gaussian shells.}
In our computational framework, we adopt a discrete variant based on a set of Gaussian-weighted radial shells, which are widely used in harmonic analysis and computational imaging ~\cite{freeden1998constructive, makadia2003direct, kazhdan2003rotation} since they are spatially localized in radius and permit simple and stable numerical quadrature.  
We approximate the radial coordinate using $R_{\max}$ localized \emph{Gaussian shells} centered at $\{r_q\}_{q=1}^{R_{\max}} \subset(0,1)$ with width $\sigma_{\mathrm{sh}}>0$:
\begin{align}
    W_q(r) &= C_q\exp\Big(-\tfrac{1}{2}\tfrac{(r-r_q)^2}{\sigma_{\mathrm{sh}}^2}\Big),  \qquad C_q^{-2} = \int_{0}^{1} \exp\Big(-\tfrac{(r-r_q)^2}{\sigma_{\mathrm{sh}}^2}\Big)\,r^2 \,dr ,
    \label{eq:W_def}
\end{align}
so that each $W_q$ has unit norm in $L^2([0,1],r^2dr)$. The overlap matrix is
\begin{align}
    S_{ab} \triangleq  \int_0^1 W_a(r)\,W_b(r)\,r^2\,dr,
    \qquad a,b\in\{1,\dots,R_{\max}\},
    \label{eq:W_overlap}
\end{align}
with $S_{aa}=1$ by construction. If the centers $\{r_q\}$ are well-separated compared to $\sigma_{\mathrm{sh}}$, then $S_{ab}\approx 0$ for $a\neq b$. The corresponding \emph{shell coefficients} are defined as
\begin{align}
    \widehat{F}_{\ell m}[q] \triangleq \int_0^1 F_{\ell m}(r)\,W_q(r)\,r^2\,dr, \qquad q=1,\dots,R_{\max} .
    \label{eq:shell_coeff_scalar}
\end{align}

\subsubsection{Spherical-harmonic form of the moments}
\label{sec:so3_invariants}

Before describing the recovery procedure from the extracted moments, it is convenient to express the second-order and third-order $\mathsf{SO}(3)$ moments in the spherical-harmonic basis $\{Y_{\ell m}\}$.  
In this representation, the angular dependence is captured entirely by the \emph{Gaunt coefficients} (equivalently, the Wigner-$3j$ symbols), which encode the coupling between angular modes, while the radial dependence enters only through the overlaps of the window functions $\{W_c\}$~\eqref{eq:W_def}.  
%This separation of angular and radial components provides a compact and structured formulation of $M_f^{(2)}$ and $M_f^{(3)}$, which will be used in the subsequent analysis of the reconstruction algorithm.

For degrees $\ell_i$ and orders $|m_i| \leq \ell_i$ ($i=1,2,3$), define the Gaunt coefficients by
\begin{align}
  G^{\ell_1\ell_2\ell_3}_{m_1m_2m_3}
   = \int_{S^2}
    Y_{\ell_1 m_1}(\omega)\,
    Y_{\ell_2 m_2}(\omega)\,
    Y_{\ell_3 m_3}(\omega)\,
  d\omega,
  \label{eq:gaunt_def}
\end{align}
which can be written in closed form using Wigner-$3j$ symbols as
\begin{align}
  G^{\ell_1\ell_2\ell_3}_{m_1m_2m_3} =  \sqrt{\tfrac{(2\ell_1+1)(2\ell_2+1)(2\ell_3+1)}{4\pi}}\,
  \begin{pmatrix}\ell_1&\ell_2&\ell_3\\[2pt]0&0&0\end{pmatrix}
  \begin{pmatrix}\ell_1&\ell_2&\ell_3\\[2pt]m_1&m_2&m_3\end{pmatrix}.
  \label{eq:gaunt_w3j}
\end{align}
Here, the $2\times3$ arrays denote Wigner-$3j$ symbols 
\cite{varshalovich1988quantum,edmonds1996angular}.  
These coefficients encode the Clebsch-Gordan coupling rules in a symmetric, orthonormal form and automatically enforce the standard selection rules: (i) the triangle condition $|\ell_1-\ell_2| \le \ell_3 \le \ell_1+\ell_2$; (ii) the parity constraint $\ell_1+\ell_2+\ell_3$ even; (iii) and the order constraint $m_1+m_2+m_3=0$.

\paragraph{Invariants corresponding to $M_f^{(2)}$ and $M_f^{(3)}$.}
The second-order moments correspond to Gram matrices of the shell coefficients, while the third-order moments (the bispectrum) couple angular modes through the Gaunt coefficients and radial windows, producing scalar invariants and linear systems that constrain the unknown spherical-harmonic coefficient vectors.
To make these relationships explicit, it is convenient to collect all spherical-harmonic coefficients on a given degree $\ell$ into a matrix form.

\begin{definition}[Shell spherical-harmonic coefficient matrix]
\label{def:Fmatrix}
Recall the definition of the shell coefficients $\widehat{F}_{\ell m}[q]$, for $1 \leq q \leq R_{\max}$, in~\eqref{eq:shell_coeff_scalar}.
For each degree $\ell$, define the matrix of shell coefficients as
\begin{align}
    \widehat{\mathbf{F}}_\ell \in \mathbb{C}^{(2\ell+1)\times R_{\max}},
    \qquad [\widehat{\mathbf{F}}_\ell]_{m,q} = \widehat{F}_{\ell m}[q],
    \quad m=-\ell,\dots,\ell.
\end{align}
\end{definition}

Under a rotation $\mathcal{R}\in\mathsf{SO}(3)$, the spherical-harmonic coefficients transform blockwise according to the irreducible representation $D^{(\ell)}(\mathcal{R})$:
\begin{align}
  \widehat{\mathbf{F}}_\ell \mapsto D^{(\ell)} (\mathcal{R})\, \widehat{\mathbf{F}}_\ell,
  \qquad \big[\widehat{\mathbf{F}}_\ell^{(\mathcal{R})}\big]_{m,q} = \sum_{m'=-\ell}^{\ell} D^{(\ell)}_{m m'}(\mathcal{R})\,[\widehat{\mathbf{F}}_\ell]_{m',q},
  \label{eq:left_action}
\end{align}
where $D^{(\ell)}(\mathcal{R})$ is the unitary Wigner $D$-matrix of degree $\ell$.

\begin{definition}[Degree-2 Gram invariants]
\label{def:gram}
For each $\ell\in\{0,1,\dots,L_{\max}\}$, define the Gram matrix
\begin{align}
  \mathbf{G}_{2,\ell} \triangleq \widehat{\mathbf{F}}_\ell^{H} \widehat{\mathbf{F}}_\ell \in \mathbb{C}^{R_{\max}\times R_{\max}}.
  \label{eq:G2_def}
\end{align}
\end{definition}

The matrix $\mathbf{G}_{2,\ell}$ is Hermitian, positive semidefinite, and invariant under the left action~\eqref{eq:left_action}, since $D^{(\ell)}(\mathcal{R})$ is unitary.  
Thus, $\mathbf{G}_{2,\ell}$ constitutes the spherical-harmonic realization of the second-order moment $M_f^{(2)}$, capturing radial correlations within each angular frequency band~$\ell$.

\begin{definition}[Degree-3 bispectrum invariants]
\label{def:bispectrum}
For shells $a,b,c \in \{1,\dots,R_{\max} \}$, define the radial triple-overlap
\begin{align}
  T_{a,b,c} \triangleq \int_0^1 W_a(r)\,W_b(r)\,W_c(r)\,r^2\,dr.
  \label{eq:triple_overlap}
\end{align}
Let $\widehat{F}_{\ell m}[q]$ be the shell coefficients from~\eqref{eq:shell_coeff_scalar} and $G^{\ell_1\ell_2\ell_3}_{m_1m_2m_3}$ the Gaunt coefficients from~\eqref{eq:gaunt_def}--\eqref{eq:gaunt_w3j}.
For each admissible angular triple $(\ell_1,\ell_2,\ell_3)$ and shells $(a,b,c) \in \{1 ,2 \ldots , R_{\max} \}^3$, the degree-3 bispectrum invariant is defined by
\begin{align}
    B_{\ell_1,\ell_2,\ell_3}(a,b,c) \triangleq
    T_{a,b,c}  \sum_{\substack{|m_i|\le \ell_i \\ i=1,2,3}}
    G^{\ell_1\ell_2\ell_3}_{m_1m_2m_3}\,
    \widehat{F}_{\ell_1 m_1}[a]\,
    \widehat{F}_{\ell_2 m_2}[b]\,
    \overline{\widehat{F}_{\ell_3 m_3}[c]}.
    \label{eq:bispectrum_def}
\end{align}
\end{definition}

\subsection{Extraction of SO(3) moments from SE(3) autocorrelations}
\label{subsec:numericalImplementation3D}

Theorem~\ref{thm:reductionFromAutocorrelationToTensorMoment} establishes that the $\mathsf{SO}(3)$ moments can be obtained from its $\mathsf{SE}(3)$ autocorrelations through a boundary-limit relation.  
In this subsection, we describe a practical numerical implementation of this extraction in three dimensions.  
The overall structure closely parallels the two-dimensional case, with appropriate modifications to account for the transition from planar to spherical geometry and the associated three-dimensional integration domains.

For any collection of points   $\bm{\eta}_1,\ldots,\bm{\eta}_d$ in the ball $\mathcal{B}_R \triangleq  \mathcal{B}_R^{(3)}\subset \mathbb{R}^3$, the $d$th-order $\SO(3)$ moment of $f$ is given, according to Theorem~\ref{thm:reductionFromAutocorrelationToTensorMoment}, by
\begin{align}
    M_{f}^{(d)}(\bm{\eta}_1, \dots, \bm{\eta}_d)
    &= \lim_{\delta \to 0^+} 
    \frac{\displaystyle 
      \int_{S^{2}} 
      A^{(d+2)}_{f}\left( \bm{\tau}_0^{(\delta)}(\omega), \bm{\tau}_1^{(\delta)}(\omega), \bm{\eta}_1, \dots, \bm{\eta}_d
      \right)\, d\omega}
      {\displaystyle 
      \int_{S^{2}} 
      A^{(2)}_{f}\left( \bm{\tau}_0^{(\delta)}(\omega), \bm{\tau}_1^{(\delta)}(\omega)
      \right)\, d\omega},
    \label{eqn:mainTheoremExtraction3D}
\end{align}
where the boundary points approaching antipodal locations on the sphere are defined as
\begin{align}
    \bm{\tau}_0^{(\delta)}(\omega) = (R(1 - \delta),\omega), 
    \qquad
    \bm{\tau}_1^{(\delta)}(\omega) = (-R(1 - \delta), \omega),
    \qquad \omega \in S^2.    \label{eqn:boundaryAntipodes3D}
\end{align}
%The integration over $S^2$ extracts the rotationally invariant component of the $(d+2)$-point autocorrelation by averaging over all boundary directions, while the denominator provides the corresponding normalization factor.  
%We focus on the cases $d=2$ and $d=3$, which suffice for the complete reconstruction pipeline.

The numerical procedure follows three steps:
(i) rewrite the $\SE(3)\to\SO(3)$ extraction in terms of spherical shells;
(ii) discretize translations, radii (via Gaussian shells), and sphere angles (via a quadrature grid);
(iii) compute the translation accumulators.

\subsubsection{Continuous setting} \label{subsubsec:continuous_reduction_3D}
Let $f:\mathcal{B}_R\subset\mathbb{R}^3\to\mathbb{R}$ be supported on the ball of radius $R$.
For each translation center $\bm{t}\in\mathbb{R}^3$, define the \emph{sphere signal} at radius $r\in[0,R]$ by
\begin{align}
    a^{(\bm{t})}_r(\omega)  =  f\bigl(\bm{t}+r\,\omega\bigr), 
    \qquad \omega\in S^2 .
    \label{eq:3D_sphere_signal}
\end{align}
Introduce the boundary factor
\begin{align}
    s(\bm{t}) = \int_{S^2} f \bigl(\bm{t}+\bm{\tau}_0(\omega)\bigr)\, f \bigl( \bm{t} + \bm{\tau}_1(\omega) \bigr) \,d \omega,
    \label{eq:3D_boundary_factor}
\end{align}
where $\bm{\tau}_0(\omega), \bm{\tau}_1(\omega)$ are defined in~\eqref{eqn:boundaryAntipodes3D}, and we fix a small $\delta > 0$. This coincides with the denominator of the extraction in~\eqref{eqn:mainTheoremExtraction3D}, namely,
\begin{align}
    D = \int_{S^2} A_f^{(2)}\bigl(\bm{\tau}_0(\omega),\bm{\tau}_1(\omega)\bigr)\,d \omega =  \int_{\mathbb{R}^3} s(\bm{t})\,d\bm{t}.
    \label{eq:3D_denominator}
\end{align}

\paragraph{Second-order moment extraction.}
For $d=2$, the extracted second-order $\SO(3)$ moment at radii $(r_1,r_2)$ can be expressed, after rotational averaging, as a function depending only on the angular separation $\gamma \in [0,\pi]$:
\begin{align}
    \Gamma^{(2)}(r_1,r_2;\gamma) =    \int_{\mathbb{R}^3} s(\bm{t})\, \Biggl[\int_{S^2} a^{(\bm{t})}_{r_1}(\omega)\, a^{(\bm{t})}_{r_2}\bigl(\mathcal{R}_\gamma\omega\bigr)\, d \omega\Biggr]\,d\bm{t},
    \label{eq:3D_gamma2_def}
\end{align}
where $\mathcal{R}_\gamma$ denotes rotation mapping the north pole to a direction at geodesic distance $\gamma$ from it.  
The inner integral depends on $\gamma$ only through this relative angle.

Expanding each spherical signal $a^{(\bm{t})}_r$ in spherical harmonics gives
\begin{align}
    c_{\ell m}^{(\bm{t})}(r) = \int_{S^2} a^{(\bm{t})}_r(\omega)\, \overline{Y_{\ell m}(\omega)}\, d \omega,
    \qquad
    0\le \ell\le L_{\max},\,  |m|\le \ell,
    \label{eq:3D_SHT_coef}
\end{align}
where $Y_{\ell m}$ are orthonormal on $S^2$ with respect to $d\omega$. 
We note that at the origin, the spherical-harmonic coefficients coincide with the intrinsic volume coefficients $F_{\ell m}(r)$, defined in~\eqref{eq:sph_expansion}, that is, $c_{\ell m}^{(\bm{t}=0)}(r) = F_{\ell m}(r)$.

The rotationally averaged correlation theorem on the sphere then implies the Legendre expansion
\begin{align}
  \frac{1}{4\pi}\!\int_{S^2}\! 
    a^{(\bm{t})}_{r_1}(\omega)\,
    a^{(\bm{t})}_{r_2}(\mathcal{R}_\gamma\omega)\, d \omega
   = \sum_{\ell=0}^{L_{\max}} 
    \frac{2\ell+1}{4\pi}\,
    \Biggl(\sum_{m=-\ell}^{\ell} c_{\ell m}^{(\bm{t})}(r_1)\, \overline{c_{\ell m}^{(\bm{t})}(r_2)} \Biggr)\,
    P_\ell(\cos\gamma),
  \label{eq:3D_zonal_corr}
\end{align}
where $P_\ell$ denotes the Legendre polynomial associated with the spherical harmonics via the addition theorem
\begin{align}
    P_\ell(\omega_1 \cdot \omega_2) = \frac{4\pi}{2\ell+1} \sum_{m=-\ell}^{\ell} Y_{\ell m}(\omega_1)\, \overline{Y_{\ell m}(\omega_2)},
    \label{eq:addition_theorem}
\end{align}
so that $P_\ell(\cos\gamma)$ represents the rotation-invariant component of degree~$\ell$. Substituting into \eqref{eq:3D_gamma2_def} yields the Legendre expansion
\begin{align}
    \Gamma^{(2)}(r_1,r_2;\gamma)
    = \sum_{\ell=0}^{L_{\max}}
    (2\ell+1)\,P_\ell(\cos\gamma) 
    \int_{\mathbb{R}^3} s(\bm t)\,
    \Biggl(\sum_{m=-\ell}^{\ell}
    c_{\ell m}^{(\bm{t})}(r_1)\,
    \overline{c_{\ell m}^{(\bm{t})}(r_2)}
    \Biggr)\,d\bm t.
    \label{eq:Gamma2_legendre_exact}
\end{align}
Consequently, the extracted second-moment invariant is
\begin{align}
    M^{(2)}_{\text{ext}}(r_1,r_2;\gamma)
    = \frac{\Gamma^{(2)}(r_1,r_2;\gamma)}{D},
\end{align}
where $D$ is defined in~\eqref{eq:3D_denominator}.

\paragraph{Third-order moment extraction.}
For $d=3$, the extracted third-order $\SO(3)$ moment admits the convenient representation in the spherical-harmonic basis using Gaunt coefficients, defined in~\eqref{eq:gaunt_def}.
Then, define the translation-weighted bispectral accumulator
\begin{align}
    \nonumber \Gamma^{(3)} & (r_1,r_2,r_3;\ell_1,\ell_2,\ell_3)
    \\ &  =  \int_{\mathbb{R}^3} 
    s(\bm{t}) \sum_{\substack{|m_i|\le \ell_i \\ i=1,2,3}} G^{\ell_1\ell_2\ell_3}_{m_1m_2m_3}\, c_{\ell_1 m_1}^{(\bm{t})}(r_1)\, c_{\ell_2 m_2}^{(\bm{t})}(r_2)\, \overline{c_{\ell_3 m_3}^{(\bm{t})}(r_3)} \, d\bm{t}.
  \label{eq:3D_bispec_accum}
\end{align}
which is the numerator the right-hand side in~\eqref{eqn:mainTheoremExtraction3D}.
Thus, combining~\eqref{eq:3D_bispec_accum} and~\eqref{eq:3D_denominator} as the numerator and denominator of~\eqref{eqn:mainTheoremExtraction3D}, respectively, results,
\begin{align}
    M^{(3)}_{\text{ext}}(r_1,r_2,r_3;\ell_1,\ell_2,\ell_3) = \frac{\Gamma^{(3)}(r_1,r_2,r_3;\ell_1,\ell_2,\ell_3)}{D}.
    \label{eq:3D_M3_red}
\end{align}

\subsubsection{Discrete setting}
\label{subsubsec:discrete_reduction_3D}

We now describe the discrete counterparts of the continuous formulas presented in Section~\ref{subsubsec:continuous_reduction_3D}, as used in our numerical implementation. 
In the continuous formulation, expectations are taken over all translations $\bm{t}\in\mathbb{R}^3$, and integrals are carried out over the radial and angular variables $(r,\omega)\in[0,1]\times S^2$. 
In practice, these integrals are approximated using finite grids and quadrature rules.

Let $\Lambda=\{\bm{t}_j\}_{j=1}^{N_t}\subset\mathbb{Z}^3$ be a symmetric voxel grid of translations (after zero-padding of $f$). 
Let $\{(\theta_i,\varphi_i),w_i\}_{i=1}^{N_{\mathrm{ang}}}$ denote a quadrature rule on $S^2$ (e.g., Fibonacci nodes), with
\begin{align}
    \omega_i = \big(\sin\theta_i\cos\varphi_i,\ \sin\theta_i\sin\varphi_i,\ \cos\theta_i\big).
\end{align}
For the radial coordinate, fix a quadrature $\{(r_n,w^{(r)}_n)\}_{n=1}^{N_r}$ on $[0,1]$ with weight $r^2\,dr$, and define a family of $R_{\max}$ Gaussian shell windows $\{W_q\}_{q=1}^{R_{\max}}$ normalized as in \eqref{eq:W_def}.
Finally, for the Legendre expansion in~\eqref{eq:Gamma2_legendre_exact}, we sample the angles $\gamma$ by $\mu = \cos\gamma \in [-1,1]$ using a Gauss-Legendre quadrature $\{(\mu_k,w_k)\}_{k=1}^{K_{\mathrm{GL}}}$. 
These discretization parameters are summarized in Table~\ref{tab:discretization_so3}.

\begin{table}[h]
\centering
\begin{tabular}{ll}
\toprule
    Symbol & Description/ Discretization \\
    \midrule
    $\Lambda=\{\bm{t}_j\}_{j=1}^{N_t}$ & Translation grid in $\mathbb{Z}^3$ covering the unit ball (with padding) \\
    $\{(\theta_i,\varphi_i),w_i\}_{i=1}^{N_{\mathrm{ang}}}$ & Angular quadrature on $S^2$, weights $w_i$ \\
    $\{(r_n,w^{(r)}_n)\}_{n=1}^{N_r}$ & Radial quadrature on $[0,1]$ with weight $r^2\,dr$ \\
    $\{W_q\}_{q=1}^{R_{\max}}$ & Gaussian shell windows (normalized as in \eqref{eq:W_def}) \\
    $\{(\mu_k,w_k)\}_{k=1}^{K_{\mathrm{GL}}}$ & Gauss--Legendre quadrature on $\mu=\cos\gamma\in[-1,1]$ \\
    \bottomrule
    \end{tabular}
    \caption{Discretization parameters used in the $\mathsf{SE}(3)$ to $\mathsf{SO}(3)$ extraction.}
    \label{tab:discretization_so3}
\end{table}

The continuous boundary factor $s(\bm{t})$ in~\eqref{eq:3D_boundary_factor} is replaced by its discrete approximation
\begin{align}
    s_h(\bm{t}_j) = \sum_{i=1}^{N_{\mathrm{ang}}} w_i\, f\bigl(\bm{t}_j + \bm{\tau}_0(\omega_i)\bigr)\,   f\bigl(\bm{t}_j + \bm{\tau}_1(\omega_i)\bigr),
    \label{eq:3D_disc_boundary}
\end{align}
and the global normalization constant~\eqref{eq:3D_denominator} becomes
\begin{align}
    D_h = \sum_{j=1}^{N_t} s_h(\bm{t}_j).
    \label{eq:3D_disc_denominator}
\end{align}
For each translation $\bm{t}_j$ and discrete radius $r_n$, the spherical signal $a^{(\bm{t}_j)}_{r_n}(\omega)$ from~\eqref{eq:3D_sphere_signal} is sampled over the quadrature nodes $\omega_i$, and the spherical-harmonic coefficients~\eqref{eq:3D_SHT_coef} are approximated as
\begin{align}
  c_{\ell m}^{(\bm{t}_j)}(r_n) = \sum_{i=1}^{N_{\mathrm{ang}}} w_i\, f\bigl(\bm{t}_j + r_n \omega_i\bigr)\, \overline{Y_{\ell m}(\omega_i)},
  \qquad 
  0\le \ell\le L_{\max},\ |m|\le \ell.
  \label{eq:3D_disc_SHT}
\end{align}

We define the projection of $c_{\ell m}^{(\bm{t})}(r)$ onto the Gaussian shell basis, similar to \eqref{eq:shell_coeff_scalar}:
\begin{align}
    H^{(\bm{t})}_{\ell m}[q]\
    \triangleq  
    \int_0^1 c_{\ell m}^{(\bm{t})}(r)\,W_q(r)\,r^2\,dr,
    \qquad |m|\le\ell.
    \label{eq:relation_f_to_c}
\end{align}
Then, substituting \eqref{eq:3D_SHT_coef} into~\eqref{eq:relation_f_to_c} results
\begin{align}
    H^{(\bm{t})}_{\ell m}[q]\
   =  \int_0^1 \int_{S^2} 
    a^{(\bm{t})}_r(\omega)\,
    \overline{Y_{\ell m}(\omega)}\,
  \,W_q(r)\,r^2 \,dr \, d \omega, \label{eqn:app-B15} 
\end{align}
for $0 \leq \ell \leq L_{\max}, \, |m| \leq \ell, 1 \leq q \leq R_{\max}$.
The Gaussian-shell coefficients defined in~\eqref{eq:relation_f_to_c} are discretized by applying the radial quadrature rule:
\begin{align}
  H^{(\bm{t}_j)}_{\ell m}[q] = \sum_{n=1}^{N_r} 
    w^{(r)}_n\,W_q(r_n)\, c_{\ell m}^{(\bm{t}_j)}(r_n), \qquad 1\le q\le R_{\max}.  \label{eq:discrete_shell_projection}
\end{align}

\paragraph{Discrete second-order moment extraction.}
For each translation $\bm{t}_j$ and shell pair $(q_1,q_2) \in \{1 ,2, \ldots, R_{\max} \}^2$, define the discrete cross-power spectra
\begin{align}
  G_\ell^{(\bm{t}_j)}[q_1,q_2] = \sum_{m=-\ell}^{\ell} H^{(\bm{t}_j)}_{\ell m}[q_1]\,
  \overline{H^{(\bm{t}_j)}_{\ell m}[q_2]}.
  \label{eq:cross_power_discrete}
\end{align}
Substituting these expansions into~\eqref{eq:Gamma2_legendre_exact} and approximating the integration over $\cos\gamma \in [-1,1]$ by a Gauss-Legendre quadrature yields the discrete Legendre representation
\begin{align}
  \Gamma^{(2)}_h(r_{q_1},r_{q_2};\mu_k)
   = \sum_{j=1}^{N_t} s_h(\bm{t}_j)
  \sum_{\ell=0}^{L_{\max}} (2\ell+1)\, G_\ell^{(\bm{t}_j)}[q_1,q_2]\,
    P_\ell(\mu_k), \qquad \mu_k \in [-1,1],
  \label{eq:Gamma2_discrete}
\end{align}
where each $G_\ell^{(\bm{t}_j)}[q_1,q_2]$ collects the shellwise spherical-harmonic correlations at translation~$\bm{t}_j$.
Finally, the normalized extracted moment is defined as
\begin{align}
  M^{(2)}_{\text{ext},h}(r_{q_1},r_{q_2};\mu_k)
   = \frac{\Gamma^{(2)}_h(r_{q_1},r_{q_2};\mu_k)}{D_h},
  \label{eq:M2_discrete}
\end{align}
where $D_h = \sum_{j=1}^{N_t}s_h(\bm{t}_j)$ is the discrete normalization factor defined in~\eqref{eq:3D_disc_denominator}.

\paragraph{Discrete third-order moment extraction.}
For the bispectrum, using the precomputed Gaunt coefficients $G^{\ell_1\ell_2\ell_3}_{m_1m_2m_3}$ from~\eqref{eq:gaunt_def}, the discrete counterpart of~\eqref{eq:3D_bispec_accum} is given by
\begin{align}
    \nonumber \Gamma_h^{(3)}&(r_{q_1},r_{q_2},r_{q_3};
    \ell_1,\ell_2,\ell_3)
    \\ &  = 
    \sum_{j=1}^{N_t} s_h(\bm{t}_j)
    \sum_{m_1=-\ell_1}^{\ell_1}
    \sum_{m_2=-\ell_2}^{\ell_2}
    \sum_{m_3=-\ell_3}^{\ell_3} G^{\ell_1\ell_2\ell_3}_{m_1m_2m_3}\,
    H^{(\bm{t}_j)}_{\ell_1 m_1}[q_1]\,
    H^{(\bm{t}_j)}_{\ell_2 m_2}[q_2]\,
    \overline{H^{(\bm{t}_j)}_{\ell_3 m_3}[q_3]},
    \label{eq:bispec_accum_discrete}
\end{align}
and the corresponding normalized third-order moment is
\begin{align}
    M^{(3)}_{\text{ext},h}(r_{q_1},r_{q_2},r_{q_3};\ell_1,\ell_2,\ell_3) = \frac{\Gamma_h^{(3)}(r_{q_1},r_{q_2},r_{q_3};\ell_1,\ell_2,\ell_3)}{D_h}.
    \label{eq:M3_discrete}
\end{align}

\subsection{Orbit recovery from SO(3) moments via frequency marching}
\label{subsec:freq_march}

We are now ready to describe the recovery of the volume from the extracted moments.  
Although the notation below is expressed in terms of the full Gram and bispectrum invariants, the actual inputs to the algorithm are their extracted counterparts, obtained from the extraction of the $\mathsf{SO}(3)$ moments from the  $\mathsf{SE}(3)$ autocorrelations described earlier.

The recovery procedure implemented here follows the constructive \emph{frequency marching} strategy introduced in~\cite{bendory2025orbit}, with all algorithmic details made explicit for reproducibility.  
The goal is to reconstruct the shell coefficient vectors
\begin{align}
    \widehat{\mathbf{f}}_{\ell}(c)  \triangleq   \widehat{\mathbf{F}}_{\ell}(:,c) \in \mathbb{C}^{2\ell+1}, \qquad c \in \{ 1, 2 , \ldots, R_{\max} \} \label{eqn:shell-vector-def}
\end{align}
where $\widehat{\mathbf{f}}_\ell(c)$ denotes the $c$-th column of $\widehat{\mathbf{F}}_\ell$ (Definition~\ref{def:Fmatrix}), from the degree-2 and degree-3 invariants of $f$.
The central idea is that the Gram invariants (Definition~\ref{def:gram}) determine the lowest two angular degrees, while the higher degrees are recovered inductively from the bispectrum (Definition~\ref{def:bispectrum}), progressively building the full set of spherical-harmonic coefficients of the volume.

\paragraph{Linear systems from the bispectrum.}
Each entry in the bispectrum~\eqref{eq:bispectrum_def} with $\ell_3=\ell$ yields a linear equation in the unknown vector ${\widehat{\mathbf{f}}_\ell(c)}$~\eqref{eqn:shell-vector-def}.  
Fix a target pair $(\ell,c)$.  
Let $\Pi_{\ell,c}$ denote the set of admissible tuples $(\ell_1,\ell_2,a,b)$ with $\ell_1,\ell_2<\ell$ and $a,b\in\{1,\dots,R_{\max} \}$ for which $B_{\ell_1,\ell_2,\ell}(a,b,c)$ is available.  
Its cardinality $N_{\ell,c}=|\Pi_{\ell,c}|$ is the number of bispectrum equations available for $(\ell,c)$.  
By stacking these equations we obtain the linear system
\begin{align}
    \mathbf{A}_{\ell,c}\,\widehat{\mathbf{f}}_\ell(c) = \mathbf{b}_{\ell,c},
    \qquad  \mathbf{A}_{\ell,c}\in\mathbb{C}^{N_{\ell,c}\times (2\ell+1)},  
    \quad  \mathbf{b}_{\ell,c}\in\mathbb{C}^{N_{\ell,c}}.
  \label{eq:bisys_ls}
\end{align}
More explicitly, for each $(\ell_1,\ell_2,a,b)\in\Pi_{\ell,c}$, we can write
\begin{align}
    B_{\ell_1,\ell_2,\ell}(a,b,c) = \mathbf{v}_{\ell_1,\ell_2,a,b,c}^{H}\,\widehat{\mathbf{f}}_\ell(c),
\end{align}
where the vector $\mathbf{v}_{\ell_1,\ell_2,a,b,c}\in\mathbb{C}^{2\ell+1}$ is computed from lower-degree coefficients via
\begin{align}
    \big[\mathbf{v}_{\ell_1,\ell_2,a,b,c}\big]_{m_3}
    = T_{a,b,c} \sum_{m_1,m_2} G^{\ell_1\ell_2\ell}_{m_1m_2m_3}\,
    \widehat{F}_{\ell_1 m_1}[a]\,
    \widehat{F}_{\ell_2 m_2}[b].
\end{align}
The corresponding row and entry in~\eqref{eq:bisys_ls} are then
\begin{align}
    \mathbf{A}_{\ell,c}[i,:]^{H}
    = \mathbf{v}_{\ell_1,\ell_2,a,b,c}^{H},
    \qquad \mathbf{b}_{\ell,c}[i]
    = B_{\ell_1,\ell_2,\ell}(a,b,c),
    \label{eq:bisys_row}
\end{align}
so that $\mathbf{A}_{\ell,c}$ depends only on known Gaunt coefficients, precomputed overlaps, and the previously recovered lower-degree shell matrices $\{\widehat{\mathbf{F}}_{\ell'}\}_{\ell'<\ell}$.

\paragraph{Frequency marching procedure.}
Recovery proceeds degree by degree:
\begin{enumerate}
  \item \emph{Base case $\ell=0$.}  
  The Gram matrix $\mathbf{G}_{2,0}$ is rank-1 and determines $\widehat{\mathbf{F}}_0\in\mathbb{C}^{1\times R_{\max}}$ up to a global phase, obtained from its principal eigenvector/eigenvalue.

  \item \emph{Gauge fixing at $\ell=1$.}  
  The Gram invariant $\mathbf{G}_{2,1}$ specifies only the product $\widehat{\mathbf{F}}_1^{H}\widehat{\mathbf{F}}_1$, leaving a rotational $\mathsf{SO}(3)$ ambiguity. A gauge must therefore be fixed, either by supplying $\widehat{\mathbf{F}}_1$ externally (oracle gauge) or by choosing an arbitrary representative consistent with the bispectrum equations, which is a valid step due to the inherent global $\mathsf{SO}(3)$ ambiguity.

  \item \emph{Inductive step $\ell\ge 2$.}  
  For each shell $c$, assemble~\eqref{eq:bisys_ls} from $\Pi_{\ell,c}$ using the lower-degree solutions, then recover $\widehat{\mathbf{f}}_\ell(c)$ by solving the regularized least-squares problem
  \begin{align}
    \widehat{\mathbf{f}}_\ell(c)
     = \argmin_{x\in\mathbb{C}^{2\ell+1}}
    \|\mathbf{A}_{\ell,c}x -\mathbf{b}_{\ell,c}\|_2^2 + \lambda \|x\|_2^2,
    \label{eq:regLS}
  \end{align}
  with $\lambda\ge 0$ an optional small ridge parameter to improve stability.  
\end{enumerate}

In summary, the Gram invariants fully determine the degree-$0$ coefficients and fix the norms at degree $1$, while the bispectrum invariants inductively determine all higher degrees.  
This constructive scheme is called frequency marching, as the reconstruction proceeds sequentially in $\ell$, each step building on the lower-degree solutions. In principle, one may extend the approach by reiterating this pass to further refine the estimates.

\end{appendices}
\end{document}